\documentclass[11pt]{article}
\usepackage[a4paper]{geometry}
\usepackage{subcaption}

\usepackage{amssymb}
\usepackage{fix-cm}
\usepackage{authblk}
\usepackage{booktabs} 
\usepackage{indentfirst}
\usepackage{times}
\usepackage{graphicx}
\graphicspath{{Figs/}}
\usepackage{comment}
\usepackage{caption}
\usepackage[utf8]{inputenc} 
\usepackage[T1]{fontenc}  
\usepackage{url}   
\usepackage{booktabs}       
\usepackage{nicefrac}     
\usepackage{microtype}     
\usepackage{makecell}
\usepackage{xcolor}
\usepackage[framemethod=tikz]{mdframed}
\usepackage{mathtools}
\usepackage{amsmath,amsfonts,amssymb,amscd,xspace}
\usepackage{mathrsfs}
\usepackage{amsthm}
\usepackage{hyperref}   
\usepackage{cleveref}
\usepackage{extarrows}
\usepackage{appendix}
\usepackage{tocbasic}

\usepackage[shortlabels]{enumitem}
\usepackage{bbm}
\usepackage{bm}
\usepackage{booktabs}
\usepackage{authblk}

\AtBeginDocument{%
  }

\usepackage{algorithm,algpseudocode}
\usepackage{tocloft} 

\usepackage{xcolor}

\usepackage{amssymb}
\usepackage{algorithmicx}
\usepackage{wrapfig}
\usepackage{graphicx}
\usepackage{multirow}
\graphicspath{{Figs/}}
\usepackage{comment}
\usepackage{caption}
\usepackage[utf8]{inputenc} 
\usepackage[T1]{fontenc}  
\usepackage{url}   
\usepackage{booktabs}   
\usepackage{nicefrac}   
\usepackage{microtype}     
\usepackage[ruled,vlined,onelanguage,algo2e]{algorithm2e}
\usepackage{makecell}
\usepackage[framemethod=tikz]{mdframed}
\usepackage{mathtools}
\usepackage{dsfont}
\usepackage{amsmath,amsfonts,amssymb,amscd,xspace}
\usepackage{subcaption}
\usepackage{amsthm}
\usepackage[numbers,sort&compress]{natbib}
\hypersetup{colorlinks={true},linkcolor={blue},citecolor=green}
\usepackage{dsfont}
\usepackage{extarrows}

\usepackage{mathrsfs}

\usepackage{bbm}
\usepackage{bm}
\usepackage{booktabs}
\usepackage{enumitem,kantlipsum}
\hypersetup{
    colorlinks=true,
    linkcolor=blue,
    citecolor=teal,
    filecolor=teal,
    urlcolor=cyan,
}
\usepackage{empheq}

\newtheorem{theorem}{Theorem}
\numberwithin{theorem}{section}
\newtheorem{lemma}{Lemma}

\numberwithin{corollary}{section}

\newtheorem{definition}{Definition}
\newtheorem{proposition}{Proposition}

\newtheorem{remark}{Remark}

\newcommand{\mI}{\mathcal{I}}
\newcommand{\mX}{\mathcal{F}}
\newcommand{\mY}{\mathcal{F}}
\newcommand{\alg}{\mathrm{ALG}}
\newcommand{\opt}{\mathrm{OPT}}
\newcommand{\con}{\beta}
\newcommand{\rob}{\gamma}
\newcommand{\lam}{\lambda}
\newcommand{\weak}{weakly\xspace}
\newcommand{\strong}{strongly\xspace}
\newcommand{\A}{A_{\omega}}
\newcommand{\weakly}{\text{weakly}-optimal\xspace}
\newcommand{\strongly}{\text{strongly}-optimal\xspace}
\newcommand{\B}{B}
\newcommand{\x}{x}
\newcommand{\y}{y}
\newcommand{\I}{I}
\newcommand{\CR}{\alpha}
\newcommand{\KD}{\textsc{KD}\xspace}
\newcommand{\KR}{\textsc{KR}\xspace}

\newcommand{\RSR}{\textsc{PRSR}\xspace}
\newcommand{\ratio}{\mathcal{R}}
\newcommand{\prob}{\pi}
\newcommand{\eq}{\pi^{eq}}
\newcommand{\opA}{\textsc{Operation A}\xspace}
\newcommand{\opB}{\textsc{Operation B}\xspace}

\newcommand{\str}{\textsc{Strong}\xspace}
\newcommand{\Del}{\Delta}
\newcommand{\eps}{\epsilon}
\newcommand{\pri}{p^*}
\newcommand{\du}{d^*}
\newcommand{\OMS}{\textsc{PST}\xspace}
\newcommand{\mS}{\mathcal{S}}
\newcommand{\mL}{\mathcal{L}}
\newcommand{\ep}{\epsilon}

\newcommand{\mP}{\mathscr{P}}
\newcommand{\mA}{\mathcal{A}}
\newcommand{\mU}{\mathcal{U}}
\newcommand{\tP}{\textsc{P}}
\newcommand{\tQ}{\textsc{Q}}

\newcommand{\PDSR}{\textsc{PDSR}}
\newcommand{\PRSR}{\textsc{PRSR}}

\newcommand{\PST}{\textsc{PST}}

\newcounter{mycounter}

\newcommand{\myproblem}[2]{%
  \refstepcounter{mycounter}
  \textbf{Problem \themycounter #1:} #2
}

\title{Prediction-Specific Design of Learning-Augmented Algorithms}

\author[1]{Sizhe Li\thanks{\textit{Email:}~\href{mailto:lisizhe@link.cuhk.edu.cn}{lisizhe@link.cuhk.edu.cn}}}
\author[2]{Nicolas Christianson\thanks{\textit{Email:}~\href{mailto:christianson@stanford.edu}{christianson@stanford.edu}}}
\author[1]{Tongxin Li\thanks{\textit{Email:}~\href{mailto:litongxin@cuhk.edu.cn}{litongxin@cuhk.edu.cn}}}

\affil[1]{School of Data Science, The Chinese University of Hong Kong, Shenzhen}
\affil[2]{Department of Management Science and Engineering, Stanford University}

\date{\vspace*{-1cm}}
\begin{document}
\maketitle

\linespread{1.1}

\begin{abstract}
\emph{Algorithms with predictions} has emerged as a powerful framework to combine the robustness of traditional online algorithms with the data-driven performance benefits of machine-learned (ML) predictions. However, most existing approaches in this paradigm are overly conservative, {as they do not leverage problem structure to optimize performance in a prediction-specific manner}. In this paper, we show that such prediction-specific performance criteria can enable significant performance improvements over the coarser notions of consistency and robustness considered in prior work. Specifically, we propose a notion of \emph{strongly-optimal} algorithms with predictions, which obtain Pareto optimality not just in the worst-case tradeoff between robustness and consistency, but also in the prediction-specific tradeoff between these metrics. We develop a general bi-level optimization framework that enables systematically designing strongly-optimal algorithms in a wide variety of problem settings, and we propose explicit strongly-optimal algorithms for several classic online problems: deterministic and randomized ski rental, and one-max search. Our analysis reveals new structural insights into how predictions can be optimally integrated into online algorithms by leveraging a prediction-specific design. To validate the benefits of our proposed framework, we empirically evaluate our algorithms in case studies on problems including dynamic power management and volatility-based index trading. Our results demonstrate that prediction-specific, strongly-optimal algorithms can significantly improve performance across a variety of online decision-making settings.

\end{abstract}

\section{Introduction}

Online algorithms operate in environments where decisions must be made sequentially without full knowledge of future inputs. Traditionally, these algorithms are designed to guarantee robust performance on adversarial problem instances, providing competitive ratio bounds that hold under worst-case inputs~\cite{Borodin2005}. While theoretically robust, this adversarial perspective often yields overly pessimistic strategies that can underperform in the real world, where worst-case instances are uncommon. Recent work on \emph{algorithms with predictions} addresses this limitation by integrating machine-learned (ML) predictions into classical online decision-making frameworks~\cite{Kumar2018,Lykouris2021}. This learning-augmented algorithm design paradigm has been applied to a variety of online problems including ski rental, caching, and metrical task systems \cite{Kumar2018,Lykouris2021,Antoniadis2020} and applications including GPU power management~\cite{antoniadis2021dpm}, battery energy storage system control~\cite{li2024lado}, carbon-aware workload management~\cite{lechowicz2023opr,lechowiczLearningAugmentedCompetitiveAlgorithms2025a}, and electric vehicle (EV) charging~\cite{lechowicz2024ocs,sun2021omkp}---with algorithms that achieve near-optimal performance when predictions are accurate, while maintaining robust worst-case guarantees when they are not. 

Existing approaches to learning-augmented algorithm design typically evaluate performance through a so-called \emph{consistency-robustness tradeoff}. In particular, consistency measures worst-case algorithm performance when the prediction is perfectly accurate, while robustness measures the worst-case competitive ratio over all possible predictions and instances. 

Importantly, both metrics are inherently worst-case in nature: consistency reflects the algorithm’s performance under the least favorable instance with the least favorable, yet accurate prediction, and robustness measures performance under the least favorable (inaccurate) prediction and instance. Because they are worst-case, neither of these metrics measures whether algorithms can achieve better performance for specific predictions. 

As such, while much prior work has sought to design algorithms that obtain the optimal tradeoff between consistency and robustness (e.g., \cite{Kumar2018,Sun2021}), and some existing algorithms have sought to improve performance by optimizing decisions in a prediction-specific manner (e.g., \cite{lechowiczChasingConvexFunctions2024,lechowiczLearningAugmentedCompetitiveAlgorithms2025a}), no prior work has considered the question of how to design online algorithms with \textit{prediction-specific} guarantees on optimality. 

Motivated by this gap and the potential to improve the performance of algorithms with predictions, our work explores the following question:

\begin{center}
\textit{How can we design algorithms that achieve Pareto-optimal tradeoffs between consistency and robustness that are tailored to specific prediction values?}\end{center}

To this end, we introduce a prediction-specific framework for algorithm design, enabling the development of explicit and tractable online algorithms that adapt to each prediction's characteristics to ensure Pareto-optimal robustness and consistency for each individual prediction value. Instantiating this framework, we design explicit prediction-specific algorithms for the problems of \emph{deterministic and randomized ski rental} and \emph{one-max search}, two classic online problems with connections to real-world applications including TCP acknowledgment~\cite{karlin2001tcpack}, cloud cost management~\cite{ai2014multishop}, dynamic power management~\cite{antoniadis2021dpm}, and energy market operations~\cite{lee2021peakaware, lee2024energymarkets}.

\defcitealias{Kumar2018}{Kumar et al.}
\defcitealias{Sun2021}{Sun et al. }
\begin{table*}[t]
\centering
\caption{Comparison of our technical results on algorithm (weak and strong) optimality.}
\label{tab:optimality_results} 
\small 
\renewcommand{\arraystretch}{1.2} 
\begin{tabular}{l|l|l|l|l|l}
\specialrule{.09em}{.1em}{.1em} 
& \textbf{General} & \textbf{DSR} (\textsc{large} $b$) & \textbf{RSR} (\textsc{large} $b$)  & \textbf{OMS} & $\epsilon$-\textbf{OMS} \\
\hline
\textsc{Weak} & Algorithm~\ref{alg:meta} &  \KD \citetalias{Kumar2018} & \KR \citetalias{Kumar2018} & \citetalias{Sun2021} & $\qquad \backslash$ \\
\hline
\multirow{2}{*}{\textsc{Strong}}  & Algorithm~\ref{alg:meta}  &  \hyperref[alg:dsr]{\PDSR} & \hyperref[alg:rsr]{\PRSR} & 
\hyperref[alg:one-max search]{\PST}& 
\hyperref[alg:ep-one-max search]{$\ep$-Tolerant \PST}\\
\cline{2-6}
 & {\Cref{prop:oba}} & \Cref{thm:dsr} & {\Cref{thm:rsr}} & {\Cref{thm:oms}} & {\Cref{thm:tol_3}}\\
\hline 
\multicolumn{6}{l}{\textbf{Note:}\textit{ Algorithm~\ref{alg:meta}, \PDSR, \PRSR, \PST, and~$\ep$-Tolerant $\PST$ are results from this work. }} \\ 
\multicolumn{6}{l}{\textbf{DSR}/\textbf{RSR}\textit{=Deterministic/Randomized Ski Rental}, \textbf{OMS}\textit{=One-Max Search}, $\epsilon$-\textbf{OMS}\textit{=Error Tolerant} \textbf{OMS}} \\ 
\specialrule{.09em}{.1em}{.1em} 
\end{tabular}
\end{table*}

\begin{figure}[t]
  \centering
  \begin{minipage}[t]{0.32\textwidth}
    \centering
\includegraphics[width=\linewidth]{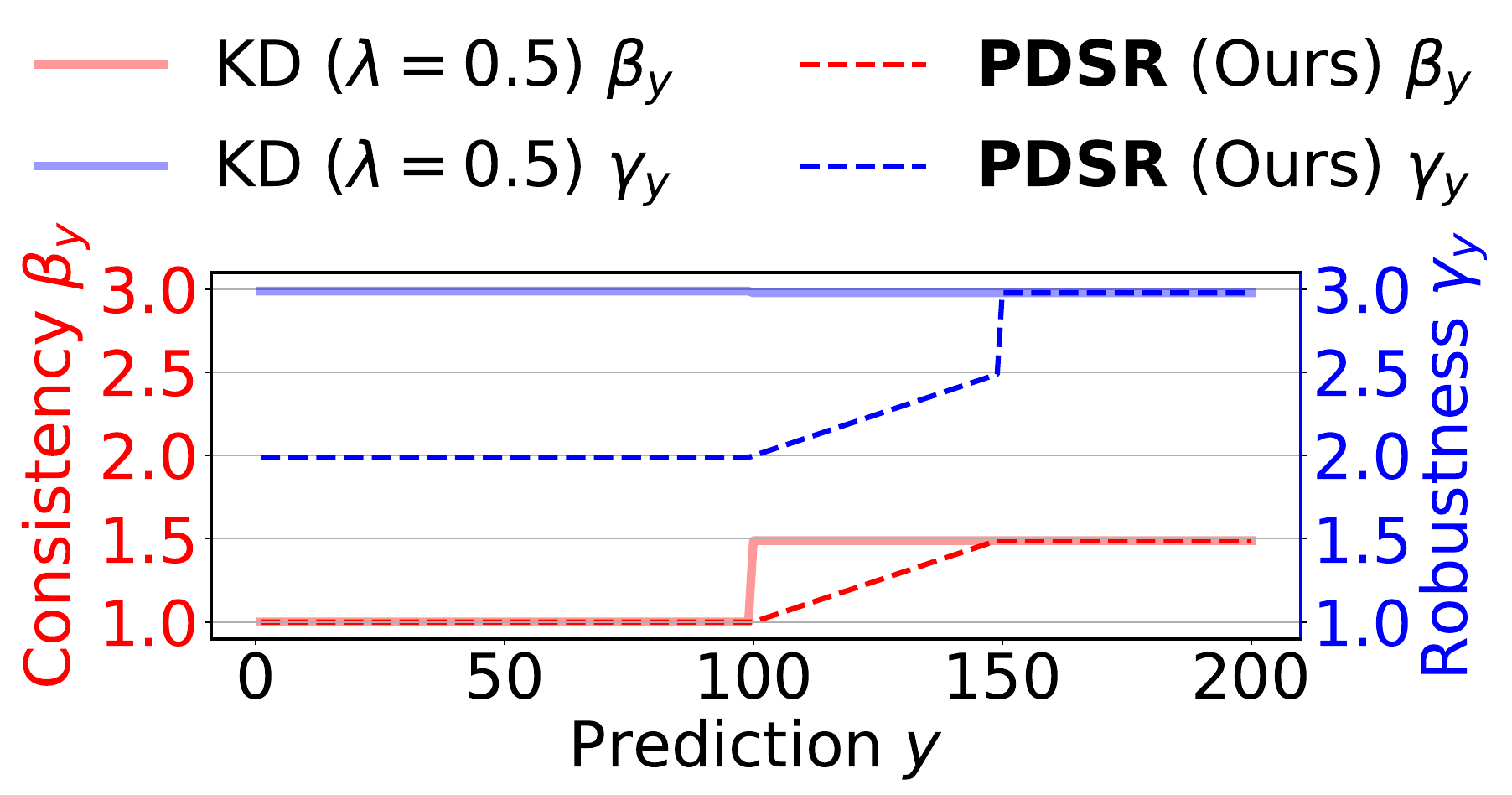}
    \label{fig:dsr}
  \end{minipage}
  \hfill
  \begin{minipage}[t]{0.32\textwidth}
    \centering
    \includegraphics[width=\linewidth]{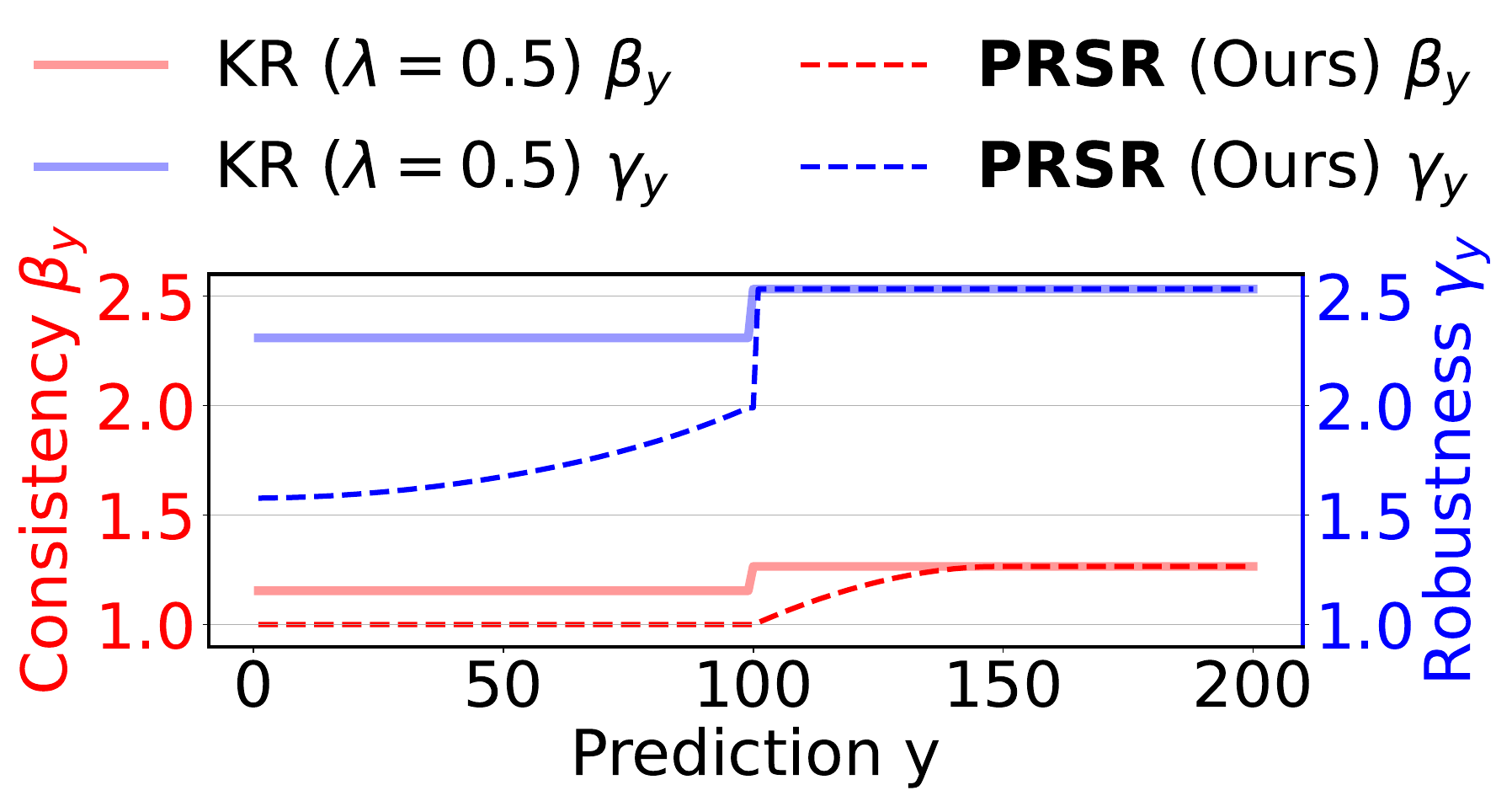}
    \label{fig:rsr}
  \end{minipage}
  \hfill
  \begin{minipage}[t]{0.33\textwidth}
    \centering
    \includegraphics[width=\linewidth]{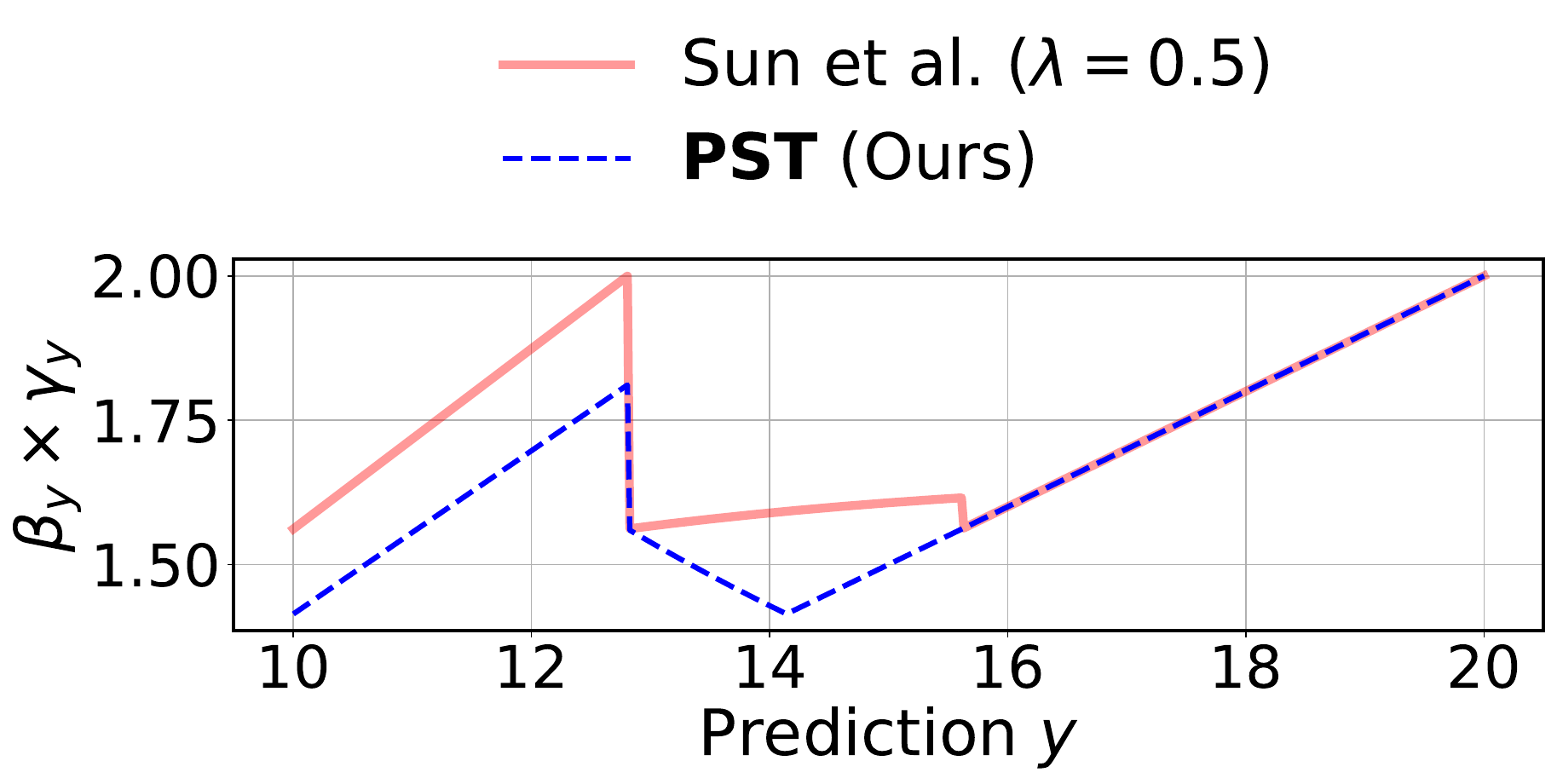}
    \label{fig:oms}
  \end{minipage}
    
  \caption{\textbf{Prediction-specific} consistency $\con_y$ and robustness $\rob_y$ under different predictions $y$ for DSR with $b=100$ (\textsc{Left}), RSR with $b=100$ (\textsc{Middle}), and OMS with $L=10, U=20$ (\textsc{Right}).}
  \vspace{-10pt}
  \label{fig:comparison}
\end{figure}

\subsection{Contributions} 

The main contributions of this paper are as follows. 

\textbf{Framework and Theoretical Results.} We introduce a \emph{novel} prediction-specific framework for the design and analysis of learning-augmented algorithms. Specifically, \Cref{def:ps_con_rob} extends the classic notions of consistency and robustness by defining the prediction-specific consistency $\con_y$ and robustness $\rob_y$ for each possible prediction value $y$. Here, $\con_y$ evaluates the algorithm’s performance on the worst-case instance which is consistent with the prediction $y$, while $\rob_y$ measures its worst-case performance under prediction $y$. Furthermore, \Cref{def_strong} introduces the notion of \emph{strong optimality}: an algorithm is strongly optimal if it is Pareto-optimal not only in the classic sense (referred to as \emph{weak optimality}, see \Cref{def_weak}), but also in the $(\con_y, \rob_y)$ plane for \emph{every} prediction value $y$ in the prediction space $\mY$. This reframes the evaluation of algorithm performance from a single worst-case trade-off 
to a richer, per-prediction perspective, thereby enabling a more fine-grained analysis of algorithm behavior across different prediction values.

Under this new framework, we show that the existing weakly-optimal algorithms for several canonical online problems---\emph{deterministic ski rental}, \emph{randomized ski rental}, and \emph{one-max search} (see Table~\ref{tab:online_problems})---are not strongly-optimal (see Theorems~\ref{thm:kd}, \ref{thm:KR}, and \ref{thm:sun}). As such, we propose new  algorithms for these problems (Algorithms~\ref{alg:dsr}, \ref{alg:rsr}, and \ref{alg:one-max search}) and prove their strong optimality (see Theorems~\ref{thm:dsr}, \ref{thm:rsr}, and \ref{thm:oms}). Notably, our strongly-optimal algorithms can obtain significant performance improvements over prior weakly-optimal algorithms for certain prediction values; Table~\ref{tab:con_rob_ratio} summarizes the best improvement ratio of the consistency–robustness score ($\con_y \cdot \rob_y$) across all $y$ for each problem, reflecting the relative benefit of our algorithms over existing ones under the most favorable $y$. The parameter $\lambda$ represents some \textit{confidence} governing the tradeoff between robustness and consistency used in~\cite{Kumar2018, Sun2021}. 
Note that for both deterministic and randomized ski rental, the maximum improvement ratio is obtained as $\lambda \downarrow 0$, where it diverges to infinity with asymptotic order $\mathcal{O}(\lambda^{-1})$. In contrast, for one-max search, the maximum improvement ratio is $\sqrt{\theta}$, which is also attained as $\lambda \downarrow 0$. Consequently, our methods can yield substantial performance improvements for certain predictions and choices of the parameter $\lambda$.

We also show that there exist problems for which existing weakly-optimal algorithms in the literature are already strongly-optimal: in particular, we establish prediction-specific bounds in \Cref{thm:scheduling} for the \weakly algorithm of Wei and Zhang~\cite{Wei2020} for non-clairvoyant job scheduling when the job size $n=2$, and prove its strong optimality in this case.

\begin{table*}[t!] 
\centering
\small 
\renewcommand{\arraystretch}{1.2} 
\begin{tabular}{c|c|c|c}
\specialrule{.09em}{.1em}{.1em} 
& \textbf{DSR} & \textbf{RSR}  & \textbf{OMS} \\
\hline
\textsf{Existing Results} & $\KD$~(\cite{Kumar2018}, Algorithm 2) & $\KR$~(\cite{Kumar2018}, Algorithm 3) & Sun et al.~\cite{Sun2021} \\ 
\hline
\textsf{This Work}
& $\PDSR$ (Algorithm~\ref{alg:dsr}) & $\PRSR$ (Algorithm~\ref{alg:rsr}) & $\PST$ (Algorithm~\ref{alg:one-max search})\\ 
\hline
\textsf{Improvement Ratio}
& $ \left(1 +\frac{1}{\lambda}\right)/2  $ & $(\frac{e-1}{e})  \cdot(1 - e^{-\lambda})^{-1}   $ & $\frac{\sqrt{(1-\lambda)^2 + 4 \lambda \theta} - (1 - \lambda) }{ 2\lambda \sqrt{\theta}}  $\\ 
\specialrule{.09em}{.1em}{.1em} 
\end{tabular}
\caption{
Comparison of proposed algorithms in this work with existing results. The \textsf{Improvement Ratio} quantifies the maximum ratio between the \textbf{consistency-robustness score} $(\con_y \cdot \rob_y)$ of our algorithm over the best existing score for all $y$.}
\label{tab:con_rob_ratio} 
\end{table*}

\textbf{Novel Techniques and Methodology.} We propose a bi-level optimization methodology (see Problems \ref{optimization_A} and \ref{optimization_B} in~\Cref{sec:bi-level_opt_method}) for systematically deriving strongly-optimal algorithms. Given an initial robustness target $\overline{\rob}$, Problem~\ref{optimization_A} finds the best prediction-specific consistency $\con_\y$, and Problem~\ref{optimization_B} then finds a decision with the best prediction-specific robustness $\rob_\y$ given the fixed $\con_\y$. The bi-level optimization pipeline naturally forms a meta-algorithm (Algorithm~\ref{alg:meta}), which we prove in \Cref{prop:oba} yields a strategy on the prediction-specific Pareto front for each $\y$, guaranteeing strong optimality. This approach offers two key benefits. (1) \textbf{Generality} – the bi-level optimization framework is broadly applicable: it can be used to design algorithms for online problems beyond those described in Table~\ref{tab:online_problems}. (2) \textbf{Flexibility} – this approach exhibits two kinds of flexibility. First, the consistency–robustness trade-off in the meta-algorithm (Algorithm~\ref{alg:meta}) can be tuned by adjusting the robustness target $\overline{\rob}$. Second, the bi-level optimization problem can be adapted to different kinds of predictions and performance objectives. For example, the objectives and constraints in Problems~\ref{optimization_A} ($\mathcal{P}_1$) and~\ref{optimization_B} ($\mathcal{P}_2$) can be extended to incorporate error tolerance to enable the design of algorithms which perform well even if predictions are erroneous, thus alleviating, to some extent, the issues of \emph{brittleness} and \emph{smooth performance degradation} which have been considered in a number of recent works~\cite{elenter2024overcoming,Benomar2025,benomar2025tradeoffs}. More specifically, in~\Cref{sec:brittle}, we introduce the notion of \emph{$\epsilon$-consistency} to formalize the goal of preserving consistency (and avoiding algorithm brittleness) when faced with small prediction error $\eps$, and we show that it is possible to obtain both classic Pareto optimality between $\epsilon$-consistency and robustness, as well as a corresponding version of strong optimality in this case (see \Cref{thm:tol_3}).

In addition, our analysis of strong optimality in the randomized ski rental problem (see Appendix~\ref{appendix:RSR}) employs a novel prediction-specific primal–dual technique. Specifically, we first use a perturbation-based analysis to derive structural properties of the optimal randomized distribution, and then formulate and analyze a primal–dual optimization problem for each prediction. This provides a structured and potentially generalizable approach to establish Pareto optimality with respect to a given prediction.

\textbf{Experimental Results and Real-World Implications.}
We evaluate our algorithms through both synthetic simulations and real-world case studies, spanning a range of online decision-making problems. Overall, our methods consistently outperform state-of-the-art learning-augmented and classical baseline algorithms, confirming their theoretical soundness and practical value.

More specifically, we apply the deterministic and randomized ski rental algorithms to Dynamic Power Management (DPM) traces~\cite{antoniadis2021dpm}, an important benchmark for energy-efficient computing, and we apply the one-max search algorithms to VIX trading, a representative financial market task marked by high volatility and uncertainty. In both domains, our methods deliver significant improvements across most prediction regimes, illustrating how prediction-specific design can translate improved theoretical guarantees into tangible real-world gains.

\subsection{Related Work}
\label{app:related_works}
\textbf{Algorithms with Predictions.} 
Although the study of online algorithms with ML predictions is still relatively new~\cite{Kumar2018,Lykouris2021}, significant research progress has been made in recent years on various problems such as online optimization~\cite{li2024lado}, control~\cite{tongxinli2022control, li2024disentangling}, and reinforcement learning~\cite{golowich2022can,li2024beyond}. Similar frameworks have also been adopted to design and analyze algorithms for a number of other online problems,
such as caching~\cite{Lykouris2021,Rohatgi2019,im2022parsimonious, wei2020better}, online knapsack \cite{Im2021,Daneshvaramoli2024, lechowicz2024time}, secretary and bipartite matching \cite{antoniadis2020secretary}, metrical task systems~\cite{Antoniadis2020,bubeck2021online,Christianson2023}, and convex body/function chasing~\cite{bubeck2021online, christianson22chasing}. Most of these works make no assumptions about prediction quality, seeking to balance \textit{consistency} (competitive ratio when predictions are accurate) and \textit{robustness} (worst-case competitive ratio regardless of prediction quality), though the precise, formal definitions of these metrics vary slightly across works.

Beyond their theoretical appeal, these frameworks have begun to influence practical systems domains, enabling advances in areas such as data-center scheduling~\cite{mao2019learning}, energy-aware computing~\cite{lechowicz2023opr,lee2021peakaware}, and networked control systems~\cite{li2024disentangling,li2024lado}, where ML-driven forecasts are increasingly available but inherently imperfect.

Note also that some recent works depart from the standard paradigm of robustness and consistency and consider alternative prediction models and performance measures. Sun et al. \cite{Sun2024} proposed online algorithms with uncertainty-quantified (UQ) predictions, which leverage UQ to assess prediction quality. They introduced the \textit{distributionally robust competitive ratio} (DRCR), which weighs both the setting where the UQ prediction accurately describes the instance, and the worst-case adversarial setting, and applied this metric to the problems of ski rental and online search. Mahdian et al.~\cite{mahdian2012online} proposed a general framework for online optimization under uncertain inputs, where the algorithm has access to an optimistic strategy that performs well when the future unfolds favorably. They developed a meta-algorithm that balances between this optimistic policy and a robust fallback, achieving a trade-off between worst-case guarantees and performance under accurate predictions without relying on a formal error model.

\textbf{Ski Rental and Scheduling with ML Predictions.} Regarding online problems that are closely related to those specific examples considered in this work, 
Kumar et al.~\cite{Kumar2018} studied \textit{ski rental} and \textit{non-clairvoyant scheduling} with ML predictions. Their framework introduces a tunable trade-off between consistency and robustness through a user-specified hyperparameter. Wei and Zhang~\cite{Wei2020} subsequently gave general lower bounds on the trade-off obtainable in these problems, thus proving that the deterministic and randomized algorithms of Kumar et al.~\cite{Kumar2018} for ski rental achieve the Pareto-optimal trade-off between consistency and robustness. Furthermore, they demonstrated that the meta-algorithm proposed by Kumar et al. \cite{Kumar2018} for non-clairvoyant scheduling does not achieve the tight trade-off, and introduced a novel \textit{two-stage scheduling} strategy that is provably tight for the case of $n=2$. 

\textbf{One-Max Search with ML Predictions.} 
In the learning-augmented setting of one-max search, where algorithms receive a prediction of the maximum element value, Sun et al.~\cite{Sun2021} established a fundamental lower bound on the trade-off between consistency and robustness. They also proposed a threshold-based algorithm and showed that it achieves this lower bound, making it Pareto-optimal in terms of these two performance measures.

\textbf{Algorithm Smoothness.}
A recent line of work has shown that some existing learning-augmented algorithms are brittle, suffering sharp performance drops under small prediction errors. Elenter et al.~\cite{elenter2024overcoming} addressed this by designing algorithms that follow user-specified error-performance profiles, ensuring controlled degradation in performance for the one-way trading problem. Benomar et al.~\cite{Benomar2025} further proposed smooth one-max search algorithms that are Pareto-optimal and exhibit gradual performance degradation with increasing prediction error, achieving a “triple Pareto optimality” among consistency, robustness, and smoothness. Unlike prior work that focuses on structural smoothness of the Pareto frontier, our formulation provides a principled relaxation of consistency itself, leading to algorithms that are both robust and tolerant to small predictive errors.

\section{Problem Formulation \label{problem_formulation}}

We consider online cost minimization problems over a set of instances $\mI$. For each instance $I \in \mI$, let $\alg (A, I)$ and $\opt(I)$ denote the cost incurred by an online algorithm $A$ and the cost of the offline optimum, respectively. We assume that the algorithms have no prior knowledge of the instance $I$. Under the competitive analysis framework \cite{Borodin2005}, the goal is to find an online algorithm $A$ that minimizes the worst-case competitive ratio $\CR (A)$, which is defined as\footnote{In online profit maximization problems, the competitive ratio is defined as the worst-case ratio between the optimal offline profit and that obtained by the online algorithm, which is the inverse of the minimization setting.}
\[
\CR (A) = \max_{I \in \mI} \frac{\alg(A, I)}{\opt(I)}.
\]

This worst-case focus, however, can be overly pessimistic in practical settings. To move beyond the limitations of worst-case algorithms, machine-learned predictions can be integrated into algorithm design (see~\cite{Lykouris2021, Kumar2018}). 
In this setting, an online algorithm $\A$, potentially parameterized by $\omega \in \Omega$, receives not only the instance $\I \in \mI$ online but also an ML prediction $\y \in \mY$ concerning some relevant but unknown feature $\x(\I) \in \mX$ of $\I$, where $\mX$ denotes a prediction space. The feature $\x(\I)$ encapsulates useful information about $\I$ (e.g., the maximum element
value for online one-max search) or may even fully specify $\I$ in some problems (e.g., the total number of skiing
days for discrete-time ski rental). Let $\alg(\A, \I, \y)$ denote the (expected) cost incurred by algorithm $\A$ on instance $\I$ given prediction $\y$.

\textbf{Classic Consistency and Robustness.} Since the ML prediction $\y\in\mY$ may be erroneous, a number of algorithms have been designed (see those summarized in Sections~\ref{app:related_works} and \ref{sec:online_problems}) to (1) achieve near-optimal performance when the prediction $\y$ is accurate (\emph{consistency}), while (2) simultaneously maintaining a bounded performance guarantee even when the prediction is arbitrarily wrong (\emph{robustness}).
To formalize this, let $\mI_\y \subseteq \mI$ represent the set of instances for which prediction $\y$ is considered ``accurate'' or ``consistent''; the precise definition depends on the specific problem, the form of prediction, and the prediction quality measure. The classic consistency and robustness metrics are defined as follows.    
\begin{definition}[Classic Metrics]
\label{def:cl_con_rob}
Given an online algorithm $\A$ that takes predictions, the \textbf{consistency} $\con(\A)$ and \textbf{robustness} $\rob(\A)$ are defined as:
\begin{align}
    \con(\A) \coloneqq \sup_{\y \in \mY} \sup_{ \I \in \mI_\y } \frac{\alg(\A, \I, \y)}{\opt(\I)}, \quad 
    \rob(\A) &\coloneqq \sup_{\y \in \mY} \sup_{\I \in \mI } \frac{\alg(\A, \I, \y)}{\opt(\I)}. \label{eq:std_cons_rob_pred}
\end{align}
If an algorithm $\A$ achieves $\con(\A)$ consistency and $\rob(\A)$ robustness, it is called $\con(\A)$-consistent and $\rob(\A)$-robust, respectively.
\end{definition}
A typical choice of $\mI_\y$ is  $\mI_\y\coloneqq\{I\in\mI:\x(\I)=\y\}$ (following \cite{ Kumar2018,Lykouris2021}). The equality constraint in $\mI_\y$ can be generalized in a number of ways---for instance, to a ball centered at $\x(\I)$, which will emphasize small error tolerance of the algorithm (see~\Cref{sec:brittle} for more details). Here, the \textit{consistency} $\con(\A)$ captures the algorithm's performance guarantee assuming the prediction is correct, taking the worst case over all possible correct predictions, and the \textit{robustness} $\rob(\A)$ represents the worst-case guarantee under all possible instances and any prediction $\y$, regardless of accuracy. An ideal algorithm minimizes both $\con(\A)$ (aiming for near 1) and $\rob(\A)$. 

\subsection{Prediction-Specific Consistency and Robustness}

The standard metrics in~\eqref{eq:std_cons_rob_pred} evaluate performance by taking the worst case over both instances $\I$ and predictions $\mY$. While adopting a worst-case perspective for instances $\I$ is a standard and reasonable approach in competitive analysis due to future uncertainty, applying this same perspective to the entire prediction space $\mY$ is unnecessarily conservative, as better prediction-specific performance may be possible. This motivates a more nuanced framework to evaluate performance conditioned on the \emph{specific} prediction value $\y \in \mY$ that the algorithm receives. We thus introduce the notions of \emph{prediction-specific consistency} and \emph{prediction-specific robustness}.

\begin{definition}[Prediction-Specific Metrics]
\label{def:ps_con_rob}
Given an online algorithm $\A$ and a specific prediction value $\y \in \mY$, the \textbf{prediction-specific consistency} $\con_\y(\A)$ and \textbf{prediction-specific robustness} $\rob_\y(\A)$ relative to $\y$ are defined as:
\begin{align}
    \con_\y(\A) &\coloneqq \sup_{\I \in \mI_\y} \frac{\alg(\A, \I, \y)}{\opt(\I)}, \quad
    \rob_\y(\A) \coloneqq \sup_{\I \in \mI} \frac{\alg(\A, \I, \y)}{\opt(\I)}. \label{eq:ps_cons_rob}
\end{align}
If an algorithm $\A$ achieves $\con_\y(\A)$ and $\rob_\y(\A)$ for a given prediction value $\y$, we say $\A$ is $\con_\y(\A)$-consistent and $\rob_\y(\A)$-robust under $\y$.
\end{definition}

The prediction-specific metrics in Definition~\ref{def:ps_con_rob} enable a finer-grained analysis for algorithms with predictions. Instead of considering only the worst-case consistency or robustness over all predictions, we may instead tailor an algorithm's strategy based on the characteristics of \emph{each} observed prediction value $\y$. This adaptiveness can enable better performance compared to standard algorithms, which only optimize for worst-case consistency and robustness \eqref{eq:std_cons_rob_pred}.

\subsection{\textsc{Weak} and \textsc{Strong} Optimality}

While these prediction-specific metrics provide valuable insight into algorithm performance conditioned on the received prediction, the standard metrics of consistency and robustness still remain valuable to characterize an algorithm's overall worst-case guarantees across all instances and predictions. 

To formally capture the notion of algorithms that perform well both overall and for specific predictions, we introduce two notions of optimality based on these different settings. Our goal is to distinguish algorithms which both achieve the optimal tradeoff between robustness and consistency for the worst-case prediction, as well as on a per-prediction basis. This leads to the following definitions of 
\textit{\weakly} and \textit{\strongly} algorithms.

We begin by defining weak optimality, which characterizes algorithms with the optimal tradeoff between the standard notions of consistency and robustness.

\begin{definition}[\textsc{Weak} Optimality]
\label{def_weak}
For a fixed online problem, consider an online algorithm $A$ that achieves $\con$ consistency and $\rob$ robustness. $A$ is \textbf{\textsc{weakly}-optimal} if there does not exist another online algorithm with consistency $\con'$ and robustness $\rob'$ such that $\con' \leq \con$ and $\rob' \leq \rob$, with at least one inequality being strict.
\end{definition}

Building upon this, we define the stricter notion of strong optimality incorporating prediction-specific performance.
\begin{definition}[\textsc{Strong} Optimality]
\label{def_strong}
For a fixed online problem, consider an online algorithm $A$ that achieves $\con_y$ prediction-specific consistency and $\rob_y$ prediction-specific robustness under prediction $y\in\mY$. $A$ is \textbf{\textsc{strongly}-optimal} if 
\begin{enumerate}
    \item  $A$ is \weak-optimal;
    \item  for any prediction $y \in \mY$, there does not exist another online algorithm with prediction-specific consistency $\con'_y$ and robustness $\rob'_y$ such that $\con'_y \leq \con_y$ and $\rob'_y \leq \rob_y$, with at least one inequality being strict.
\end{enumerate}
\end{definition}

The notion of strong optimality in~\Cref{def_strong} extends classic Pareto optimality to the prediction-specific setting. This definition requires Pareto optimality in the two-dimensional consistency-robustness plane ($\con_y, \rob_y$) \textit{for each specific prediction $y \in \mY$}, in addition to the baseline weak optimality. This stricter criterion captures whether an algorithm is unimprovable in its fundamental performance trade-off across the entire prediction space $\mY$.

\subsection{Online Problems Studied}
\label{sec:online_problems}

We now briefly introduce each of the problems studied in this work; these problems are summarized in~\Cref{tab:online_problems}.

\begin{table*}[ht]
\caption{Instantiation of the general problem components $\I,\x$, and $\y$ for problems analyzed. } 
\label{tab:online_problems} 
\centering
\small 
\renewcommand{\arraystretch}{1.2} 
\begin{tabular}{l|l|l|l}
\specialrule{.09em}{.1em}{.1em} 
& \textbf{Discrete-Time Ski Rental} & \textbf{One-Max Search}  & \textbf{Scheduling} \\
\hline
Instance $\I$ & Number of days & $n$ values in $[L, U]$ & Set of $n$ jobs \\
\hline
Feature $\x$   & Number of days & Maximum value & Processing times $(x_1, \dots, x_n)$ \\ 
\hline
Prediction $\y$  & Predicted number of days & Predicted maximum value & Predicted times $(y_1, \dots, y_n)$ \\ 
\specialrule{.09em}{.1em}{.1em} 
\end{tabular}
\vspace{-1em}
\end{table*}

\textbf{Discrete-Time Ski Rental.} 
The discrete-time ski rental problem is a classic online decision-making problem, wherein a decision-maker decides each day whether to rent skis for $1$ unit of cost per day or purchase them outright for a fixed cost $b \in \mathbb{N}^+$, without prior knowledge of the total number of skiing days $\x$.  The optimal cost is given by $\min \{b, \x\}$.
In the standard competitive analysis framework, where no predictions are available, the optimal deterministic algorithm employs a simple \textit{rent-or-buy} strategy, purchasing the skis on day $M=b$; this strategy achieves a competitive ratio of $2-1/b$. The optimal randomized algorithm is known as Karlin’s algorithm~\cite{Karlin1988}, which strategically balances renting and buying through a carefully-designed probability distribution, achieving a competitive ratio of approximately $e/(e-1)\approx 1.582$.

For the ski rental problem, Kumar et al. \cite{Kumar2018} proposed algorithms that trade off consistency (performance under accurate predictions) and robustness (worst-case performance).  They presented a deterministic algorithm achieving $(1+\lambda)$ consistency and $(1 + 1/\lambda)$ robustness (for $\lambda \in (0,1)$), and a randomized algorithm with $(\frac{\lambda}{1 - e^{-\lambda}})$ consistency and $(\frac{{1 + 1/b}}{1 - e^{-(\lambda - 1/b)}})$ robustness (for $\lambda \in (1/b ,1)$). Subsequently, Wei and Zhang \cite{Wei2020} established fundamental lower bounds on this trade-off, showing that as $b \to \infty$, for deterministic algorithms, any $(1+\lambda)$-consistent algorithm must have a robustness of at least $(1+ 1/\lambda)$ for $\lambda \in (0,1)$, while for randomized algorithms, any $\gamma$-robust algorithm must have a consistency of at least $\gamma\log(1+{1}/{(\gamma-1)})$. These results show that the algorithms in \cite{Kumar2018} achieve weak optimality in the limit $b \to \infty$ (see \Cref{def_weak}).

In this work, we provide a deeper analysis of the algorithms proposed by Kumar et al.~\cite{Kumar2018}, examining their prediction-specific consistency and robustness.
To maintain consistency and comparability with \cite{Kumar2018, Wei2020}, we also consider the asymptotic regime $b \to \infty$. Our findings indicate that neither the deterministic nor randomized algorithms in~\cite{Kumar2018} are \strongly within this framework (see Theorems~\ref{thm:kd} and \ref{thm:KR}). Consequently, we propose novel algorithms (see Algorithms~\ref{alg:dsr} and~\ref{alg:rsr}) that achieve strong optimality (see Theorems~\ref{thm:dsr} and~\ref{thm:rsr}), thereby improving upon existing approaches for this problem setting with ML predictions.

\textbf{One-Max Search.} 
 The one-max search problem considers a sequence of $n$ elements with values in the range $[L, U]$, where $L$ and $U$ are known positive constants. At each step, an element is observed, and the algorithm must decide whether to accept it immediately or irrevocably discard it. The objective is to select the element with the maximum value. The instance's difficulty is characterized by the ratio $\theta = U / L$. In classical competitive analysis, the optimal competitive ratio is $\sqrt{\theta}$, achieved by an algorithm using the fixed threshold $\sqrt{LU}$ \cite{ElYaniv2001}. Note that this is a reward maximization problem, rather than a loss minimization; thus, for this problem, we will consider definitions of classic \eqref{eq:std_cons_rob_pred} and prediction-specific \eqref{eq:ps_cons_rob} consistency and robustness with the ratio between $\alg$ and $\opt$ inverted. 

 For learning-augmented one-max search, where algorithms receive a prediction of the maximum element value, Sun et al.~\cite{Sun2021} established a fundamental lower bound on the consistency-robustness trade-off: any $\gamma$-robust algorithm must have a consistency of at least $\theta / \gamma$. They further proposed a threshold-based algorithm that achieves this lower bound, implying its weak optimality.

In this work, we provide a deeper analysis of the algorithm proposed by Sun et al.~\cite{Sun2021} within this prediction-specific framework. Our analysis reveals that their algorithm is not \strongly (see Theorem~\ref{thm:sun}). Consequently, we propose a novel algorithm (see Algorithm~\ref{alg:one-max search}) that is \strongly (see Theorem~\ref{thm:oms}) and offers improved performance. 

\textbf{Non-Clairvoyant Scheduling.}
The non-clairvoyant scheduling problem concerns scheduling $n$ jobs on a single machine without prior knowledge of their processing times. We focus on the \textit{preemptive} setting, where jobs can be interrupted and resumed without cost. In the standard competitive analysis framework~\cite{Borodin2005}, {Round-Robin (RR)} achieves an optimal competitive ratio of $\frac{2n}{n+1}$~\cite{Motwani1994}. 

Extending this, Wei and Zhang~\cite{Wei2020} established a fundamental lower bound on the consistency-robustness trade-off, showing that any algorithm must have robustness of at least {\small{$\frac{n + n(n-1)\lambda}{1 + \lambda(n+1)(n+2)/2}$}} if it is $(1+\lambda)$-consistent. 
They also propose a \textit{two-stage schedule} algorithm (see Algorithm~\ref{alg:two-stage schedule} in \Cref{appendix:two-stage schedule}) and show it achieves the tight tradeoff for $n=2$ jobs in this learning-augmented setting. 
In the next section, as a warm-up, we provide a deeper analysis of this two-stage algorithm's prediction-specific consistency and robustness, establishing that it is \strongly.

\subsection{Warm-Up: An Existing Algorithm that is Strongly-Optimal \label{sec:scheduling}}

While our proposed notion of strong optimality is much stricter than weak optimality, some existing algorithms known only to be \weakly can be shown to satisfy strong optimality. In this section, we consider the problem of non-clairvoyant scheduling with $n = 2$ jobs (with length predictions $\y=(\y_1, \y_2)$ with $\y_1 \leq \y_2$). It is well known that the two-stage scheduling algorithm proposed by Wei and Zhang~\cite{Wei2020} (see Algorithm~\ref{alg:two-stage schedule} in \Cref{appendix:two-stage schedule}) achieves the optimal trade-off between classic consistency and robustness and is thus weakly-optimal. Notably, we can show that this algorithm also satisfies prediction-specific Pareto optimality:

\begin{theorem}
\label{thm:scheduling}
 The two-stage algorithm in~\cite{Wei2020} with $n=2$ is $1 + \min\{{y_1}/({2y_1+y_2}), \lambda\} $-consistent and $1 + \max\{{1}/{3}, {y_1}/{(y_1 + 2\lambda (2y_1 + y_2)})\} $-robust under $y = (y_1, y_2)$. Moreover, it is \strongly.
\end{theorem} 

We prove this theorem in~\Cref{appendix:two-stage-schedule-optimality}. Despite the fact that this algorithm is strongly optimal for non-clairvoyant scheduling, for many other canonical online problems, including the ski rental and one-max search problems, we shall soon see that existing weakly-optimal algorithms (such as those in~\cite{Kumar2018, Sun2021}) fail to achieve strong optimality (see~\Cref{thm:kd,thm:KR,thm:sun}). As such, the rest of this paper will consider the development of new algorithms that can achieve strong optimality in these settings.

\section{Optimization-Based Meta-Algorithm \label{sec:opt-meta}}

In this section, we introduce a general optimization-based approach to systematically identify \strongly algorithms for online problems.
Our goal is to obtain a general meta-algorithm that, given a prediction $\y\in\mY$ and a target robustness upper bound $\overline{\gamma}$, returns an algorithm $A_\omega$ which is \strongly.

\subsection{Bi-Level Optimization Formulation}
\label{sec:bi-level_opt_method}

We begin by considering the following question: given some target upper bound $\overline{\rob}$ on robustness and a prediction $y \in \mY$, how can we obtain an algorithm that both satisfies the target robustness bound, and obtains a prediction-specific optimal tradeoff between consistency and robustness? Here, $\overline{\rob} \in \Lambda_\rob \coloneqq [\CR^*, +\infty)$ can be any possible robustness level, where $\CR^*$ denotes the optimal competitive ratio achievable for the problem.

To achieve this goal, we specify a bi-level optimization framework comprising Problems~\ref{optimization_A} and \ref{optimization_B}. Recall that a prediction $y$ determines a consistent subset of instances $\mI_y \subseteq \mI$, as described in~\Cref{problem_formulation}, and that $\alg(A_\omega, I, y)$ and $\opt(I)$ denote, respectively, the costs of the algorithm $A_\omega$ augmented by the prediction $y$ and optimal offline algorithm under instance $I \in \mI$. Our first optimization problem (Problem~\ref{optimization_A}), referred to as $\mathcal{P}_1(\overline{\rob}, y)$,
determines the best achievable prediction-specific consistency of any $\overline{\rob}$-robust algorithm under the prediction $y\in\mY$. Let $\{\con_y^*, \overline{\omega}\}$ denote its optimal solution. We then solve a second optimization problem (Problem~\ref{optimization_B}), referred to as $\mathcal{P}_2(\con_y^*, y)$, 
to find the most robust algorithm among those achieving $\con_y^*$ consistency under the prediction $y$. Let $\{\rob_y^*, \omega^*\}$ denote its optimal solution. The resulting algorithm $A_{\omega^*}$ with the optimal solution $\omega^*$ is then implemented for the prediction $y$.


\noindent 
\begin{minipage}[t]{0.48\textwidth} 

\vspace{0pt} 
\myproblem{ ($\mathcal{P}_1$)}{
\textbf{Minimize Consistency} 
\begin{subequations}
\begin{align}
  & \mP_1(\overline{\rob}, y) \coloneqq  \min_{\con_y, \omega \in \Omega}\quad \con_y \quad \text{s.t.}\\
& \alg(A_\omega, I, y) \leq \overline{\rob} \, \opt(I), \ \forall I \in \mI, \label{cons_1a}\\
     &\alg(A_\omega, I, y) \leq \con_y  \opt(I), \ \forall I \in \mI_y. \label{cons_1b}
\end{align}   
\end{subequations}
Constraint~(\ref{cons_1a}) enforces the initial robustness level $\overline{\rob}$. \label{optimization_A}}
\end{minipage}
\hfill 
\begin{minipage}[t]{0.48\textwidth} 
\vspace{0pt} 
\myproblem{ ($\mathcal{P}_2$)}{\textbf{Minimize Robustness}
\begin{subequations}
\begin{align}
  & \mP_2 (\con_y^*, y)\coloneqq \min_{\rob_y, \omega \in \Omega}\quad \rob_y \quad \text{s.t.}\\
 & \alg(A_\omega, I, y) \leq \rob_y \opt(I), \ \forall I \in \mI, \label{cons_2a}\\
     &\alg(A_\omega, I, y) \leq \con_y^*  \opt(I), \forall I \in \mI_y.\label{cons_2b}
\end{align}   
\end{subequations}
Constraint~(\ref{cons_2b}) enforces the optimal prediction-specific consistency $\con_y^*$ found in Problem~\ref{optimization_A}. \label{optimization_B}}
\end{minipage}

\subsection{Meta-Algorithm Formulation}

The preceding approach (Problems~\ref{optimization_A} and~\ref{optimization_B}) focuses on deriving an instance-specific solution given a fixed prediction $y$. To operationalize this idea in the online setting, where predictions may vary,
we introduce a meta-algorithm (Algorithm~\ref{alg:meta}) that invokes this two-stage optimization procedure for any realized prediction. This enables prediction-specific robust and consistent decision-making across varying inputs, with a user-specified robustness bound $\overline{\rob}$.

\begin{algorithm}[t!]
\caption{\textsc{Bi-Level Optimization-Based Meta-Algorithm}}
\label{alg:meta}
\begin{algorithmic}[1]
\State \textbf{Input:} Desired robustness upper bound $\overline{\rob} \in \Lambda_\rob \coloneqq [\CR^*, +\infty)$
\State Receive a prediction $y$;
\State Compute $\{\con_y^*, \overline{\omega}\}$ by solving $\mathcal{P}_1 (\rob, y)$ (Problem~\ref{optimization_A});
\State Obtain $\{\rob_y^*, \omega^*\}$ by solving $\mathcal{P}_2 (\con_y^*, y)$ (Problem~\ref{optimization_B}) \label{alg:bilevel_opt_2};
\State Deploy the online decision rule induced by $A_{\omega^*}$;
\end{algorithmic}
\end{algorithm}

In the following proposition, which we prove in~\Cref{appendix:oba}, we show that this meta-algorithm yields decisions on the prediction-specific Pareto frontier and satisfying the desired robustness bound. Moreover, the condition that $\overline{\gamma}$ is on the (non-prediction-specific) Pareto frontier guarantees the weak optimality of Algorithm~\ref{alg:meta}; this holds if tight Pareto fronts are available off-the-shelf for the problem at hand, such as the tight tradeoffs between classical consistency and robustness available for ski rental and one-max search~\cite{Kumar2018,Wei2020,Sun2021}.

\begin{proposition}
    \label{prop:oba}
Suppose there exists a weakly-optimal algorithm with robustness $\overline{\rob}$.
With $\omega^*$ being an optimal solution of Problem~\ref{optimization_B} in line~\ref{alg:bilevel_opt_2} of Algorithm~\ref{alg:meta}, $A_{\omega^*}$ is $\con_y^*$-consistent and $\rob_y^*$-robust with respect to the prediction $y$, with $\rob_y^* \leq \overline{\rob}$. Furthermore, Algorithm~\ref{alg:meta} is \strongly.
\end{proposition}

Notably, the meta-algorithm (Algorithm~\ref{alg:meta}) offers a systematic pipeline to achieve Pareto optimality for each prediction $y \in \mY$ while also guaranteeing weak optimality. While the tractability of solving the constituent bi-level optimization problems $\mathcal{P}_1 (\rob, \y)$ and $\mathcal{P}_2 (\con_y^*, \y)$ depends on the structure of the specific online problem, this framework provides a foundation for deriving explicit and tractable \strongly algorithms for the problems we discuss in the remainder of this paper.

\subsection{Generalizations of this framework}

We briefly note that the bi-level optimization framework we have proposed is quite flexible, and could be generalized in a number of ways to enable its application to different problem settings or objectives. For example, Problems~\ref{optimization_A} and \ref{optimization_B} 
could be augmented to include practical considerations such as risk- or uncertainty-sensitive constraints and objectives~\cite{dinitz2024controlling,christianson2024risk,Sun2024} or 
tolerance to erroneous predictions. In particular, in this work, we develop an error-tolerant variant of this framework, motivated by the continuous nature of certain problems (e.g., prices in the one-max search problem) and the fact that prediction errors are often unavoidable in practice. Instead of analyzing the tradeoff between consistency and robustness under perfectly accurate predictions, we can instead study the tradeoff between $\epsilon$-consistency and robustness, where $\epsilon$-consistency denotes the worst-case guarantee when the prediction error is bounded by a chosen constant $\epsilon$. We further describe this generalized framework, and how to design algorithms in this setting, in \Cref{sec:brittle}.

While our approach could, in principle, be extended to more general and complex dynamic multi-round problems, we anticipate that such an extension would likely require nontrivial technical developments. In particular, such an extension would require a number of additional structural assumptions, e.g., that the overall prediction space $\mY$ remains fixed and independent of the algorithm’s actions, and that the algorithm’s actions exert negligible influence on future prediction values. Moreover, depending on the problem, the resulting optimization problems may be considerably harder to analyze. For complex problems like metrical task systems and non-clairvoyant scheduling with $n$ jobs, achieving even weak optimality remains an open question \cite{Christianson2023,Wei2020}; as such, designing strongly-optimal algorithms is inherently challenging. Thus, while our methodology in its current form may not be directly applicable to these settings, we posit that it may still serve as a useful conceptual framework for designing prediction-specific algorithms for these problems. In the rest of this paper, we will consider problem settings where this methodology can be tractably instantiated.

\section{Deterministic Ski Rental}
\label{sec:dsr}

We begin our investigation of strongly-optimal algorithm design in specific problem settings by considering deterministic algorithms for the discrete-time ski rental problem described in~\Cref{sec:online_problems}.
If a deterministic decision-maker buys skis at the start
of day $M\in\mathbb{N}$ (that may depend on the predicted day $\y$), then the induced cost is 
\[\alg_{\mathrm{DSR}}(M,\x,\y)\coloneqq x\cdot\mathds{1}_{M(\y)>\x}+(b+M(\y)-1)\cdot\mathds{1}_{M(\y)\leq \x},\]
where $b \in \mathbb{N}$ denotes the price of the skis. Then, for this problem, the general definitions of the prediction-specific consistency and robustness in~\Cref{eq:ps_cons_rob} are instantiated as:
\begin{align}
    \con_\y^{\mathrm{DSR}}(A_{M(\y)}) &\coloneqq \frac{\alg_{\mathrm{DSR}}(M,\y,\y)}{\min \{b, \y\}}, \quad 
    \rob_\y^{\mathrm{DSR}}(A_{M(\y)})\coloneqq \sup_{\x \in \mathbb{N}} \frac{\alg_{\mathrm{DSR}}(M,\x,\y)}{\min \{b, \x\}}. \label{eq:ps_con_rob_sr}
\end{align}

\subsection{The Deterministic Algorithm of Kumar et al. is Not Strongly-Optimal}

We first provide a deeper analysis of the deterministic algorithm proposed by Kumar et al.~\cite{Kumar2018}, examining its prediction-specific consistency and robustness.

The deterministic algorithm of Kumar et al.~\cite[Algorithm~2]{Kumar2018}, which we denote by $\KD$, purchases at the beginning of day $\lceil \lambda b \rceil$ if $y \geq b$, and at the beginning of day $\lceil b/\lambda \rceil$ otherwise. It achieves $(1+\lambda)$-consistency and $(1+1/\lambda)$-robustness, where $\lambda \in (0,1)$ is a tunable hyper-parameter. By the lower bound of Wei and Zhang~\cite{Wei2020}, $\KD$ is \weakly as the price $b \to +\infty$. However, as we show in the following theorem (which is proved in~\Cref{appendix:dsr_1}), the $\KD$ algorithm is not \strongly.

\begin{theorem}
    \label{thm:kd}
    $\KD$ is $1$-consistent and $(1 + 1/\lam)$-robust when $\y < b$, and $(1 + \lam)$-consistent and $(1+1/\lam)$-robust when $\y \geq b$. Furthermore, $\KD$ is not \strongly even for $b \to +\infty$.
\end{theorem}  

\subsection{A Strongly-Optimal Algorithm for Deterministic Ski Rental}

We now turn to the design of an algorithm that is \strongly in the asymptotic regime $b \to \infty$. Suppose the decision maker buys at the beginning of day $M$. The objective is to determine $M$ for each $y \in \mY$ such that the resulting prediction-specific consistency-robustness pair $(\con_y(M), \rob_y(M))$ is not dominated by that of any alternative decision $M' \neq M$, while remaining within the universal consistency–robustness bound $(1+\lambda,\, 1+1/\lambda)$. We consider the following two cases.

\noindent \textbf{Case I:} $y < b$.

Given the decision $M$, the prediction-specific consistency and robustness are
\[
\con_y (M) = \begin{cases}
        \frac{M-1+b}{y} & \text{if } M \leq y,\\
        1 & \text{if } M > y.
    \end{cases}
\qquad 
\rob_y (M) = \begin{cases}
        \frac{M-1+b}{M} & \text{if } M \leq b,\\
        \frac{M-1+b}{b} & \text{if } M > b.
    \end{cases}
\]

It's clear that the decision $M = b$ simultaneously achieves the optimal consistency 1 and the optimal robustness $2 - \frac{1}{b}$ with respect to $y$, thus dominating any other decisions.

\noindent \textbf{Case II:} $y \geq b$.

Given the decision $M$, the prediction-specific consistency and robustness are
\[
\con_y (M) = \begin{cases}
    \frac{M-1+b}{b} & \text{if } M \leq y,\\
    \frac{y}{b} & \text{if } M > y.
\end{cases}
\qquad 
\rob_y (M) = \begin{cases}
        \frac{M-1+b}{M} & \text{if } M \leq b,\\
        \frac{M-1+b}{b} & \text{if } M > b.
    \end{cases}
\]

As shown in \Cref{fig:kd_dsr_1} and \ref{fig:kd_dsr_2}, we observe that choosing $M = y + 1$ dominates all options with $M > y + 1$, as it offers better robustness without compromising consistency, making $M = y + 1$ a strong choice. Note that ${\con_y(y+1) = {y}/{b}}$ and ${\rob_y (y+1) = {(y+b)}/{b}}$. Therefore, two critical decision boundaries within $(0, b)$ are $M_1 = y+1-b$ and ${M_2 = {(b^2-b)}/{y}}$.\footnote{Note that $\frac{b^2 - b}{y}$ may not be an integer.} Specifically, when $M_1 < M_2$, $M=y+1$ dominates any choice within $(M_1, M_2)$.
\begin{figure}[htbp]
  \centering
  \begin{minipage}[t]{0.48\textwidth}
    \centering
    \includegraphics[width=\linewidth]{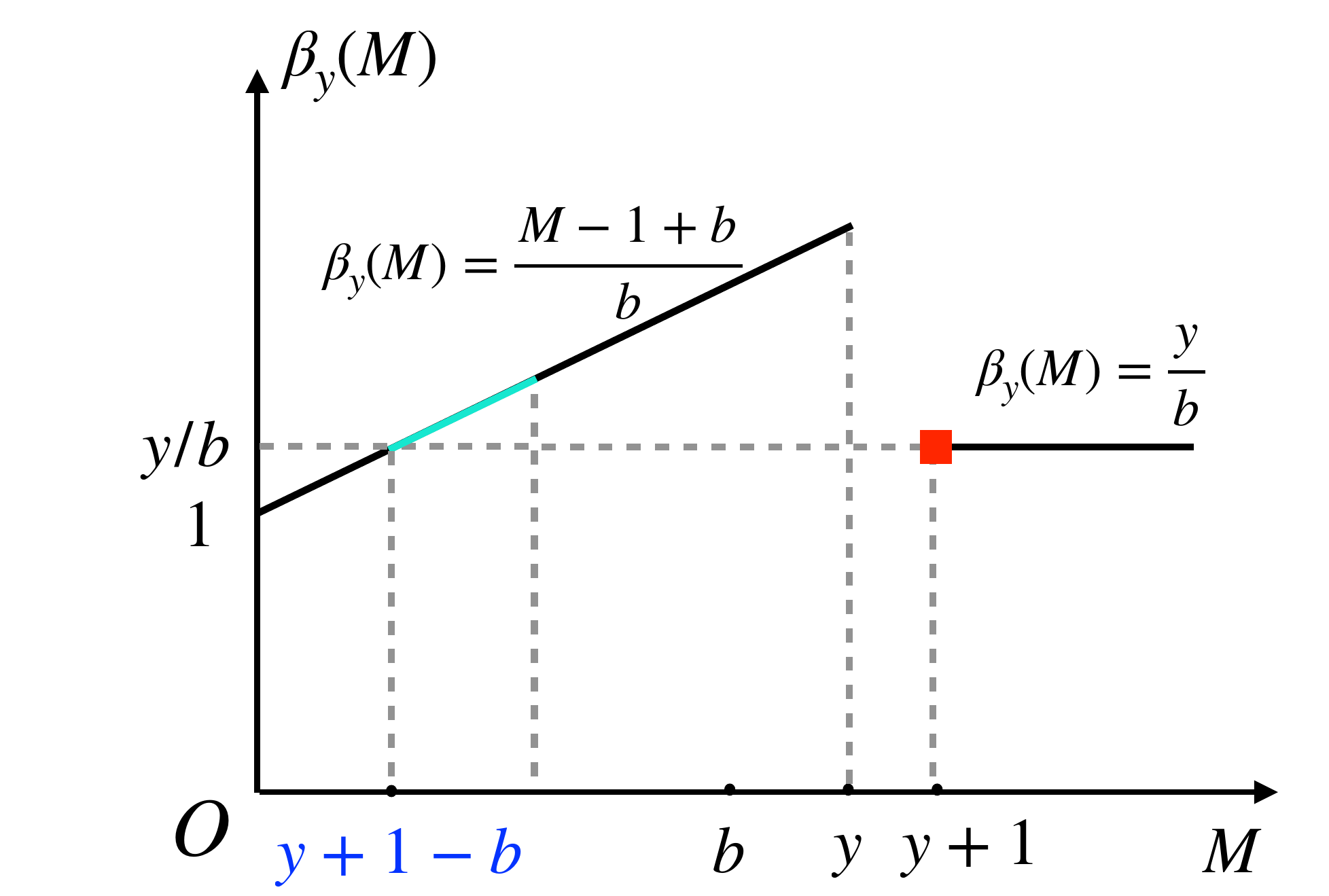}
    \caption{Prediction-specific consistency $\con_y(M)$ versus decision $M$ given $y \geq b$.}
    \label{fig:kd_dsr_1}
  \end{minipage}
  \hfill
  \begin{minipage}[t]{0.48\textwidth}
    \centering
    \includegraphics[width=\linewidth]{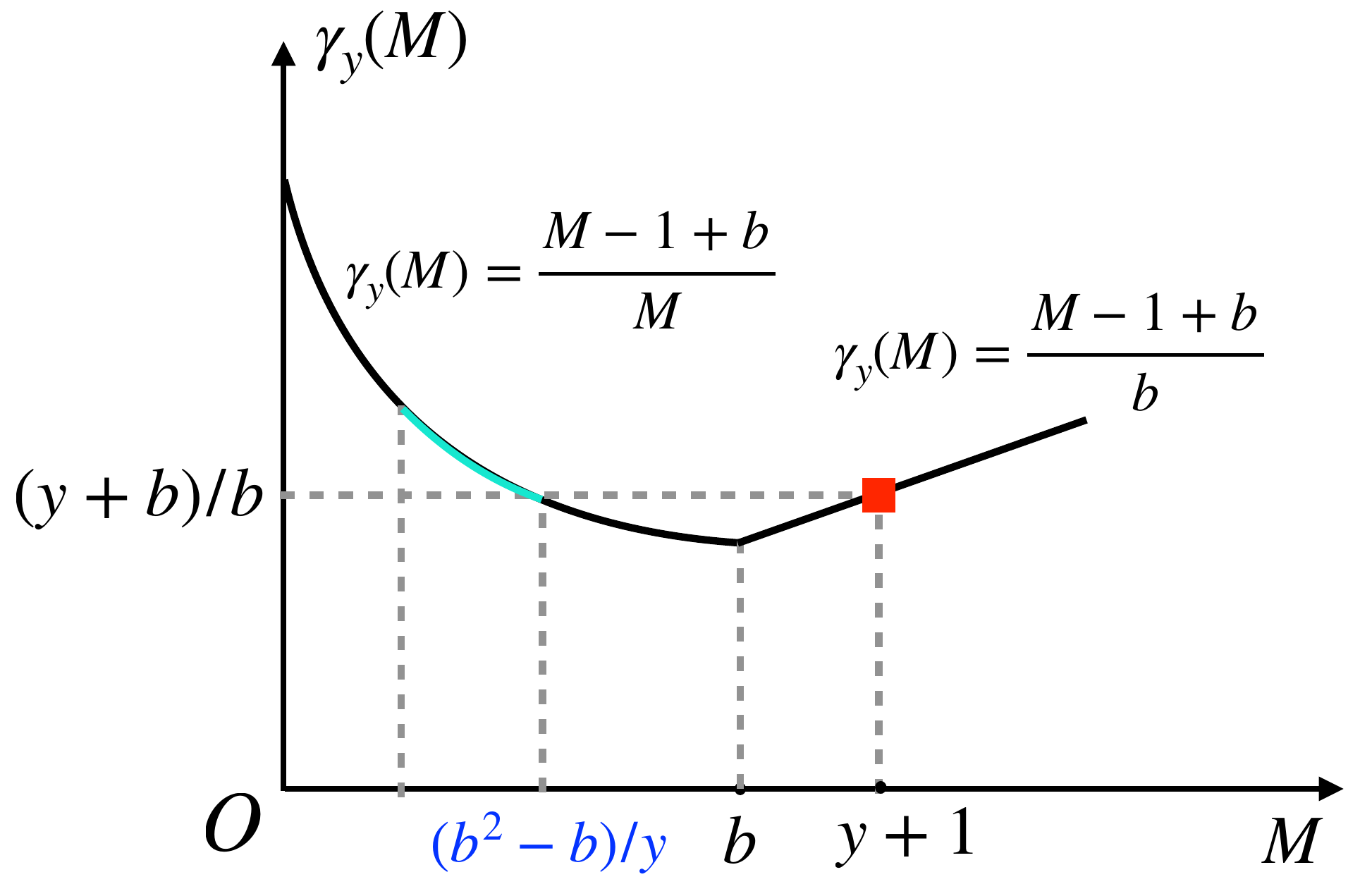}
    \caption{Prediction-specific robustness $\rob_y(M)$ versus decision $M$ given $y \geq b$.}
    \label{fig:kd_dsr_2}
  \end{minipage}
\end{figure}
Our strategy is to set $M = \lceil \lambda b \rceil$ as the primary decision, and to take $M = y+1$ whenever the primary choice is dominated. Given $\lambda \in (0,1)$, the condition that $y + 1$ dominates $\lceil \lambda b \rceil$ in the asymptotic regime $b\to \infty$ requires that $y$ simultaneously satisfies the following constraints: (i) $\lam b \geq y+1 - b$, (ii) $\lam b \leq {(b^2 - b)}/{y}$, which together imply ${y \in [b,  \min \left\{b (\lam+1) - 1, {(b-1)}/{\lam}\right\}]}$.

Therefore, when
$y \in \left[ b,\ \min \left\{ b(\lambda + 1) - 1,\ ({b - 1})/{\lambda} \right\} \right]$, $\lceil \lambda b \rceil$ is dominated by $y + 1$, thereby making $y + 1$ a better choice. On the other hand, when
$y > \min \left\{ b(\lambda + 1) - 1,\ ({b - 1})/{\lambda} \right\}$, $
\lceil \lambda b \rceil$ remains on the Pareto front. 

All together, these cases motivate the design of our deterministic algorithm, $\PDSR$ (Algorithm~\ref{alg:dsr}). In the following theorem, we characterize the prediction-specific consistency and robustness of $\PDSR$, establishing its strong optimality in the asymptotic regime $b \to +\infty$; the full proof is deferred to \Cref{appendix:dsr_2}.

\begin{algorithm}[t!]
\caption{\textsc{\textbf{PDSR}: \textbf{P}rediction-Specific \textbf{D}eterministic \textbf{S}ki \textbf{R}ental}}
\label{alg:dsr}
\begin{algorithmic}[1]
\State \textbf{Input:} $\lam \in (0,1)$
\State \textbf{If} $\y < b$ \textbf{then} determine $M = b $; 
\State \textbf{Else if} {$\y \in[b, \min\{b(\lam+1)-1, \frac{b-1}{\lam}\}]$}  \textbf{then} determine $M = y+1$; 
\State \textbf{Else if} {$ y >\min\{b(\lam+1)-1, \frac{b-1}{\lam}\}$} \textbf{then} determine $M = \lceil \lam b\rceil$; 
\State Buy skis on day $M$.
\end{algorithmic}
\end{algorithm}

\begin{theorem}
    \label{thm:dsr}
    $\PDSR$ presented in Algorithm~\ref{alg:dsr} is \strongly when $b \to +\infty$, and is
    \[
    \begin{cases}
        1\text{-consistent and } (2 - \tfrac{1}{b})\text{-robust} & \text{if } \y < b; \\
        \tfrac{\y}{b}\text{-consistent and } \left(1 + \tfrac{\y}{b}\right)\text{-robust} & \text{if } \y \in \left[ b,\ \min\big\{ b(\lam + 1) - 1,\ \tfrac{b - 1}{\lam} \big\} \right];\\
        (1 + \lam)\text{-consistent and } (1 + \tfrac{1}{\lam})\text{-robust} & \text{if } \y > \min\big\{ b(\lam + 1) - 1,\ \tfrac{b - 1}{\lam} \big\}.
    \end{cases}
    \]
\end{theorem}

\section{Randomized Ski Rental \label{sec:rsr}}

We now turn to considering the design of randomized algorithms for the discrete-time ski rental problem.
Consider a randomized algorithm $\pi=(\pi_i)_{i\in\mathbb{N}_+}$ (that may depend on the predicted day $\y$) that chooses to buy skis at the start of day $i\in\mathbb{N}_+$ with some probability $\pi_i$. The algorithm's average cost is 
\[{\alg_{\mathrm{RSR}}(\pi,\x,\y)\coloneqq \sum_{i=1}^{+\infty} \pi_i(y) \cdot[x\cdot\mathds{1}_{i>\x}+(b+i-1)\cdot\mathds{1}_{i\leq \x}]},\] 
where $b \in \mathbb{N}$ denotes the price of the skis. For this problem, the prediction-specific consistency and robustness in~\Cref{eq:ps_cons_rob} are instantiated as:
\begin{align}
    \con_\y^{\mathrm{RSR}}(A_{\pi(\y)}) &\coloneqq \frac{\alg_{\mathrm{RSR}}(\pi,\y,\y)}{\min \{b, \y\}}, \quad 
    \rob_\y^{\mathrm{RSR}}(A_{\pi(\y)})\coloneqq \sup_{\x \in \mathbb{N}} \frac{\alg_{\mathrm{RSR}}(\pi,\x,\y)}{\min \{b, \x\}}. 
    \label{ps_con_rob_rsr}
\end{align}

\subsection{The Randomized Algorithm of Kumar et al. is Not Strongly-Optimal}

Kumar et al.~\cite{Kumar2018} propose a randomized ski rental algorithm, which we denote by $\KR$, which is a variant of Karlin’s classical randomized strategy~\cite{Karlin1988}. Given a prediction $y$ and a hyper-parameter $\lambda \in (1/b,1)$, $\KR$ chooses the purchase day $i$ according to the distribution
\begin{align}
\pi_i =
\begin{cases}
\left( \tfrac{b-1}{b} \right)^{m-i} \cdot \tfrac{1}{\,b \left( 1 - \left( 1 - \tfrac{1}{b} \right)^m \right)}, & \text{if } i \le m, \\
0, & \text{if } i > m,
\end{cases}
\label{eq:KR-distribution}
\end{align}
where
\[
m \coloneqq 
\begin{cases}
\lceil b/\lambda \rceil, & \text{if }y < b, \\
\lfloor \lambda b \rfloor, & \text{if } y \ge b .
\end{cases}
\]
Under this strategy, $\KR$ achieves consistency $\smash{{\lam}/({1 - e^{-\lam}})}$ and robustness $\smash{{1}/({1 - e^{-(\lam - 1/b)}})}$~\cite{Kumar2018}.
It is known that $\KR$ is \weakly as $b \to +\infty$~\cite{Wei2020}; however, as we show in the following theorem, which is proved in \Cref{appendix:kr}, $\KR$ is not 
\strongly.

\begin{theorem}
    \label{thm:KR}
    $\KR$ is not \strongly, even for $b \to +\infty$.
\end{theorem}

\subsection{A Strongly-Optimal Algorithms for Randomized Ski Rental}

In the analysis that follows, we focus on designing a \strongly algorithm for randomized ski rental. We use $\ratio (\prob, \x)$ to denote the expected ratio between the cost achieved by a randomized algorithm that uses a distribution \( \prob \) over purchase days and the offline solution when the actual ski season lasts \( \x \) days, i.e.,
\[
    \ratio (\prob, \x) \coloneqq{\alg_{\mathrm{RSR}}(\pi,\x,\y)}/{\min\{b, x\}}.
\]
Let $\con_y(\pi)$ and $\rob_y(\pi)$ denote the prediction-specific consistency and robustness of $\prob$; thus,
\[\con_y (\pi) = \ratio(\pi, y) \qquad \rob_y (\pi) = \max_{x \in \mathbb{N}_+} \ratio(\pi, x).\]

\textbf{An Optimization-Based Algorithm.} While the bi-level optimization-based meta-algorithm (Algorithm~\ref{alg:meta}) provides a general method for computing strongly-optimal strategies, solving the underlying optimization problems can be complicated in general---in particular, since they require optimizing over all possible algorithms. Fortunately, for randomized ski rental, it is possible to restrict the support of the randomized strategy to a finite set $\mU(b, y) \coloneqq [b] \cup \{y+1\}$ without loss of optimality, thus enabling a computationally tractable solution to these problems (see \Cref{lemma:restrict} in \Cref{appendix:supporting_lemma}). Given a target robustness level $\overline{\rob} \in \Lambda_\rob \coloneqq \left[e_b/(e_b-1), +\infty\right)$, where $e_b \coloneqq (1 + \frac{1}{b-1})^b$, this structural result allows us to specialize the bi-level optimization framework introduced in~\Cref{sec:bi-level_opt_method} as follows:

\noindent 
\begin{minipage}[t]{0.48\textwidth} 
\vspace{0pt} 
\textbf{Problem 1 ($\mathcal{P}^\mathrm{RSR}_1$): Minimize Consistency}
\begin{subequations}
\label{optimization_1}
\begin{align}
    \min_{\pi, \con_y}\quad& \con_y\\
    \text{s.t.}\quad& \ratio(\pi,x) \leq \overline{\rob}, \forall x \in \mU(b,y) \\
    &\ratio (\pi, y) \leq \con_y
\end{align}   
\end{subequations}
Let $\{\overline{\pi}, \con^* \}$ denote the optimal solution.
\end{minipage}
\hfill 
\begin{minipage}[t]{0.48\textwidth} 
\vspace{0pt} 
\textbf{Problem 2 ($\mathcal{P}^\mathrm{RSR}_2$): Minimize Robustness}
\begin{subequations}
\label{optimization_2}
\begin{align}
    \min_{\pi, \rob_y}\quad& \rob_y\\
    \text{s.t.}\quad& \ratio(\pi,x) \leq \rob_y, \forall x \in \mU(b,y) \\
    &\ratio(\pi, y) \leq \con_y^*
\end{align}   
\end{subequations}
Let $\{\pi^*, \rob_y^* \}$ denote the optimal solution.
\end{minipage}

\begin{algorithm}[t!]
\caption{Variant of Algorithm~\ref{alg:meta} for randomized ski rental}
\label{alg:orsr}
\begin{algorithmic}[1]
\State \textbf{Input:} $\overline{\rob} \in \Lambda_\rob \coloneqq [e_b/(e_b - 1), \infty)$;
\State Compute $\{\con_y^*, \overline{\pi}\}$ by solving $\mathcal{P}^\mathrm{RSR}_1 (\rob, y)$ (Problem~\ref{optimization_1});
\State Obtain $\{\rob_y^*, \pi^*\}$ by solving $\mathcal{P}^\mathrm{RSR}_2 (\con_y^*, y)$ (Problem~\ref{optimization_2});
\State Choose $i$ randomly according to the distribution $\pi^*$;
\State Buy the skis at the start of day $i$.
\end{algorithmic}
\end{algorithm}

Note that both of these problems are linear programs, and thus the corresponding meta-algorithm (Algorithm~\ref{alg:orsr}) is tractable to implement in this case. As such, we obtain as a consequence of \Cref{prop:oba} the following strong optimality result (see \Cref{app:orsr} for a full proof).

\begin{theorem}
\label{thm:orsr}
    Algorithm~\ref{alg:orsr} is \strongly when $b \to \infty$.
\end{theorem}

\textbf{Explicit Algorithm Design.} While the optimization-based algorithm described in \Cref{alg:orsr} is strongly-optimal, it does not provide much insight into the analytic structure of the resulting probability distribution over purchase days. Complementing this result, we can in fact leverage the problem structure to derive a novel and \emph{explicit} \strongly randomized algorithm \RSR (Algorithm~\ref{alg:rsr}), whose construction we detail in Appendix~\ref{app:rsr}. The algorithm builds on two transformation procedures, \opA (Algorithm~\ref{alg:op_A}) and \opB (Algorithm~\ref{alg:op_B}), which systematically adjust the “equalizing distributions” (Theorem~\ref{thm:eq_form}, Equation~\ref{dis:eq}) by reallocating probability mass to trace the prediction-specific Pareto frontier for each $y$. As we show in the following theorem, \RSR, like the optimization-based approach, achieves strong optimality.

\begin{theorem}
    \label{thm:rsr}
    The algorithm \RSR~(\Cref{alg:rsr}) is \strongly when $b \to +\infty$.
\end{theorem}

The proof of \Cref{thm:rsr} is detailed in \Cref{app:proof_prsr}. We begin with a perturbation-based analysis that characterizes the structure of the optimal randomized distribution, and then establish prediction-specific optimality via a primal–dual formulation.

\color{black}

\section{One-Max Search}
\label{sec:one-max_search}

In this section, we shift our focus to the one-max search problem described in~\Cref{sec:online_problems} with predicted maximum price. Let $x$ denote the true highest price and let $y$ denote the predicted highest price. If the decision-maker decides to set their purchase threshold to $\Phi\in[L, U]$---i.e., they purchase at the first price exceeding $\Phi$, which may depend on the prediction $\y$---then their reward is 
\[\alg_{\mathrm{OMS}}(\Phi,\x,\y)\coloneqq \Phi(y)\cdot\mathds{1}_{\Phi(\y) \leq \x}+L\cdot\mathds{1}_{\Phi(\y) > \x}.\] 
For this problem, the prediction-specific consistency and robustness are:
\begin{align}
    \con_\y^{\mathrm{OMS}}(A_{\Phi(\y)}) &\coloneqq \frac{\y}{\alg_{\mathrm{OMS}}(\Phi,\y,\y)}, \quad 
    \rob_\y^{\mathrm{OMS}}(A_{\Phi(\y)})\coloneqq \sup_{\x \in [L, U]} \frac{x}{\alg_{\mathrm{OMS}}(\Phi,\x,\y)}. \label{eq:ps_con_rob_oms}
\end{align}

\subsection{The Algorithm of Sun et al. is Not Strongly-Optimal}
Sun et al.~\cite{Sun2021} proposed a $\con$-consistent and $\rob$-robust  Online Threshold-based Algorithm (\textsc{OTA}) using the threshold
\[
\Phi = 
\begin{cases}
    \label{Sun's}
    L \con, &\text{if } y \in [L, L \con );\\
    \lam L \rob + (1-\lam) y / \con &\text{if } y \in [L\con, L \rob );\\
    L\rob &\text{if } y \in [L\rob, U],
\end{cases}
\]
where $\con = 2\lam\theta/[\sqrt{(1-\lam)^2 + 4\lam \theta} - (1-\lam)]$, $\rob  = \theta / \con$, and $\lam \in [0,1]$. Furthermore, they show that any $\rob$-robust deterministic algorithm for one-max search must have consistency $\con \geq \theta / \rob$, implying their algorithm is \weakly. However, as we establish in the following theorem (which is proved in \Cref{appendix:oms_1}), their algorithm is not \strongly.

\begin{theorem}
    \label{thm:sun}
    Sun's algorithm is not \strongly.
\end{theorem}


\subsection{A Strongly-Optimal Algorithm for One-Max Search}
We propose in \Cref{alg:one-max search} a new approach, $\PST$, which, by more carefully selecting the purchase threshold, achieves strong optimality. In particular, it achieves the prediction-specific consistency and robustness values established in the following theorem.


\begin{algorithm}[t!]
\small
\caption{\textsc{\textbf{$\PST$}: \textbf{P}rediction-\textbf{S}pecific \textbf{T}hresholding}}
\label{alg:one-max search}
\begin{algorithmic}[1]
\State \textbf{Input:} $\lam \in [0,1]$
\State Determine $M = \lam L + (1-\lam)\sqrt{LU}$;
\State \textbf{If} $y \in [L, M]$ \textbf{then} set $\Phi = \sqrt{LU} $; 
\State \textbf{Else if} {$y \in(M, \sqrt{LU}]$}  \textbf{then} set $\Phi = y$; 
\State \textbf{Else if} {$y \in (\sqrt{LU}, U] $} \textbf{then}
determine $\mu = \frac{(1 - \lam)\sqrt{\theta}}{(1-\lam)\sqrt{\theta}+\lam}$ and
 set $\Phi = \mu \sqrt{LU} + (1 - \mu)y$;
\State Perform \textsc{OTA} with threshold $\Phi$.
\end{algorithmic}
\end{algorithm}

\begin{theorem}
    \label{thm:oms}
  PST (Algorithm~\ref{alg:one-max search}) is \strongly, and is
    \[
    \begin{cases}
        (\y / L)\text{-consistent and } \sqrt{\theta}\text{-robust} & \text{if } \y \in [L, \lam L+(1-\lam)\sqrt{LU}); \\
        1\text{-consistent and } (U/\y)\text{-robust} & \text{if } \y \in [\lam L+(1-\lam)\sqrt{LU}, \sqrt{LU}]; \\
        \left(\frac{(1-\lam)\sqrt{\theta}\y + \lam \y}{(1-\lam)U+\lam \y}\right)\text{-consistent and } \left(\frac{(1-\lam)U + \lam \y}{(1-\lam)\sqrt{LU}+\lam L}\right)\text{-robust} & \text{if } \y  \in (\sqrt{LU}, U].
    \end{cases}
    \]
\end{theorem}

We prove \Cref{thm:oms} in \Cref{appendix:oms_2}; the proof identifies prediction-specific optimal thresholds ($\Phi$) by partitioning the prediction space and deriving each threshold (which is a convex combination involving $\mu$) that ensures Pareto-optimality for each segment of predictions.

\color{black}
\section{Error-Tolerant Algorithms}
\label{sec:brittle}

Pareto-optimal algorithms can exhibit \emph{brittleness}, a vulnerability noted by~\cite{elenter2024overcoming, Benomar2025}, where the competitive ratio degrades sharply toward the worst-case robustness bound even with small prediction errors.  This issue stems from the standard definition of consistency (see~\Cref{def:cl_con_rob}), which assumes strictly perfect predictions ($\x(\I)=\y$) and thus fails to consider performance under erroneous predictions.  To address this, we use the one-max search problem as an example to demonstrate how explicit and tractable Pareto-optimal algorithms incorporating a ``generalized consistency'' can be constructed in our prediction-specific framework, offering good performance tradeoffs when faced with minor prediction errors. It is worth emphasizing that the notion of “error tolerance” considered in this section differs from the concept of “smoothness” discussed in~\cite{elenter2024overcoming,Benomar2025,benomar2025tradeoffs}. The former concerns the tradeoff between generalized consistency (allowing small prediction errors) and robustness (with arbitrarily large errors), whereas the latter requires that the algorithm’s performance degrade smoothly as the prediction error increases. Nevertheless, we view the two notions as related, both contributing to alleviating "brittleness". 

To account for small prediction errors, we specify a desired error tolerance $\epsilon$ and define a relaxed consistent set
$
\mI_y^\epsilon \coloneqq \left\{ I \in \mI : \mL(x(I), y) \leq \epsilon \right\},
$
where $\mL : \mX \times \mY \to \mathbb{R}_{\geq 0}$ is a chosen loss function that measures prediction error.
When substituting $\mI_y$ with $\mI_y^{\ep}$, we refer to the corresponding consistency metrics in \Cref{eq:std_cons_rob_pred} and \Cref{eq:ps_cons_rob} as $\epsilon$-consistency and prediction-specific $\epsilon$-consistency, denoted by $\con^\epsilon$ and $\con_y^\epsilon$, respectively. Substituting $\con$, $\con_y$ and $\mI_y$ with $\con^\eps$, $\con_y^\ep$ and $\mI_y^\eps$ in \Cref{sec:opt-meta}, we can generalize our meta-algorithm to incorporate error tolerance. 

For the one-max search problem, its inherent continuity renders prediction errors unavoidable in practice, underscoring the necessity of incorporating error tolerance. We fix $\mathcal{L}(x, y) = |x - y|$ and assume that $\ep$ is small  relative to the scale of the problem. We propose an error-tolerant algorithm for this problem, $\epsilon$-tolerant \OMS, in \Cref{alg:ep-one-max search}. To conclude the section, we present three theorems that characterize the performance of this algorithm. In particular, \Cref{thm:tol_1} establishes the tradeoffs between $\epsilon$-consistency and robustness achieved by $\epsilon$-tolerant \OMS, together with their prediction-specific counterparts, while \Cref{thm:tol_2} provides a lower bound on the global $\epsilon$-consistency–robustness tradeoff. Finally, \Cref{thm:tol_3} shows that the $\epsilon$-tolerant \OMS algorithm is strongly optimal. The proofs of all three results are given in \Cref{appendix:brittle}.

\begin{algorithm}[t!]
\small
\caption{\textbf{$\ep$-Tolerant \PST}}
\label{alg:ep-one-max search}
\begin{algorithmic}[1]
\State \textbf{Input:} $\lam \in [0,1], \ep > 0$
\State Determine $M = \lam (L+3\ep)+ (1-\lam)(\sqrt{LU} - \ep)$;
\State \textbf{If} $y \in [L, M-2\ep]$ \textbf{then} set $\Phi = \sqrt{LU} $; 
\State \textbf{Else if} {$y \in(M-2\ep, M)$} set $\Phi = M-\ep $; 
\State \textbf{Else if} {$y \in[M, \sqrt{LU} + \ep]$}  \textbf{then} set $\Phi = y - \ep$; 
\State \textbf{Else if} {$y \in (\sqrt{LU} + \ep, U-\ep) $} \textbf{then}
 set $\mu = \frac{(U-2\ep) - LU/(M-\ep)}{(U-2\ep) - \sqrt{LU}}$, $\Phi = \mu \sqrt{LU} + (1 - \mu)(y-\ep)$;
\State \textbf{Else if} {$y \in [U-\ep, U] $} \textbf{then} set $\Phi = LU/(M-\ep)$;
\State Perform \texttt{OTA} with threshold $\Phi$.
\end{algorithmic}
\end{algorithm}

\begin{theorem} 
\label{thm:tol_1}
$\ep$-Tolerant \OMS achieves $[(M-\ep)/L]$ $\ep$-consistency and $U/(M-\ep)$ robustness.
    Specifically, $\ep$-Tolerant \OMS achieves
\[
    \begin{cases}
    [(y+\ep)/L]\ \ep\text{-consistency and } \sqrt{\theta}\ \text{robustness} 
    & \text{if } y \in [L,\ M - 2\ep]; \\[4pt]

    (M-\ep)/L\ \ep\text{-consistency and } U/(M-\ep)\ \text{robustness} 
    & \text{if } y \in (M - 2\ep,\ M); \\[4pt]

    (y+\ep)/(y-\ep)\ \ep\text{-consistency and } U / (y-\ep)\ \text{robustness} 
    & \text{if } y \in [M,\ \sqrt{LU} + \ep]; \\[4pt]

    \frac{(y+\ep)}{\mu\sqrt{LU} + (1-\mu)(y-\ep)}\ \ep\text{-consistency and } 
    \frac{\mu\sqrt{LU} + (1-\mu)(y-\ep)}{L}\ \text{robustness} 
    & \text{if } y \in (\sqrt{LU} + \ep,\ U - \ep); \\[4pt]

    (M-\ep)/L\ \ep\text{-consistency and } U/(M-\ep)\ \text{robustness} 
    & \text{if } y \in [U - \ep,\ U];
\end{cases}
\]
where $M = \lam (L+3\ep)+ (1-\lam)(\sqrt{LU} - \ep)$.
\end{theorem}

\begin{theorem}
\label{thm:tol_2}
    Any $\rob$-robust algorithm has at least $(\theta/\rob)$ $\ep$-consistency, and any algorithm that achieves $\con^\ep$ $\ep$-consistency must be at least $(\theta / \con^\ep)$-robust.
\end{theorem}

\begin{theorem}
\label{thm:tol_3}
    Assume $\ep \leq (\sqrt{LU} - L)/4$. The $\ep$-consistency and robustness of $\ep$-Tolerant $\PST$ (\Cref{alg:ep-one-max search}) are jointly Pareto optimal. Moreover, for every prediction $\y \in \mY = [L, U]$, $\ep$-Tolerant $\PST$ achieves prediction-specific $\ep$-consistency $\con_y^\ep$ and robustness $\rob_y$ that are jointly Pareto optimal. 
\end{theorem}

\color{black}
\section{Numerical Experiments}
\label{sec:exp}

In this section, we evaluate the performance of our proposed algorithms in three different case studies spanning synthetic and real-world settings.\footnote{Our code is publicly available at \url{https://github.com/Bill-SizheLi/Prediction_Specific_Design_of_Learning-Augmented_Algorithms}}

\subsection{Case Study 1: Synthetic Data Experiments for Ski Rental}

We begin by testing the performance of our algorithms for the ski rental problem via simulations on synthetic instances. We let the actual number of skiing days $x$ be a uniformly random integer drawn from $[1, 10b]$, where $b = 100$ is buying costs of skis. The prediction $y$ is generated with accuracy $p$: with probability $p$, the prediction is accurate (i.e., $y =x$), and with probability $(1-p)$, the prediction $y \sim \mathcal{N} (x, \sigma^2)$, where $\sigma = 500$ and the output is rounded and made positive. 

\begin{figure}[t!]
  \centering
  \begin{minipage}[t]{0.48\textwidth}
    \centering
    \includegraphics[width=\linewidth]{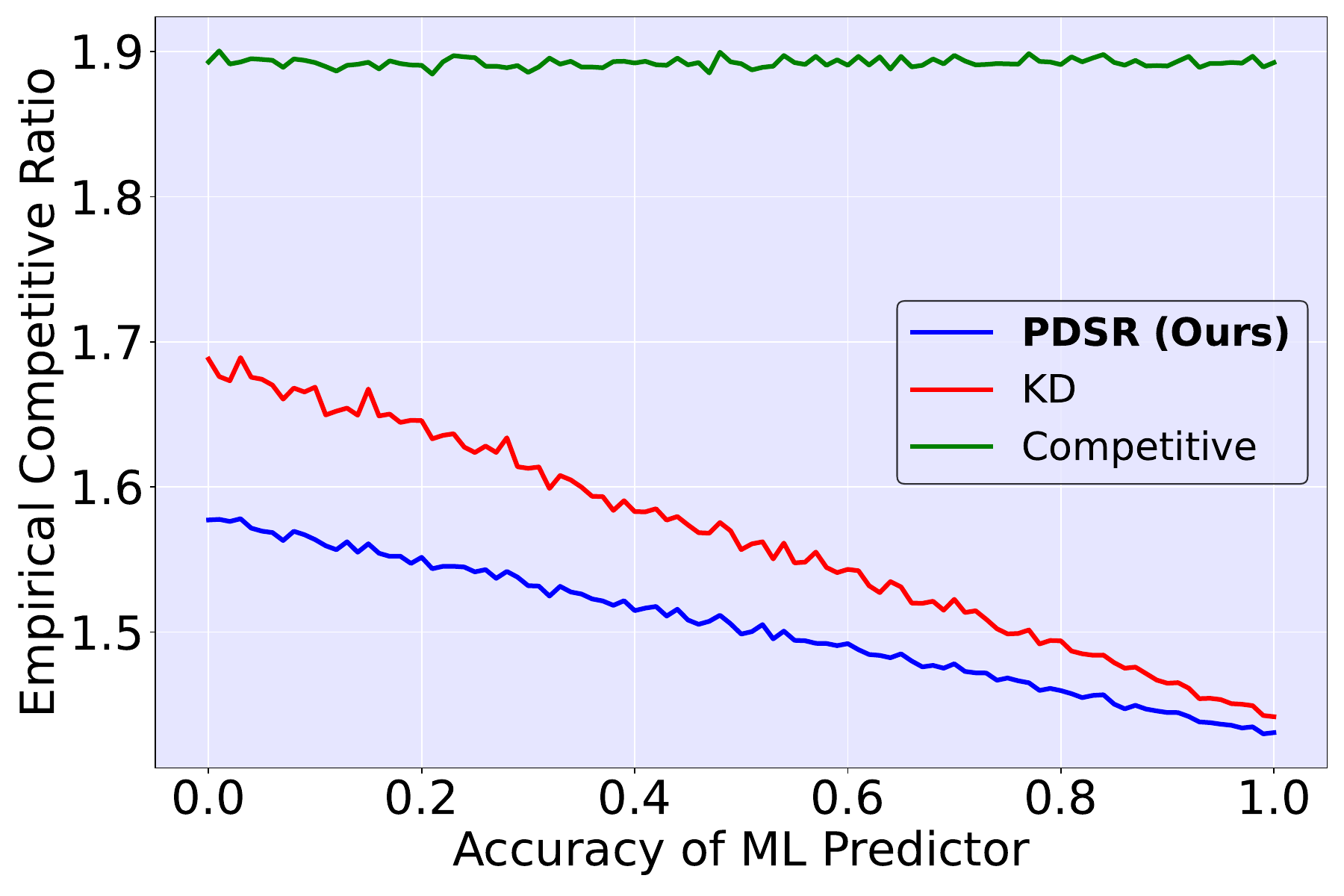}
    \caption{Empirical competitive ratios versus accuracy $p$ in \emph{deterministic setting}.}
    \label{fig:dsr_experiment}
  \end{minipage}
  \hfill
  \begin{minipage}[t]{0.48\textwidth}
    \centering
    \includegraphics[width=\linewidth]{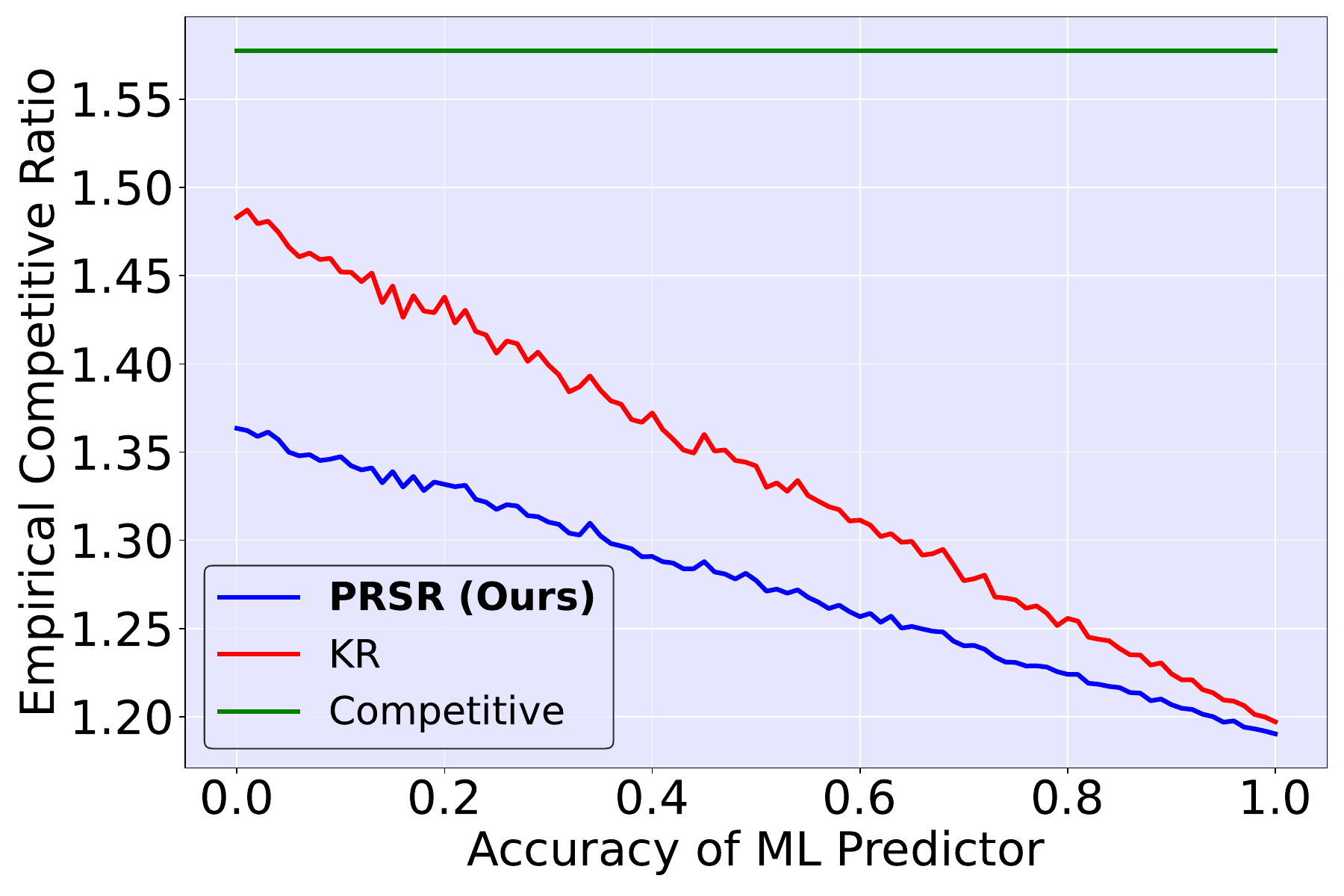}
    \caption{Empirical competitive ratios versus accuracy $p$ in \emph{randomized setting}.}
    \label{fig:rsr_experiment}
  \end{minipage}
\end{figure}

For the deterministic setting, we compare $\PDSR$ (Algorithm~\ref{alg:dsr}) with $\KD$ (\cite{Kumar2018}, Algorithm 2)  and the classic competitive algorithm (which always buys on day b), using the same parameter $\lambda = 0.5$ for both $\PDSR$ and $\KD$. For the randomized setting, we compare $\PRSR$ (Algorithm~\ref{alg:rsr}) with $\KR$ (\cite{Kumar2018}, Algorithm 3) and Karlin’s algorithm~\cite{Karlin1988}, using $\overline{\gamma} = 3$ for $\PRSR$ and $\lambda = \ln (3/2)$ for $\KR$, ensuring PRSR and KR have the same robustness $3$. Each setup is evaluated over $10000$ independent trials. Figures~\ref{fig:dsr_experiment} and~\ref{fig:rsr_experiment} present the empirical results of average competitive ratio versus accuracy~$p$. We observe that our proposed algorithms, $\PDSR$ and $\PRSR$, consistently outperform both classic online algorithms and existing learning-augmented algorithms across both settings.

\subsection{Case Study 2: Ski Rental on Dynamic Power Management}

We next evaluate our ski rental algorithms on \emph{real-world traces} for a Dynamic Power Management (DPM) problem, where we control the idle and active periods of a computing device. Modern processors typically support multiple power states: deeper states disable more components, leading to lower operating cost/energy but higher wake-up penalties/overhead. During each idle interval, a DPM controller must decide whether to stay active or transition into a deeper sleep state without knowing the remaining idle duration.

The two-state DPM system (one active and one sleep state with zero operating cost) is equivalent to the ski rental problem, where remaining active corresponds to \emph{renting} and transitioning to the sleep state corresponds to \emph{buying}. Moreover, Antoniadis et al.~\cite{antoniadis2021dpm} demonstrated that randomized ski rental algorithms can be converted to multi-state DPM algorithms. 

\begin{figure}[t]
  \centering
  \includegraphics[width=0.72\linewidth]{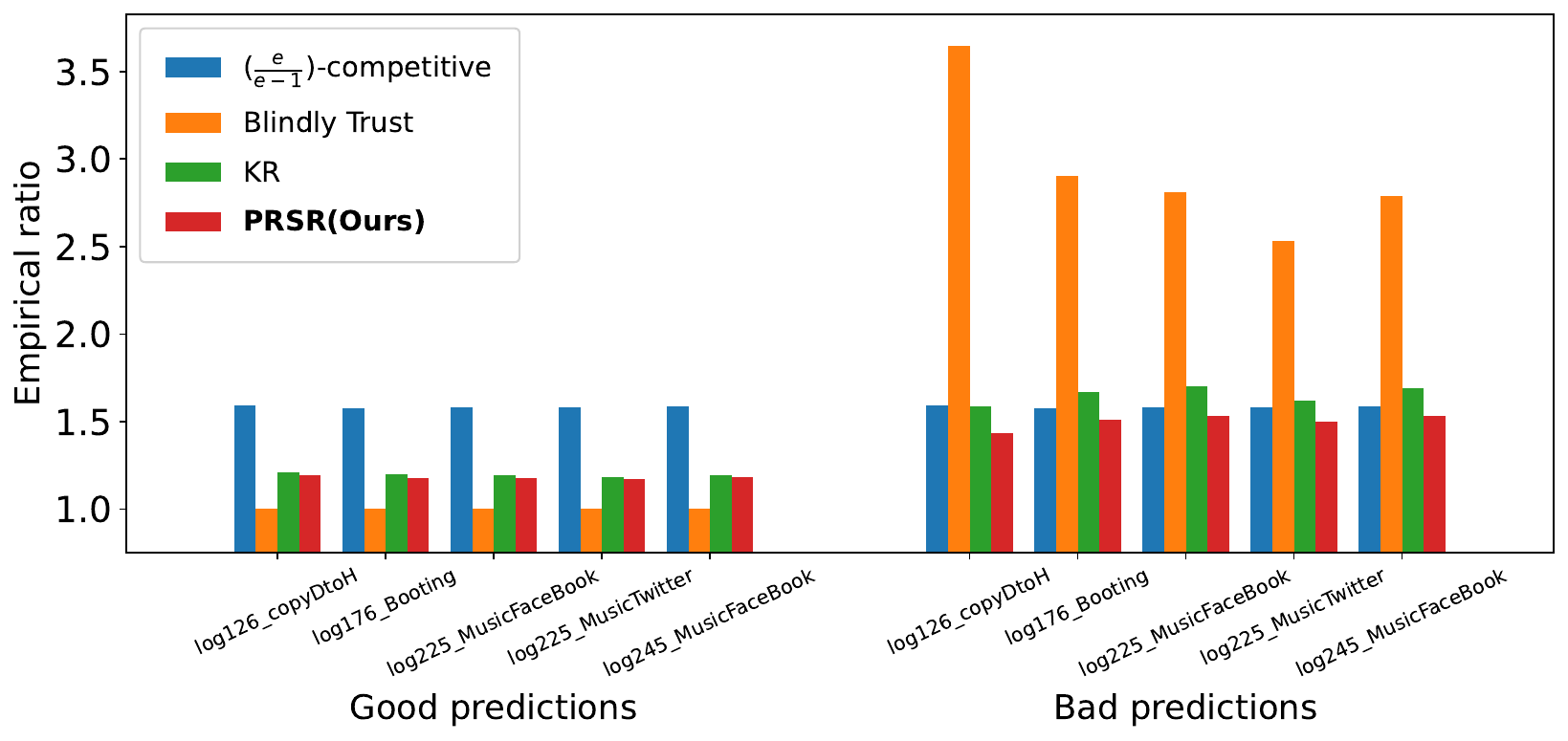}
  \caption{Empirical competitive ratios on DPM traces for good and bad predictions.}
  \label{fig:dmp}
\end{figure}

\textbf{Setup.} We consider a DPM problem with 4 states. Specifically, we use the same problem setting as Antoniadis et al.~\cite{antoniadis2021dpm}, employing I/O traces\footnote{The traces are available at \url{http://iotta.snia.org/traces/block-io}.} collected from a Nexus~5 smartphone~\cite{zhou2015io}, from which idle intervals between consecutive requests are extracted. We adopt the IBM mobile hard-drive power states reported in~\cite{irani2003online}, consistent with the setup in~\cite{antoniadis2021dpm}. The idle periods are scaled in the same way as in~\cite{antoniadis2021dpm}.

We use the five largest traces for evaluation. Since the main goal of this section is to probe the algorithms' performance under the two extremes of very good and very bad predictions, we consider the following method for generating predictions: "good predictions" and "bad predictions" are obtained by perturbing the ground truth with $\mathcal{N}(0, \sigma_{\text{good}}^2)$ and $\mathcal{N}(0, \sigma_{\text{bad}}^2)$ noises, respectively. In this experiment, we set $\sigma_{\text{good}} = 0.02$ and $\sigma_{\text{bad}} = 20$. We compare four algorithms: the classic $(\frac{e}{e-1})$-competitive algorithm, \textsf{Blindly Trust} (which treats the prediction as if it is correct and optimizes accordingly), the randomized algorithm of Kumar et~al. (\KR), and our randomized algorithm (\PRSR). For the learning-augmented algorithms, we use the same parameter values for $\lambda$ and $\overline{\rob}$ as in \textbf{Case Study~1}.

\textbf{Results.} Figure~\ref{fig:dmp} reports the empirical competitive ratios on the real DPM traces. We observe that our strongly-optimal algorithm $\PRSR$ consistently achieves the lowest competitive ratios, except for when the prediction quality is very good (in which case, the non-robust \textsf{Blindly Trust} algorithm outperforms it. 
This validates our algorithm's ability to exploit specific predictions to enable good performance regardless of prediction value or quality.

\subsection{Case Study 3: One-Max Search on VIX Trading}

\begin{wrapfigure}{r}{0.46\textwidth}
\vspace{-1em}
\centering
\includegraphics[width=0.46\textwidth]{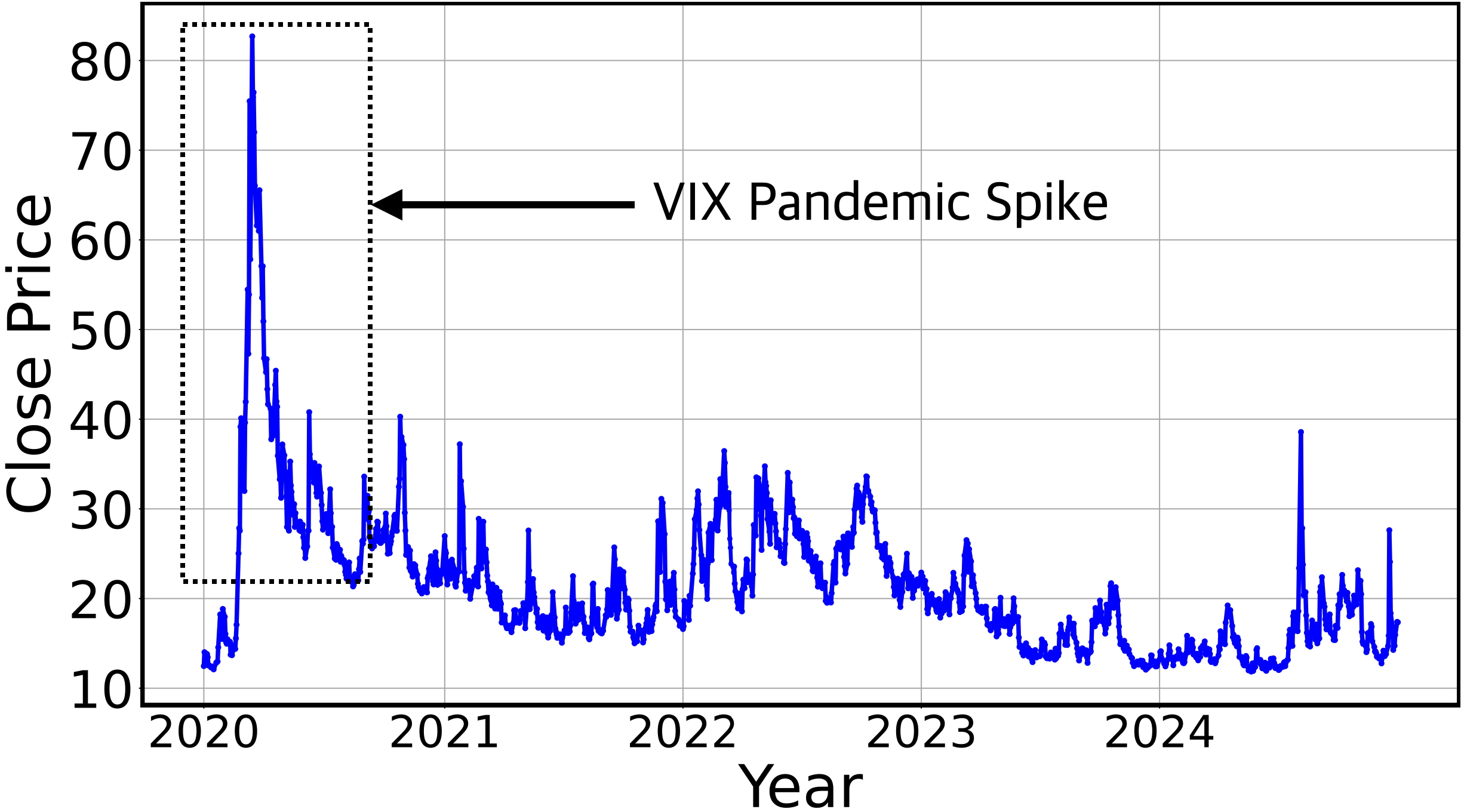}
\caption{VIX Closing Price from January 2020 to December 2024. The VIX index soared from a pre-pandemic level of $12$-$15$ to $83$ in March 2020. 
}
\label{fig:fig3}
\end{wrapfigure}

The VIX, often referred to as the \textit{fear index}, exhibits sharp volatility spikes that make it a natural benchmark for evaluating online search algorithms. Its uncertainty and heavy-tailed dynamics closely mirror the one-max search setting, where the core challenge lies in optimally timing a single exit. This characterization is vividly illustrated by the extreme volatility in early 2020: within just a few weeks at the onset of the COVID-19 market shock, the VIX surged from below \$15 to over \$80 (shown in Figure~\ref{fig:fig3}).

\textbf{Setup.} We evaluate our one-max search algorithms in a case study using the daily closing prices of the VIX index from January 2020 to December 2024 (shown in Figure~\ref{fig:fig3}), which consist of publicly available values obtained from the \href{https://www.cboe.com/tradable_products/vix/vix_historical_data/}{Cboe Options Exchange}. We assume that at the beginning of each month, an agent holds one unit of VIX and must choose a single day within that month to sell it. Over the course of five years, there are 60 trading rounds (one per month), each offering approximately 20 to 21 trading opportunities, as the VIX is only traded on weekdays. We set $L$ and $U$ as the historical minimum and maximum prices over the entire 5 years.\footnote{The focus of this paper is not on the impact of $L$ and $U$; therefore, we simply set them to historical values. In practical trading scenarios, $L$ and $U$ can be viewed as predetermined parameters representing the stop-loss and take-profit thresholds in the exit strategy of the trading process.}

\textbf{Baselines.} We compare our proposed methods, \PST (Algorithm~\ref{alg:one-max search}) and $\epsilon$-Tolerant \PST (Algorithm~\ref{alg:ep-one-max search}), to three baseline algorithms: (i) blindly trusting the prediction, (ii) the classical online algorithm of El-Yaniv~\cite{ElYaniv2001}, and (iii) prior learning-augmented algorithms (Sun’s~\cite{Sun2021} and Benomar’s~\cite{Benomar2025}).

\textbf{Experiment 1.} In this experiment, we consider a naive prediction strategy that simply uses the highest observed VIX price from the previous month. As the evaluation metric, we use the \textit{empirical ratio}\footnote{Note that the empirical ratio here is the inverse of that used in the theoretical analysis, so as to better reflect the proportion of the hindsight optimum that the online or learning-augmented algorithm can achieve (or recover).}, defined as the cumulative online outcome up to the current round divided by the cumulative offline optimum. This metric reflects how well an algorithm performs in practice relative to the hindsight-optimal strategy, averaged over time. We run the algorithms over the 60-month horizon using historical VIX data, and report the empirical ratios at each round to visualize both long-term trends and the stability of performance across different market periods in \Cref{fig:fig1}.

For our proposed algorithms, we fix the trade-off parameter $\lambda = 0.3$ in \PST, and use $\lambda = 0.3$ and $\epsilon = 1.8$ in $\epsilon$-Tolerant \PST. For baseline algorithms with tunable parameters, including those from Sun et al.~\cite{Sun2021} and Benomar et al.~\cite{Benomar2025}, we find that setting $\lambda = 1.0$ yields the best cumulative empirical ratio over the full 60-month horizon. However, to better illustrate performance variation across different regimes, we also include their results under $\lambda = 0.3$ and $\lambda = 0.6$. 

\textbf{Experiment 2.} In practical settings, machine-learned predictions are often more accurate than the naive predictor used in Experiment 1, though they remain imperfect due to model limitations and data noise. The degree of prediction accuracy varies with the capability and training of the underlying ML model. To systematically evaluate algorithmic performance under varying prediction quality, we introduce the notion of an \textit{error level} -- a scalar value between 0 and 1 that quantifies the deviation from perfect information. For each trading round, the prediction is constructed via linear interpolation between the previous month’s maximum (naive prediction) and the current month’s actual maximum (perfect prediction), where the error level determines the interpolation weight. An error level of 1.0 corresponds to the naive prediction, while 0.0 yields the perfect prediction.

We assess the cumulative empirical ratio of each algorithm over all 60 trading rounds under varying error levels from 0 to 1, and report the result in \Cref{fig:fig2}. To ensure a fair comparison across prediction regimes, we fix the trade-off parameter $\lambda = 0.5$ for all tunable methods, including those of Sun, Benomar, \PST, and $\epsilon$-Tolerant \PST. For $\epsilon$-Tolerant \PST, we additionally set $\epsilon = 0.5$ to account for moderate tolerance to prediction error.

\begin{figure}[t!]
  \centering
  \begin{minipage}[t]{0.48\textwidth}
    \centering
    \includegraphics[width=\linewidth]{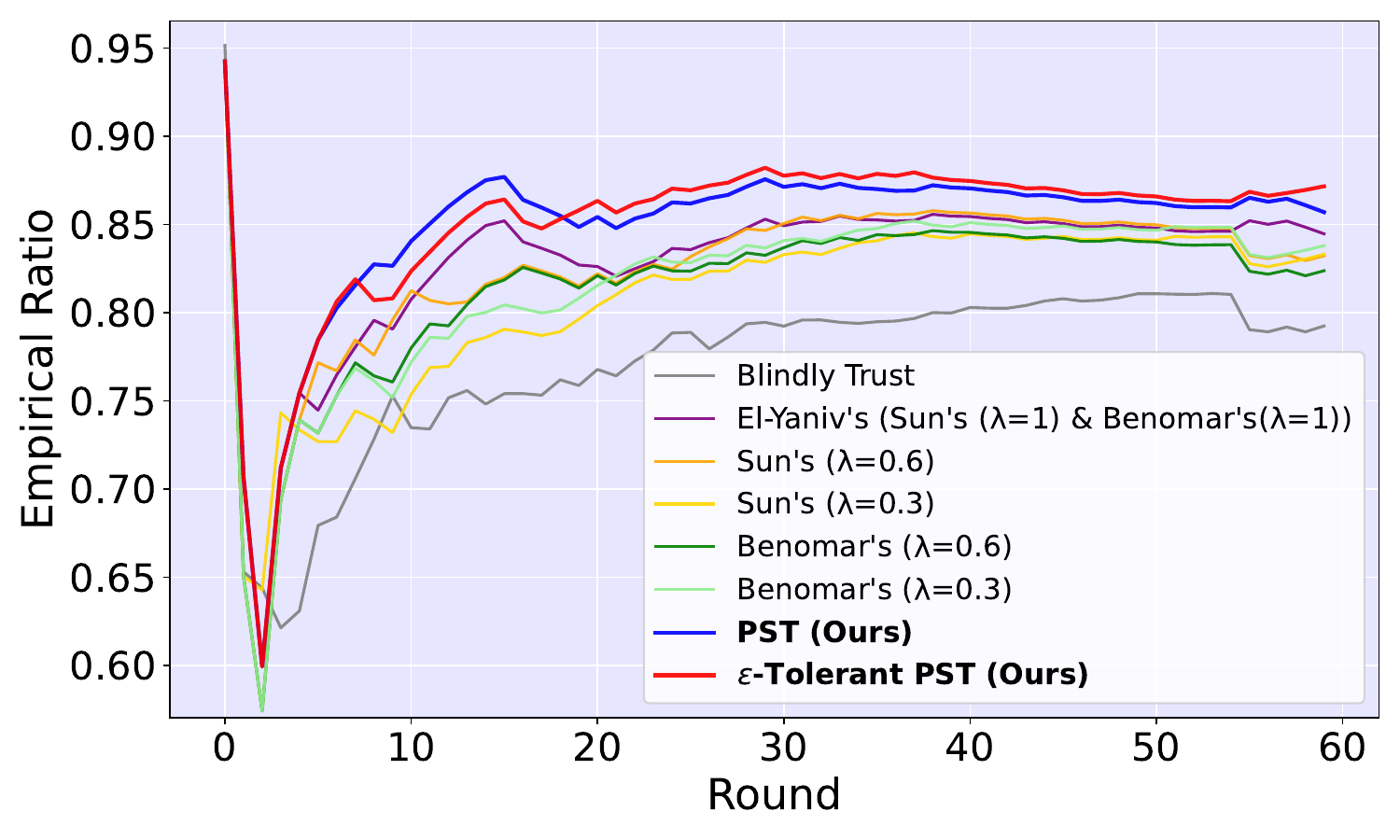}
    \caption{Empirical competitive ratios over $60$ trading rounds.}
    \label{fig:fig1}
  \end{minipage}
  \hfill
  \begin{minipage}[t]{0.48\textwidth}
    \centering
    \includegraphics[width=\linewidth]{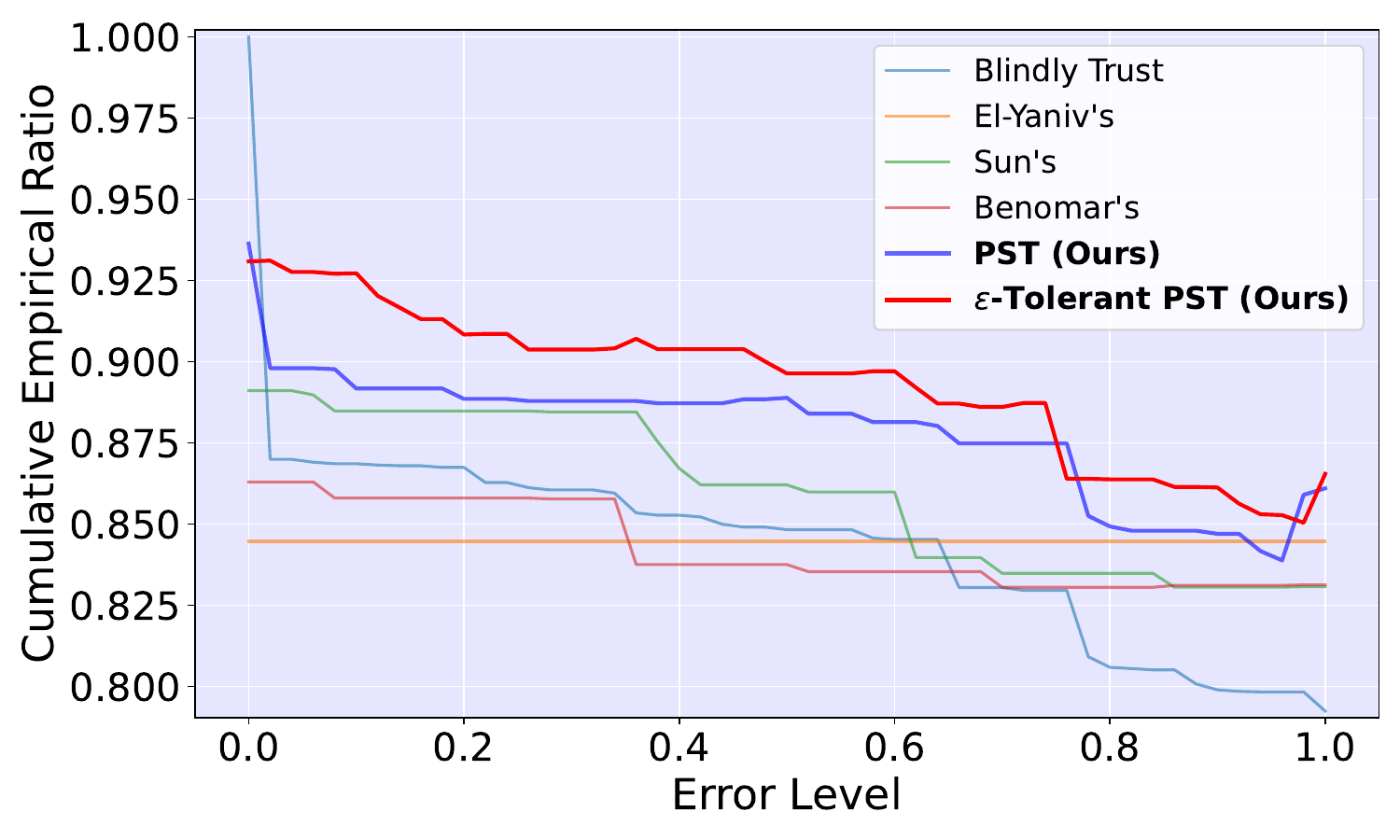}
    \caption{Cumulative empirical competitive ratio with varying prediction error levels.}
    \label{fig:fig2}
  \end{minipage}
\end{figure}

\textbf{Results.} The results for \emph{Experiment 1} are shown in Figure~\ref{fig:fig1}. Our methods consistently outperform all baselines across the decision horizon, with final empirical competitive ratios of $87.2\%$ ($\epsilon$-Tolerant \PST) and $85.7\%$ (\PST), compared to $82.4\%$–$84.4\%$ for all baselines.
The results for \emph{Experiment 2} are shown in Figure~\ref{fig:fig2}; $\epsilon$-Tolerant $\PST$ remains consistently superior across nearly all error levels.

\color{black}
\section{Concluding Remarks}
\label{sec:further_discussion}

In this work, we introduce a prediction-specific analysis framework and a finer-grained notion of \textit{strong optimality} for online algorithms with predictions. We further provide a systematic approach to designing Pareto-optimal online algorithms with better prediction-specific performance than prior algorithms, and we show how this methodology can yield significant performance improvements for the problems of ski rental (deterministic and randomized) and one-max search. 

\textbf{Future Directions.} In contrast to the ski rental and one-max search settings, the existing weakly-optimal non-clairvoyant scheduling algorithm of Wei and Zhang~\cite{Wei2020} is strongly-optimal when $n = 2$. Thus, designing a strongly-optimal algorithm for $n$-job non-clairvoyant scheduling remains an open question. Similarly, as we did for one-max search, developing explicit error-tolerant strongly-optimal algorithms for both deterministic and randomized ski rental is also an interesting future direction. In addition, exploring whether the bi-level optimization in the meta-algorithm (Algorithm~\ref{alg:meta}) in \Cref{sec:opt-meta} can be tractably solved for more complex, multi-stage problems represents a challenging but potentially impactful direction for future study.

\bibliographystyle{unsrtnat}
\bibliography{references}

\appendix 

\newpage

\appendix

\section{Additional Details and Proofs for \Cref{sec:scheduling}\label{appendix:scheduling}}

This section supplements the analysis of the non-clairvoyant scheduling problem for the special case $n = 2$, as introduced in Section~\ref{sec:online_problems}.

We focus on the two-stage scheduling algorithm proposed by Wei and Zhang~\cite{Wei2020}, which is known to achieve the optimal classic trade-off and is \weakly-optimal under Definition~\ref{def_weak}. As detailed in Algorithm~\ref{alg:two-stage schedule} in \Cref{appendix:two-stage schedule}, the algorithm proceeds in two phases based on predicted job lengths $ \y = (\y_1, \y_2) $, where $ \y_1 \leq \y_2 $. In Theorem~\ref{thm:scheduling}, we show that this algorithm also satisfies prediction-specific Pareto optimality under $n = 2$.

\subsection{Wei et al.'s Two-Stage Schedule\label{appendix:two-stage schedule}}

Wei and Zhang~\cite{Wei2020} propose an algorithm called the two-stage schedule (see Algorithm~\ref{alg:two-stage schedule}) that achieves a consistency of $(1+\lambda)$ and a robustness of $(1 + \frac{1}{1+6\lambda})$ under $n=2$. By \cite{Wei2020} and Definition~\ref{def_weak}, this two-stage schedule algorithm is \weak-optimal for $n = 2$. 

\begin{algorithm}[H]
\caption{Two-Stage Schedule}
\label{alg:two-stage schedule}
\begin{algorithmic}[1]
\Statex
\State At any point, if a job finishes with processing time less or more than its prediction, use round robin forever.
\State \texttt{Stage 1:} Round robin for at most $\lambda n \cdot \text{OPT}_y /  \binom{n}{2} $ units of time.
\State \texttt{Stage 2:} Process jobs in predicted order
\Statex
\hspace{2em} (\textit{starting from the unfinished job with the least predicted time}).
\end{algorithmic}
\end{algorithm}

\subsection{Proof of~\Cref{thm:scheduling}} \label{appendix:two-stage-schedule-optimality}

In this subsection, we further analyze the prediction-specific consistency and robustness and prove that two-stage schedule is \strongly under $n=2$.

\begin{proof}[Proof of~\Cref{thm:scheduling}]
    Note that the jobs are ranked based on their predicted lengths; thus we have $\y_1 \leq \y_2$.
    We first prove the prediction-specific consistency and robustness under a specific set of predictions $y=(y_1, y_2) $.
    
We first consider $\lambda \leq \frac{y_1}{2y_1+y_2}$, i.e. $\lambda (2y_1+y_2) \leq y_1$.
Regarding the consistency, assume that $x_1 = y_1, x_2 = y_2$. In stage 1, the algorithm runs round-robin for $2 \lambda\cdot(2y_1+y_2)$ time. Since $\lambda \cdot (2y_1+y_2) \leq y_1$, job 1 cannot finish in stage 1. Therefore, the completion time of job 1 is $$2\lambda (2y_1+y_2) + y_1 - \lambda(2y_1+y_2) = y_1 + \lambda(2y_1+y_2),$$ and that of job 2 is $y_1 + \lambda(2y_1+y_2) + y_2 - \lambda(2y_1+y_2) = y_1 + y_2$, thus, $\text{ALG} = 2y_1 + y_2 + \lambda(2y_1+y_2)$, and $\text{OPT} = 2y_1 + y_2$, yeilding a consistency of $1 + \lambda$.

Regarding the robustness, we consider an adversarial attack $x_1, x_2$. Let $\delta$ denote an infinitesimal quantity.

\noindent\textbf{Case I:} If \colorbox{gray!20}{$x_1 \leq \lambda(2y_1+y_2)$ or $x_2 \leq \lambda(2y_1+y_2)$}, i.e. some incorrect prediction is found in stage 1. In this case, the algorithm runs round-robin from beginning to end, resulting in a robustness of at most $4/3$.

\noindent\textbf{Case II:} If \colorbox{gray!20}{$\lambda(2y_1+y_2) < x_1 \leq y_1$}, i.e. job 1 finishes no later than its prediction. In this case, the algorithm runs round-robin for $2\lambda(2y_1+y_2)$ time and processes job 1 until completion, and then turns to job 2. Thus, 
\begin{align*}
    \text{ALG} &= 2\Big(2\lambda(2y_1+y_2) + \big(x_1 - \lambda(2y_1+y_2)\big)\Big) + \left(x_2 - \lambda(2y_1+y_2)\right) \\ &= 2x_1 +x_2 + \lambda(2y_1+y_2).
\end{align*}

\vspace{2pt}

\textbf{Case II(a):} \colorbox{gray!20}{$\lambda(2y_1+y_2) < x_2 \leq x_1$}
    In this case, $\text{OPT} = 2x_2 + x_1$. This yields a robustness of $1 + \frac{y_1}{y_1+2\lambda(2y_1+y_2)} \geq 4/3$, which is attained when $x_1 = y_1$ and $x_2 = \lambda(2y_1+y_2) + \delta$.

\vspace{2pt}
    
\textbf{Case II(b):} \colorbox{gray!20}{$x_2 > x_1$}
    In this case, $\text{OPT} = 2x_1 + x_2$. This results in a robustness of $1 + \frac{\lambda(2y_1+y_2)}{3y_1} \leq 4/3$, which is achieved when $x_1 = y_1$ and $x_2 = y_1 + \delta$.

\vspace{5pt}

\noindent\textbf{Case III:} If \colorbox{gray!20}{$x_1 > y_1$}, i.e. job 1 finishes later than its prediction. In this case, the algorithm first runs round-robin for $2\lambda \cdot(2y_1+y_2)$ time, then processes job 1 for $y_1 - \lambda(2y_1+y_2)$ time, and finally runs round-robin till the end.

\vspace{3pt}

\textbf{Case III(a)} \colorbox{gray!20}{$\lambda(2y_1+y_2) <x_2 < x_1 + \lambda(2y_1+y_2) - y_1$}, i.e. job 1 finishes later than job 2.
    
    In this case, $\text{ALG} = x_1 + 3x_2 - \lambda (2y_1+y_2)+y_1$ and $\text{OPT} = 2x_2 + x_1$. This yields a robustness of $1 + \frac{y_1}{y_1+2\lambda(2y_1+y_2)} \geq 4/3$, which is achieved when $x_1 = y_1 + \delta$ and $x_2 = \lambda(2y_1+y_2) + \delta$.

\vspace{3pt}
    
\textbf{Case III(b)} \colorbox{gray!20}{$x_2 \geq x_1 + \lambda(2y_1+y_2) - y_1$}, i.e. job 1 finishes no later than job 2.
    
Note that in this case, we have $\text{ALG} = 3x_1 + x_2 + \lambda (2y_1+y_2) -y_1$. 

First, if  $x_1 < x_2$, $\text{OPT} = 2x_1 + x_2$. This generates a robustness of $1 + \frac{\lambda (2y_1+y_2)}{3y_1} \leq 4/3,$ which is achieved when $x_1 = y_1 + \delta$ and $x_2 = y_1 + 2\delta$. 

Otherwise,  if $x_1 + \lambda(2y_1+y_2) - y_1 < x_2 < x_1$, $\text{OPT} = 2x_2 + x_1$. 
This yields a robustness of $1 + \frac{y_1}{y_1 + 2\lambda(2y_1+y_2)} \geq 4/3,$ which is achieved when $x_1 = y_1 + \delta$ and $x_2 = \lambda(2y_1+y_2) + 2\delta$. 

\vspace{7pt}
To sum up, if $\lambda \leq {y_1}/({2y_1+y_2})$, the algorithm is
$
(1+\lambda)$-consistent and $(1 + \frac{y_1}{y_1 + 2\lambda(2y_1+y_2)})$-robust.
Otherwise, if $\lambda > {y_1}/({2y_1+y_2})$, the algorithm runs round-robin forever, and is
\( (1 + \frac{y_1}{{2y_1+y_2}})\)-consistent and $(4/3)$-robust. 

In conclusion, given prediction $y = (y_1, y_2)$, the prediction-specific consistency and robustness are
\[\con_y = 1 + \min\{\frac{{y_1}}{{2y_1+y_2}}, \lambda\}, \qquad \rob_y = 1 + \max\{{1}/{3}, \frac{y_1}{y_1 + 2\lambda(2y_1+y_2)}\}.\] 

Since the two-stage schedule is already proven to be \weakly, we then prove that it is \strong-optimal by demonstrating the Pareto optimality of their prediction-specific consistency and robustness.

When $\lambda \leq {y_1}/({2y_1+y_2})$, the algorithm is $(1 + \lambda)$-consistent and $\left(1 + \frac{y_1}{y_1 + 2\lambda(2y_1+y_2)}\right)$-robust. Consider algorithm $\A$ that completes $r > \lambda(2y_1+y_2)$ amount of work for job 2 when it finishes job 1 in the case where predictions are accurate. Then, it follows that
\begin{align*}
    \text{ALG} &= 2\cdot(y_1 + r) + y_2 - r = 2y_1 + y_2 + r\\
\text{OPT} &= 2y_1 + y_2,
\end{align*}
which yields a competitive ratio of $1 + {r}/({2y_1+y_2}) > 1 + \lambda$. Therefore, any $(1+\lambda)$-consistent algorithm processes at most $\lambda(2y_1+y_2)$ amount of work for job 2 when it finishes job 1 or finds any incorrect prediction of job 1. Then, we consider an incorrect prediction $x_1 = y_1 $, $x_2 = \lambda(2y_1+y_2) + \delta$. Consider a $(1+\lambda)$-consistent algorithm $\B$ that completes $r \leq \lambda(2y_1+y_2)$ amount of work when it finishes job 1. Then, upon $x_1 = y_1$ and $x_2 = r + \delta$, $\text{ALG} = 2(y_1 + r) + (r+\delta -r)=  2y_1 + 2r + \delta$ and $\text{OPT} = 2r + y_1$. This leads to a robustness no better than $$\min_{r \leq \lambda(2y_1+y_2)} \frac{2y_1 + 2r}{y_1+2r} = 1 + \frac{y_1}{y_1+2\lambda(2y_1+y_2)}$$ for all $r \leq \lambda(2y_1+y_2)$.

When $\lambda > {y_1}/({2y_1+y_2})$, the algorithm, equivalent to round-robin (RR), is $(1+{y_1}/({2y_1+y_2}))$-consistent and $4/3$-robust. Note that RR is the only algorithm that achieves $4/3$ competitive ratio. Therefore, under prediction $y= (y_1, y_2)$, no other algorithm achieves a robustness equal to or less than $4/3$. 

By~\Cref{def_strong}, the two-stage schedule is \strong-optimal.
\end{proof}

\newpage
\section{Additional Details for Section~\ref{sec:opt-meta} \label{appendix:oba}}
In this section, we provide a proof of \Cref{prop:oba} and discuss the scenario in which a \weakly algorithm with robustness $ \overline{\rob} $ is unavailable.

\subsection{Proof of \Cref{prop:oba}}
\begin{proof}[Proof of~\Cref{prop:oba}]
Let $\{\con_y^*, \overline{\omega}\}$ and $\{\rob_y^*, \omega^*\}$ denote the optimal solution to Problem~\ref{optimization_A} and Problem~\ref{optimization_B}, respectively. Since $\overline{\rob} \geq \CR^*$, $\{\overline{\rob}, \overline{\omega}\}$ is always a feasible solution to Problem~\ref{optimization_B}. Thus, $\rob_y^* \leq \overline{\rob}$, where $\rob_y^*$ is the optimal objective value to Problem~\ref{optimization_B}. Since $\{\rob_y^*, \omega^*\}$ is a feasible solution to Problem~\ref{optimization_B}, we have 
\begin{align*}
    \alg(A_{\omega^*}, I, y) &\leq \rob_y^* \cdot \opt(I), \ \forall I \in \mI, \\
    \alg(A_{\omega^*}, I, y) &\leq \con_y^* \cdot \opt(I), \ \forall I \in \mI_y.
\end{align*}

By~\Cref{def:ps_con_rob}, $A_{\omega^*}$ is $\con_y^*$-consistent and $\rob_y^*$-robust with respect to $y$.

Consider $\omega' \in \Omega$. Let $\con_y'$ and $\rob_y'$ denote the prediction-specific consistency and robustness of $A_{\omega'}$ with respect to $y$. We consider two cases.

\vspace{4pt}

\noindent \textbf{Case I:} \colorbox{gray!20}{$\con_y' > \con_y^*$} If $\con_y' > \con_y^*$, then $(\con_y', \rob_y')$ produces no Pareto improvement over $(\con_y^*, \rob_y^*)$. 

\vspace{4pt}

\noindent \textbf{Case II:} \colorbox{gray!20}{$\con_y' \leq \con_y^*$}. If $\con_y' \leq \con_y^*$, since $A_{\omega'}$ is $\con_y'$-consistent and $\rob_y'$-robust with respect to $y$, by~\Cref{def:ps_con_rob},
\begin{subequations}
\begin{align}
\label{a}
    \alg(A_{\omega'}, I, y) &\leq \con_y' \cdot \opt(I), \  \forall I \in \mI_y,\\
\label{c}
    \alg(A_{\omega'}, I, y) &\leq \rob_y' \cdot \opt(I), \  \forall I \in \mI.
\end{align}
\end{subequations}
Since $\con_y' \leq \con_y^*$, by Inequality~\eqref{a}, we further have
\begin{equation}
\label{d}
    \alg(A_{\omega'}, I, y) \leq \con_y^* \cdot \opt(I), \forall I \in \mI_y.
\end{equation}
By Inequalities~\eqref{c} and \eqref{d}, $\{\rob_y', \omega'\}$ is a feasible solution to Problem~\ref{optimization_B}, thus we have $\rob_y^* \leq \rob_y'$, where $\rob_y^*$ is the optimal objective value to Problem~\ref{optimization_B}. 

\SetLabelAlign{newline}{%
  \begin{minipage}[t]{\dimexpr\linewidth-\labelsep\relax}%
    #1\vspace*{4pt}
  \end{minipage}%
}

\vspace{4pt}

\textbf{Case II (a):} \colorbox{gray!20}{$\rob_y' > \rob_y^*$}
If $\rob_y' > \rob_y^*$, then $(\con_y', \rob_y')$ produces no Pareto improvement over $(\con_y^*, \rob_y^*)$.

\vspace{4pt}

\textbf{Case II (b):} \colorbox{gray!20}{$\rob_y' = \rob_y^*$}
 If $\rob_y' = \rob_y^*$, since $\rob^*_y \leq \overline{\rob}$, we have $\rob_y' \leq \overline{\rob}$. By Inequality~\eqref{c}, we further have
\begin{equation}
    \label{b}
    \alg(A_{\omega'}, I, y) \leq \overline{\rob} \cdot \opt(I), \ \forall I \in \mI.
\end{equation}

By Inequalities~\eqref{a} and \eqref{b}, $\{\con_y', \omega'\}$ is a feasible solution to Problem~\ref{optimization_A}, thus we have $\con_y^* \leq \con_y'$, where $\con_y^*$ is the optimal objective value to problem~\ref{optimization_A}. Because $\con_y' \leq \con_y^*$ and $\con_y^* \leq \con_y'$, we have $\con_y' = \con_y^*$. Since $\rob_y' = \rob_y^*$ and $\con_y' = \con_y^*$, $(\con_y', \rob_y')$ produces no Pareto improvement over $(\con_y^*, \rob_y^*)$.

Consequently, $(\con_y^*, \rob_y^*)$ are jointly Pareto optimal in terms of the prediction-specific consistency and robustness with respect to $y$.
Since $\rob_y^* \leq \overline{\rob}$ for all $y \in \mY$, we have $\sup_{y \in \mY} \rob_y^* \leq \overline{\rob}$. By \Cref{eq:std_cons_rob_pred} and \Cref{eq:ps_cons_rob}, Algorithm~\ref{alg:meta} is $\overline{\rob}$-robust.

Now, there exists a weakly-optimal algorithm $A_{\overline{\omega}}$ with robustness $\overline{\rob}$, we assume the consistency of $A_{\overline{\omega}}$ is $\overline{\con}$ and the prediction-specific consistency of $A_{\overline{\omega}}$ under prediction $y$ is $\overline{\con}_y$. Since $A_{\overline{\omega}}$ is $\overline{\con}_y$-consistent with respect to $y$, by \Cref{def:ps_con_rob},
\begin{equation}
    \label{eq:con_overline}
    \alg(A_{\overline{\omega}}, I, y) \leq \overline{\con}_y \cdot \opt(I), \  \forall I \in \mI_y.
\end{equation}
Since $A_{\overline{\omega}}$ is $\overline{\rob}$-robust, it is also $\overline{\rob}$-robust with respect to $y$, thus we have 
\begin{equation}
    \label{eq:rob_overline}
    \alg(A_{\overline{\omega}}, I, y) \leq \overline{\rob} \cdot \opt(I), \  \forall I \in \mI_y.
\end{equation}

By Inequalities~\eqref{eq:con_overline} and \eqref{eq:rob_overline}, $\{\overline{\con}_y, \overline{\omega}\}$ is a feasible solution to Problem~\ref{optimization_A}. Therefore, $\con^*_y \leq \overline{\con}_y$, where $\con^*_y$ is the optimal objective value to Problem~\ref{optimization_A}. Consequently,
\begin{equation*}
    \sup_{y \in \mY} \con_y^* \leq \sup_{y \in \mY} \overline{\con}_y = \overline{\con}.
\end{equation*}

By \Cref{eq:std_cons_rob_pred} and \Cref{eq:ps_cons_rob}, Algorithm~\ref{alg:meta} is $\overline{\con}$-consistent.

By weak optimality of $A_{\overline{\omega}}$, any $\overline{\rob}$-robust algorithm is at least $\overline{\con}$-consistent, and any $\overline{\con}$-consistent algorithm is at least $\overline{\rob}$-robust. Since Algorithm~\ref{alg:meta} is $\overline{\con}$-consistent and $\overline{\rob}$-robust, by \Cref{def_weak}, it is \weakly.

Moreover, $\forall \ y \in \mY$, $(\con_y^*, \rob_y^*)$ are jointly Pareto optimal in terms of the prediction-specific consistency and robustness with respect to $y$. Thus,
 Algorithm~\ref{alg:meta} is \strongly by \Cref{def_strong}.



\end{proof} 

\subsection{Addressing the Absence of a Weakly-Optimal Algorithm with $\overline{\rob}$ Robustness.} 

In general, even if $(\con', \overline{\rob})$ is not on the Pareto front for any $\con' \geq 1$, a process that first determines a tight consistency bound $\con = \sup_{y' \in \mY} \mP_1 (\overline{\rob}, y')$, then  determines a tight robustness bound $\rob = \sup_{y' \in \mY} \mP_2 (\con, y')$ can generate a tight Pareto-optimal consistency-robustness tradeoff $(\con, \rob)$ so that $\gamma$ becomes a valid input of Algorithm~\ref{alg:meta}.


Let $\con, \rob$ denote $\sup_{y' \in \mY} \mP_1 (\overline{\rob}, y')$, $\sup_{y' \in \mY} \mP_2 (\con, y')$, respectively. Define $y_1 := \arg \max_{y' \in \mY} \mP_1 (\overline{\rob}, y')$, $y_2 := \arg\max_{y' \in \mY} \mP_2 (\con, y')$. 

Since $\con = \sup_{y' \in \mY} \mP_1 (\overline{\rob}, y')$, we have $\forall \ y' \in \mY, \mP_1 (\overline{\rob}, y') \leq \con$. Therefore, 
\begin{equation*}
    \forall\  y' \in \mY, \exists\ \omega' \in \Omega, \ \text{s.t.} \ \{\overline{\rob}, \omega'\} \text{ is a feasible solution to }  \mathcal{P}_2(\con, y'),
\end{equation*}
i.e., $\forall \ y' \in \mY, \overline{\rob} \geq \mP_2(\con, y')$, where $\mP_2(\con, y')$ is the optimal objective value. Consequently,
\begin{equation}
    \label{lam_geq_rob}
    \overline{\rob} \geq \sup_{\ y' \in \mY} \mP_2(\con, y') = \rob.
\end{equation}
We can use similar techniques to prove 
\begin{equation*}
    \label{con_geq}
    \con \geq \sup_{y' \in \mY} \mP_1 (\rob, y').
\end{equation*}

Since $\con = \sup_{y' \in \mY} \mP_1 (\overline{\rob}, y')$, any $\overline{\rob}$-robust algorithm is at least $\con$-consistent. We prove this by contradiction, assuming there exists a $\overline{\rob}$-robust algorithm $\mA$ that has consistency $\con^\mA < \con$. Then, under the prediction $y_1$, $\mA$ achieves a prediction-specific consistency 
\begin{equation*}
    \con_{y_1}^\mA \leq \con^\mA < \con = \mP_1 (\overline{\rob}, y_1).
\end{equation*}
Moreover, $\mA$ is $\overline{\rob}$-robust under $y_1$, thus satisfying Constraint~\ref{cons_1a}. Therefore, $\mP_1 (\overline{\rob}, y_1)$ is not the optimal objective value, since $\con_{y_1}^\mA < \mP_1 (\overline{\rob}, y_1)$, yielding a contradiction.

Note that by Inequality~\eqref{lam_geq_rob}, $\overline{\rob} \geq \rob$; thus we can further conclude that any $\rob$-robust algorithm is at least $\con$-consistent. 

Similarly, since $\rob = \sup_{y' \in \mY} \mP_2 (\con, y')$, any $\con$-consistent algorithm is at least $\rob$-robust.
Therefore, $(\con, \rob)$ is a Pareto-optimal consistency-robustness tradeoff.




\newpage
\section{Proofs for \Cref{sec:dsr} \label{appendix:dsr}}

We provide the proofs of \Cref{thm:kd} and \ref{thm:dsr} in the following.

\subsection{Proof of \Cref{thm:kd} \label{appendix:dsr_1}}

\begin{proof}[Proof of~\Cref{thm:kd}]
    We begin by analyzing the prediction-specific consistency and robustness of $\KD$, considering two distinct cases: $y < b$ and $y \geq b$.
    
\vspace{3pt}

\noindent \textbf{Case I:} If \colorbox{gray!20}{$\y < b$}, then the algorithm buys on day $\lceil b / \lam \rceil$. 

To obtain the prediction-specific consistency, we assume that the prediction is accurate, i.e. $\x = \y$. Since the algorithm postpones the purchase beyond the predicted day, $\alg = \opt = y$, thus $\con_\y = 1$. To analyze the prediction-specific robustness, we consider incorrect predictions. Since the worst-case attack arises when $\x = \lceil b/\lam\rceil$, we have $\alg = \lceil b / \lam \rceil - 1 + b$ and $\opt = b$. Thus, $\rob_\y = \alg/\opt < (b/\lam + b)/ b = 1 + 1/\lam$.

\vspace{3pt}

\noindent \textbf{Case II:} If \colorbox{gray!20}{$\y \geq b$}, then the algorithm buys on day $\lceil \lam b\rceil$. 

To obtain the prediction-specific consistency, we assume $\x = \y$. $\alg = \lceil \lam b\rceil -1 +b$ and $\opt = b$, yielding $\con_\y = \alg / \opt < 1 + \lam $. To obtain the prediction-specific robustness, we consider incorrect predictions. Observe that the worst-case attack occurs when $\x = \lceil \lam b\rceil$, we have $\alg = \lceil \lam b\rceil - 1 + b$, and $\opt = \lceil \lam b\rceil$. Hence, $\rob_\y = \alg / \opt  < (\lam b + b)/\lceil \lam b\rceil \leq (\lam b + b)/(\lam b) =  1 + 1/\lam$.

We now prove that $\KD$ is not \strongly.
Consider a simple algorithm that always buys on day $b$. It is straightforward to verify that buy-to-rent algorithm is $1$-consistent and $(2 - 1/b)$-robust under $\y < b$. Recall that under $\y < b$, $\KD$'s prediction-specific consistency and robustness are $1$ and $\frac{\lceil b/ \lam \rceil - 1 + b}{b}$, respectively. Since $$\frac{\lceil b/ \lam \rceil - 1 + b}{b} > \frac{b-1+b}{b}= 2 - 1/b$$ for all $\lambda\in (0,1)$, by~\Cref{def_strong}, $\KD$ is not strong-optimal.
\end{proof}
\subsection{Proof of \Cref{thm:dsr} \label{appendix:dsr_2}}

\begin{proof}[Proof of~\Cref{thm:dsr}]
Denote the prediction-specific consistency and robustness of Algorithm~\ref{alg:dsr} with respect to $\y$ as $\con_\y$ and $\rob_\y$. 

We start by analyzing the prediction-specific consistency and robustness of Algorithm~\ref{alg:dsr} by considering the following three different cases.

\vspace{3pt}

\noindent\textbf{Case I:} \colorbox{gray!20}{$\y < b$}. In this case, Algorithm~\ref{alg:dsr} purchases on day $b$. It is straightforward that the algorithm is $1$-consistent and $(2-1/b)$-robust with respect to $y$.

\vspace{3pt}

\noindent\textbf{Case II:} \colorbox{gray!20}{$\y \in [b, \min\{ b(\lam + 1) - 1, (b-1)/\lam \}]$}. In this case, Algorithm~\ref{alg:dsr} purchases on day $\y + 1$. To prove the prediction-specific consistency, we assume $\x = \y$. $\alg = \y$ and $\opt = b$, yielding $\con_\y = \alg/\opt = y/b$. To prove the prediction-specific robustness, we consider inaccurate predictions. Observe that $\x = \y + 1$ brings the best attack to the algorithm, in which case $\alg = \y + b$, $\opt = b$, leading to $\rob_\y = \alg / \opt = 1 + \y / b$.

\vspace{3pt}

\noindent\textbf{Case III:} \colorbox{gray!20}{$\y > \min\{ b(\lam + 1) - 1, (b-1)/\lam \}$}. In this case, Algorithm~\ref{alg:dsr} buys on day $\lceil \lam b\rceil$. To obtain the prediction-specific consistency, we assume $\x = \y$. $\alg = \lceil \lam b \rceil -1 + b$, $\opt = b$, $\con_\y = \alg / \opt < 1 + \lambda$. To get the prediction-specific robustness, we consider incorrect predictions. Observe that the worst-case attack occurs at $x = \lceil \lam b\rceil$, in which case $\alg = \lceil \lam b \rceil -1 + b$, $\opt = \lceil \lam b\rceil$. Thus, $\rob_\y = \alg / \opt < 1 + 1/\lam$.

Now, we prove that Algorithm~\ref{alg:dsr} is \strongly.

By considering the worst-case prediction $y$ over the $\con_\y$ and $\rob_\y$ in \Cref{thm:dsr}, Algorithm~\ref{alg:dsr} is $(1+\lam)$-consistent and $(1+ 1/\lam)$-robust, where the worst-case prediction occurs when $\y > \min\{ b(\lam + 1) - 1, (b-1)/\lam \}$. Based on the lower bound by Wei and Zhang~\cite{Wei2020}, Algorithm~\ref{alg:dsr} is \weakly, as $b \to \infty$.

We consider the following three cases to prove the Pareto optimality of $(\con_y, \rob_y)$ when $b \to \infty$.

\vspace{3pt}

\noindent\textbf{Case I:} \colorbox{gray!20}{$\y < b$}. Algorithm~\ref{alg:dsr} is $1$-consistent and $(2 - 1/b)$-robust,  which is already optimal. This is because consistency can be no less than $1$, and $(2 - 1/b)$ is the best competitive ratio $\CR$ achievable in competitive analysis.

\vspace{3pt}

\noindent\textbf{Case II:} \colorbox{gray!20}{$\y \in [b, \min\{ b(\lam + 1) - 1, (b-1)/\lam \}]$}. Algorithm~\ref{alg:dsr}, which purchases on day $\y+1$, is $(\y/b)$-consistent and $(1 + \y/b)$-robust with respect to $y$.
    Since $b \leq y \leq \min\{ b(\lam + 1) - 1, (b-1)/\lam \}$, we have (i) $\lam b \geq y + 1 - b$, (ii) $\lam b \leq \frac{b^2 - b}{y}$, which together implies
    \begin{equation}
    \label{M_1, M_2}
     (b^2 - b)/y \geq y+1-b.
    \end{equation}
    Consider another algorithm $\tP$ that buys on day $p$ ($\neq \y+1$) with prediction-specific consistency and robustness $\con_\y^\tP$ and $\rob_\y^\tP$. 

\vspace{3pt}

\textbf{Case II (a):} \colorbox{gray!20}{$p< (b^2-b)/y $}.
In this case, it holds that $$\rob_\y^\tP = \frac{p - 1 + b}{p} > \frac{ [(b^2-b)/y] - 1 + b}{ [(b^2-b)/y]} = 1 + \y/b = \rob_y .$$

\vspace{3pt}

\textbf{Case II (b) :} \colorbox{gray!20}{$p = (b^2-b)/y$}  (This case applies only if $\frac{b^2 - b}{y} \in \mathbb{N}_+$).
In this case, $\rob^\tP_y =\frac{p-1+b}{p}= 1 + y/b = \rob_y$. By Inequality~\eqref{M_1, M_2}, we have the following holds: $$\con^\tP_y = \frac{p-1+b}{b} =\frac{[(b^2-b)/y]+(b-1)}{b} \geq \frac{[y+1-b]+(b-1)}{b} = y/b = \con_y .$$

\vspace{3pt}

\textbf{Case II (c) :} \colorbox{gray!20}{$(b^2-b)/y< p < \y + 1$}.
In this case, by Inequality~\eqref{M_1, M_2}, it follows that $$\con_\y^\tP = \frac{p-1+b}{b} > \frac{[(b^2-b)/y]-1+b}{b} \geq \frac{[y+1-b]-1+b}{b}= \y/b = \con_y .$$

\vspace{3pt}

\textbf{Case II (d) :} \colorbox{gray!20}{$p > \y+1$}.
In this case, we see that $$\rob_\y^\tP = \frac{p-1+b}{b} > \frac{(y+1)-1+b}{b} = 1 + y / b = \rob_y .$$

In either case, $(\con_y^\tP, \rob_y^\tP)$ provides no Pareto improvement over $(\con_y, \rob_y)$.

\vspace{3pt}

\noindent\textbf{Case III:} \colorbox{gray!20}{$y > \min\{ b(\lam + 1) - 1, (b-1)/\lam \}$}. Algorithm~\ref{alg:dsr}, purchasing on day $\lceil \lam b\rceil$, achieves a consistency of $\frac{\lceil \lam b \rceil - 1 + b}{b}$ and a robustness of $\frac{\lceil \lam b \rceil - 1 + b}{\lceil \lam b \rceil}$. Consider another algorithm $\tQ$ that buys on day $q$ ($\neq \lceil \lam b\rceil$) with prediction-specific consistency and robustness $\con_\y^\tQ$ and $\rob_\y^\tQ$. 

\vspace{3pt}

\textbf{Case III (a):} \colorbox{gray!20}{$q < \lceil \lam b \rceil$}.  
In this case,  $\rob_\y^\tQ = \frac{q-1+b}{q} > \frac{\lceil \lam b \rceil - 1 + b}{\lceil \lam b \rceil} = \rob_y$. 

\vspace{3pt}

\textbf{Case III (b):} \colorbox{gray!20}{$\lceil \lam b \rceil < q \leq y $}.  
In this case, $\con_\y^\tQ = \frac{q-1+b}{b} > \frac{\lceil \lam b \rceil - 1 + b}{b} = \con_y$.

\vspace{3pt}

\textbf{Case III (c):} \colorbox{gray!20}{$q > y$}.
    If $y > b (\lam + 1) - 1$, we have $\lam b < y+1-b$. Note that $\lim_{b\to + \infty} \frac{\lceil \lam b\rceil}{b} = \lam$. Therefore, $$\con_y = \frac{\lceil \lam b \rceil - 1 + b}{b} < \frac{(y+1-b)-1+b}{b} = \frac{y}{b} = \con_y^\tQ $$ as $b \to \infty$. If $y > (b-1)/\lam$, we have $\lceil \lam b \rceil \geq \lam b > \frac{b^2 - b}{y}$. Therefore, $$\rob_y = \frac{\lceil \lam b \rceil - 1 + b}{\lceil \lam b \rceil} < \frac{[(b^2-b)/y]-1+b}{[(b^2-b)/y]} = \frac{y+b}{b} \leq \rob_y^\tQ .$$

    In either case, $(\con_y^P, \rob_y^P)$ yields no Pareto improvement over $(\con_y, \rob_y)$. By~\Cref{def_strong}, Algorithm~\ref{alg:dsr} is \strongly as $b \to \infty$.
\end{proof}

\newpage
\section{Additional Details for \Cref{sec:rsr}}
\label{app:rsr}

In this section, we provide some supplementary details for the construction of \textit{explicit} algorithms for \emph{randomized ski rental}. The design of the \emph{explicit} strongly-optimal algorithm is largely inspired by an equalizing property, which plays a central role in the classic randomized ski rental problem~\cite{Karlin1988}.

\subsection{Equalizing Distributions}

We formally define equalizing distributions as follows.


\begin{definition}
\label{def_Equalizing_dis}
    Given integers $m$ and $n$ with $1 \leq m \leq n \leq b$, a distribution $\eq [m, n]$ is  an \textbf{equalizing distribution} in $[m, n]$ if it is only supported on $\{m, m+1, \ldots,n\}$, and $\ratio(\eq[m,n], x) = \CR(\eq[m,n])$ for all $x \in \{m, m+1,\ldots,n\}$, where $\CR(\pi)$ is the competitive ratio of distribution $\pi$.

\end{definition}
Furthermore, the following theorem gives the explicit form of the \textit{equalizing distribution}.
\begin{theorem}
\label{thm:eq_form}
Let $\eq [m,n]$ be an equalizing distribution in $[m, n]$. Then it satisfies the following recursive formula:
\begin{equation} 
\label{dis:eq}
    \eq_i[m,n] = 
\begin{cases} 
\left(1 + \frac{m+b-1}{m}\cdot\left((\frac{b}{b-1})^{n-m}-1\right)\right)^{-1} & \text{for } \ i = m; \\
\eq_m[m,n]\cdot\frac{m+b-1}{m(b-1)} \cdot\left(\frac{b}{b-1}\right)^{i-m-1} , & \text{for } \ i \in \{m+1,\ldots,n\};\\
0,  & \text{otherwise}.
\end{cases}
\end{equation}
\end{theorem}

Note that $\eq[1, b]$ is exactly the same as Karlin's distribution~\cite{Karlin1988}, while $\eq[1, \lfloor\lam b\rfloor]$ coincides with the distribution used in Kumar's randomized algorithm~\cite{Kumar2018} when $y \ge b$.

\begin{theorem}
\label{thm:eq_robust}
    For any $y < b$, the equalizing distribution on $[y+1, b]$, denoted $\eq [y+1, b]$, is $1$-consistent and minimizes robustness among all $1$-consistent distributions under $y$.
\end{theorem}

The proof of~\Cref{thm:eq_form} and \ref{thm:eq_robust} is detailed in Appendix~\ref{app:eq_form_robust}.

\subsection{\opA and \opB \label{appendix:opAB}}

We now propose two types of operations (\opA and \opB) that transform an equalizing distribution into a distribution with the desired robustness level on the Pareto frontier.

\textbf{\opA: {Consistency Boosting}.} \opA (see Algorithm~\ref{alg:op_A}) is employed when $y \geq b$. It starts with an equalizing distribution $\eq[1, n]$, where the precise choice of $n$ depends on the desired robustness level $\overline{\rob}$. The core idea of \opA is to enhance the prediction-specific consistency $\con_y$ of the distribution while not sacrificing the prediction-specific robustness $\rob_y$.



\begin{algorithm}[h]
\caption{\opA: \textsc{Consistency Boosting}}
\label{alg:op_A}
\begin{algorithmic}[1]
\State \textbf{Input:} Initial distribution $\eq[1, n]$ and prediction $y \geq b$;
\State \textbf{Initialization:} $\pi \gets \eq[1,n]$, iterative index $r \gets n$, initial robustness $\rob \gets \rob_y (\pi)$;
\State \texttt{Step 0:}
\State \quad \quad {\color{gray}  \textit{// Check whether the current distribution is a two-point distribution.}}
\State \quad \quad \textbf{If} $r=1$, \textbf{then} update $\pi_1 \gets 1/ b$ and $\pi_{b+1} \gets (b-1)/b$, go to \texttt{Step 4};
\State \quad \quad {\color{gray}  \textit{// Check if further shifting enhances consistency.}}
\State \quad \quad \textbf{If} $r \leq y+1 - b$, \textbf{then} go to \texttt{Step 4};
\State \texttt{Step 1:}
\State \quad \quad {\color{gray}  \textit{// Shifting probability mass from $\pi_r$ to $\pi_{y+1}$.}}
\State \quad \quad  Update $\pi_{y+1} \gets \pi_{y+1} + \pi_r$ and $\pi_r \gets 0$;
\State \texttt{Step 2:}
\State \quad \quad {\color{gray}  \textit{// Determine if more shifting is necessary.}}
\State \quad \quad \textbf{If} $\ratio(\pi, y+1) < \rob$, \textbf{then} set $r \gets r-1$, go to \texttt{Step 0};
\State \texttt{Step 3:} 
\State \quad \quad {\color{gray}  \textit{//Determine the maximum amount of probability mass that can be shifted.}}
\State \quad \quad Set $\pi'\gets (\pi_1,\ldots,\pi_{r-1}, \pi_r',0,\ldots,0, \pi_{y+1}')$;
\State \quad \quad Solve $\pi_{r}', \pi_{y+1}'$ through
$\smash{
\ratio(\pi', y+1) = \rob \text{ and } \
    \pi_{r}' + \pi_{y+1}' = 1 - \sum_{i=1}^{r-1} \pi_i}
$;
\State \quad \quad Update $\pi_r, \pi_{y+1} \gets \pi_r', \pi_{y+1}'$;
\State \texttt{Step 4:} \ \textbf{Return} $\pi$.
\end{algorithmic}
\end{algorithm}

\begin{remark}
Reaching $r = 1$ in \texttt{Step 0} only occurs when $y = b$; otherwise, the condition $r \leq y + 1 - b$ will be satisfied earlier. In the special case where $y = b$, any two-point distribution over $\{1, b+1\}$ is trivially $1$-consistent. In this case, \opA returns the distribution among them that achieves the lowest possible robustness.
\end{remark}

\begin{remark}
When $y \geq b$, the prediction-specific consistency is $\con_y = \frac{\sum_{i=1}^{y} \pi_i \cdot (b+i-1)}{b} + \frac{\sum_{i=y+1}^{\infty}  \pi_i \cdot y}{b}$. As a result, shifting probability mass from $\pi_r$ to $\pi_{y+1}$ can potentially improve prediction-specific consistency only in cases where $b + r - 1 > y$. This explains why, when $r\leq y+1-b$, we immediately terminate the shifting process in \emph{\texttt{Step 0}}.
\end{remark}

\textbf{\opB: {Robustness Seeking}.} \opB (see Algorithm~\ref{alg:op_B}) is employed when $y < b$. \opB (see Algorithm~\ref{alg:op_B}) initializes from the most robust $1$-consistent distribution $\eq [y+1, b]$ given the prediction $y$, and incrementally sacrifices consistency to gain robustness. This process continues until a distribution is obtained that achieves Pareto-optimality under the consistency-robustness trade-off, subject to a desired robustness level. 

\begin{algorithm}[H]
\caption{\opB: \textsc{Robustness Seeking}}
\label{alg:op_B}
\begin{algorithmic}[1]
\State \textbf{Input:} Initial distribution $\pi^{eq}[y+1, b] (y < b)$ and desired robustness $\rob$;
\State \textbf{Initialization:} $\pi \gets \pi^{eq}[y+1, b]$, iterative index $r \gets 1$;
\State \texttt{Step 1:}
\State \quad \quad  {\color{gray}  \textit{// Establish an equalized distribution over $\{1, \ldots, r, y+1, \ldots, b\}$.}}
\State \quad \quad Set $\pi'\gets(\pi_1',\ldots,\pi_r', 0, \ldots, 0, \pi_{y+1}', \ldots, \pi_b')$;
\State \quad \quad Solve $(\pi', \rob')$ through
$\smash {\
 \sum_{i=1}^{b} \pi_i' = 1 \ \text{ and } \ \ratio(\pi', x)  = \rob', \ x \in \{1, \ldots, r, y+1, \ldots,b\}}$;
\State \texttt{Step 2:}
\State \quad \quad {\color{gray}  \textit{// Check if the robustness of the newly constructed distribution meets the desired level.}}
\State \quad \quad \textbf{If} $\rob' > \rob$, \textbf{then} set $r \gets r+1$, and go to \texttt{Step 1};
\State \texttt{Step 3:}
\State \quad \quad {\color{gray}  \textit{// After determining the appropriate value of r ,}}
\State \quad \quad {\color{gray}  \textit{// Prioritize assigning probability mass to $\{1, \ldots, r-1, y+1, \ldots,b\}$.}}
\State \quad \quad Solve $\pi'$ through
$
\ \sum_{i=1}^{b} \pi_i'  = 1 \ \text{ and } \  \ratio(\pi', x)  = \rob,   \ x \in \{1, \ldots, r-1, y+1, \ldots,b\}
$;
\State \quad \quad Update $\pi \gets \pi'$;
\State \texttt{Step 4:}
\ \textbf{Return} $\pi$.
\end{algorithmic}
\end{algorithm}


\subsection{\PRSR: Prediction-specific Randomized Ski Rental}

With these foundations in place, we proceed to introduce an algorithm for randomized ski rental, referred as \PRSR (see Algorithm~\ref{alg:rsr}).

\begin{algorithm}[H]
\caption{\textsc{\textbf{\PRSR}:  \textbf{P}rediction-Specific \textbf{R}andomized \textbf{S}ki \textbf{R}ental}}
\label{alg:rsr}
\begin{algorithmic}[1]
\State \textbf{Input:} $\overline{\rob} \in [e_b / (e_b - 1), b-2)$
\State {\color{gray} \textit{// Determine the smallest $n$ such that $\eq[1,n]$ is $\overline{\rob}$-robust.}}
\State Determine $n \gets \lceil \log_{b/(b-1) } (1 + 1 / (\overline{\rob} - 1))\rceil$;
\State {\color{gray} \textit{// Determine the adjusted robustness $\rob'$, defined as that of $\eq[1, n]$.}}
\State Determine $\rob' \gets [(\frac{b}{b-1})^n-1]^{-1} + 1 $;
\State \textbf{If} $y \geq b$ \textbf{then} 
\State \quad $\pi \gets \opA(\eq[1, n], y)$; (see Algorithm~\ref{alg:op_A})
\State \textbf{Else if} $\y < b$ \textbf{then}
\State \quad {\color{gray} \textit{// Ensure that the adjusted robustness does not exceed the upper bound $\rob_\nu \coloneqq \rob(\eq[y+1, b])$.}}
\State \quad Determine $\rob_\nu \gets \rob(\eq[y+1, b])$ and determine $\rob'' \gets \min\{\rob_{\nu}, \rob' \}$;
\State \quad $\pi \gets \opB(\eq[y+1, b], \rob'')$; (see Algorithm~\ref{alg:op_B})
\State Choose $i$ randomly according to the distribution $\pi$;
\State Buy the skis at the start of day $i$.
\end{algorithmic}
\end{algorithm}

\begin{remark}
    Consider the equalizing distribution $\eq [1, n]$ with $n \leq b$. Let $\hat{\rob}$ denote the robustness of $\eq [1, n]$. Note that 
    \begin{equation}
    \label{eq:rob_eq_1}
    \hat{\rob} = \ratio(\eq [1, n], 1) = (b-1)\cdot \eq_1 [1,n] + 1.
\end{equation}
Moreover, by \Cref{thm:eq_form},
\begin{equation}
    \label{eq:eq_1}
    \eq_1 [1,n] \cdot \left( \frac{1 - (\frac{b}{b-1})^n}{1 - (\frac{b}{b-1})} \right) = 1.
\end{equation}
By \Cref{eq:rob_eq_1} and \Cref{eq:eq_1},
\begin{equation*}
    n = \log_{b/(b-1)} (1 + 1/(\hat{\rob}-1)).
\end{equation*}

Therefore, choosing $n = \lceil \log_{b/(b-1)} (1 + 1/(1-\overline{\rob}))\rceil$ ensures that $\eq [1, n]$ is $\overline{\rob}$-robust.
\end{remark}

\begin{remark}
    The requirement that $\overline{\rob} < b-2$ is primarily enforced to guarantee $n > 1$.
\end{remark}

In \Cref{app:proof_prsr}, we formally verify the prediction-specific Pareto optimality and strong optimality of $\PRSR$.






\newpage
\section{Proofs for \Cref{sec:rsr} and Appendix~\ref{app:rsr} \label{appendix:RSR}}

This section provides proofs for \Cref{sec:rsr} and Appendix~\ref{app:rsr}.

\subsection{Proofs of Useful Lemmas \label{appendix:supporting_lemma}}

\begin{lemma}
    \label{lemma:restrict}
    Consider the ski rental randomized problem with $b > 1$, for any distribution $\pi$ supported on $\mathbb{N}_+$, there exists another distribution $\pi'$ supported on a finite set $[b]$ such that $\CR(\pi') \leq \CR(\pi)$. Specifically, when the best attack for $\pi$ only occurs at $x > b$, we have $\CR(\pi') < \CR(\pi)$. Furthermore, in the learning-augmented setting with prediction $y$, for any distribution $\pi$ supported on $\mathbb{N}_+$, there exists another distribution $\pi'$ supported on the finite set $[b] \cup \{y+1\}$ such that $\con_y(\pi') \leq \con_y(\pi)$ and $\rob_y(\pi') \leq \rob_y(\pi)$.
\end{lemma}

\begin{proof}
We first consider the traditional competitive analysis setting.

Let $\pi$ be a distribution with $\sum_{i=b+1}^{\infty} \pi_i \neq 0$. Consider another distribution $\pi'$ with $\pi'_i = \pi_i$ for all $i \in [b-1]$, and $\pi'_b = \sum_{i=b}^{\infty} \pi_i$. Let $r$ be the maximal index on which $\pi$ has probability mass, i.e. $r = \max\{i \in \mathbb{N} \mid \pi_i \neq 0\}$. Let $\CR(\pi), \CR(\pi')$ denote the competitive ratio of a randomized algorithm that uses $\pi$ and $\pi'$, respectively.

\vspace{4pt}

\noindent\textbf{Case I:} \colorbox{gray!20}{$x <b$}. 
By the definition of $\ratio(\pi, x)$ in~\Cref{sec:rsr}, $\ratio(\pi, x) = \ratio(\pi', x)$.

\vspace{4pt}
    
\noindent\textbf{Case II:} \colorbox{gray!20}{$x \geq b$}. 
In this range, the worst-case attack against $\pi$ occurs at $x = r$, while the worst-case attack against $\pi’$ occurs at $x = b$. According to the definition of $\ratio(\pi, x)$, we have $\ratio(\pi, r) > \ratio(\pi’, b)$.
    
Note that $\CR(\pi) = \sup_{x \in \mathbb{N}_+} \ratio(\pi, x)$ and $\CR(\pi') = \sup_{x \in \mathbb{N}_+} \ratio(\pi', x)$. We have $\CR(\pi) \geq \CR(\pi')$. Specifically, when the best attack for $\pi$ occurs at $x > b$, we have $\CR(\pi) > \CR (\pi')$.

In the learning-augmented setting, we divide our discussion into two parts: $y < b$ and $y \geq b$.

\vspace{4pt}

\noindent\textbf{Case I:} \colorbox{gray!20}{$y < b$}. 
In this case, $[b] \cup \{y+1\} = [b]$. Let $\pi'$ be the distribution after transferring all probability mass beyond $b$ of $\pi$ to $b$ (i.e., $\pi'_i = \pi_i, \forall\  i \in [b-1]$ and $ \pi'_b = \sum_{i=b}^{\infty}\pi_i$). By \Cref{lemma:restrict}, $\rob_y (\pi') \leq \rob_y (\pi)$. By \Cref{ps_con_rob_rsr}, we have $\con_y (\pi') = \con_y (\pi)$.

\vspace{4pt}

\noindent\textbf{Case II:} \colorbox{gray!20}{$y \geq b$}. 
Construct another distribution $\pi'$ by transferring all the probability mass of $\pi$ on the $\{b+1, \ldots, y\}$ to $b$, and all the mass beyond $y+1$ to $y+1$ (i.e., $\pi_i' = \pi_i, \forall\  i \in [b-1]; \pi_b' = \sum_{i=b}^{y} \pi_i$ and $\pi_{y+1}' = \sum_{i=y+1}^{\infty} \pi_i$). By a simple verification, we obtain that $\con_y (\pi') \leq \con_y(\pi)$ and $\rob (\pi') \leq \rob(\pi)$.

\end{proof}

\begin{lemma}
\label{lemma:eq_ratio}
    Consider a randomized ski rental problem with $b > 1$. Consider a probability distribution $\pi$ over $\mathbb{N}_+$. If $\ratio(\pi, x) = \CR(\pi)$ for all $ x \in \{m, m+1,\ldots,n \}$, then we have 
    \begin{equation*}
        \pi_{i+1} = \frac{b}{b-1}\cdot \pi_i , \quad \forall \ i \in \{m+1,\ldots,n-1\}.
    \end{equation*}
\end{lemma}
\vspace{-8pt}
\begin{proof}
By the definition of $\ratio (\prob, \x)$ in~\Cref{sec:rsr}, when $x \leq b$:
\begin{equation*}
    \ratio (\prob, \x) = \frac{\sum_{i=1}^\x \prob_i (\x-1+b) + \sum_{i=\x+1}^{\infty} \prob_i\cdot \x}{x}.
\end{equation*} 
Consider the following equations:
\begin{align}
\label{ratio_a}
    \ratio (\pi, i) &= \CR (\pi), \quad m \leq i \leq m+2
\end{align}
Subtracting both sides of~$(m+1)\times$\Cref{ratio_a} with $i = m+1$ by~$m\times$\Cref{ratio_a} with $i = m$,
\begin{equation}
    \label{ratio_1}
    b \cdot\pi_{m+1} + \sum_{i=m+2}^{\infty} \pi_i = \CR (\pi).
\end{equation}
Subtracting both sides of~$(m+2)\times$ \Cref{ratio_a} with $i = m+2$ by~$(m+1)\times$ \Cref{ratio_a} with $i = m+1$,
\begin{equation}
    \label{ratio_2}
    b \cdot\pi_{m+2} + \sum_{i=m+3}^{\infty} \pi_i = \CR (\pi).
\end{equation}
Finally, subtracting both sides of \Cref{ratio_2} by \Cref{ratio_1} yields
$
    \pi_{m+2} = \frac{b}{b-1}\pi_{m+1}.
$
Similarly, we can show that 
\begin{equation*}
        \pi_{i+1} = \frac{b}{b-1} \pi_i , \quad \forall \ i \in \{m+2,\ldots,n-1\}.
\end{equation*}
\end{proof}

\subsection{Proof of \Cref{thm:eq_form} and \ref{thm:eq_robust} \label{app:eq_form_robust}}
\begin{proof}[Proof of~\Cref{thm:eq_form}]
By \Cref{def_Equalizing_dis}, $\eq[m,n]$ only has positive support over $\{m, \ldots,n\}$. 

Consider $x = m$ and $x = m+1$,
\begin{equation}
    \label{eq_1}
    \frac{(m-1+b)\cdot\eq_{m}[m,n]}{m}+  (1-\eq_m[m,n]) =\CR(\pi).
\end{equation}
\begin{equation}
    \label{eq_2}
    \frac{(m-1+b)\eq_m[m,n] + (m+b)\eq_{m+1}[m,n]]}{(m+1)} +(1-\eq_{m}[m,n]-\eq_{m+1}[m,n]) = \CR(\pi).
\end{equation}
Equating the left-hand sides of \Cref{eq_1} and \Cref{eq_2} yields:
\begin{equation}
\label{ratio_m}
    \eq_{m+1}[m,n] = \frac{m(b-1)}{m+b-1} \eq_m[m,n].
\end{equation}
By~\Cref{lemma:eq_ratio}, we have
\begin{equation}
\label{ratio_m+}
        \eq_{i+1}[m,n] = \frac{b}{b-1}\cdot \eq_i[m,n] , \quad \forall \ i \in \{m+1,\ldots,n-1\}.
\end{equation}
\Cref{ratio_m}, \Cref{ratio_m+}, together with the constraint $\sum_{i=m}^n \pi_i = 1$, imply that
\begin{equation*} 
    \eq_i[m,n] = 
\begin{cases} 
\left(1 + \frac{m+(b-1)}{m}\cdot\left((\frac{b}{b-1})^{n-m}-1\right)\right)^{-1} & \text{for } \  i = m; \\
\eq_m[m,n]\cdot\frac{m+(b-1)}{m(b-1)} \cdot\left(\frac{b}{b-1}\right)^{i-m-1} , & \text{for } \  i \in \{m+1,\ldots,n\};\\
0,  & \text{otherwise}.
\end{cases}
\end{equation*}
\end{proof}

\begin{proof}[Proof of~\Cref{thm:eq_robust}]
    Applying \Cref{lemma:restrict}, we first reduce the support under consideration to $[b]$. Furthermore, the condition of $1$ consistency requires $\pi_i = 0$ for all $i \in [y]$. We consider the following optimization problem.

\begin{align*}
    \label{Primal}
    \tag{Primal Problem}
    \min_{\pi_{y+1}, \ldots, \pi_b, \rob_y} \quad & \rob_{y}   \\
    \text{s.t.} \quad \quad \quad 
    & \sum_{j=y+1}^{i} \pi_j \cdot (b + j - 1) +  \sum_{j=i+1}^{b} \pi_j  \cdot i \leq \rob_{y} \cdot i, \quad \forall \ i \in \{y+1, \dots, b\}, \\
    & \sum_{i=y+1}^{b} \pi_i = 1, \\
    & \pi_i \geq 0, \quad \forall i \in \{y+1,\dots, b\}.\\
    & \rob_{y} \geq 0
\end{align*}

Let $\pi_{y+1}^*,\ldots, \pi_b^*,\rob_y^*$ denote the optimal solution to \ref{Primal}. It's clear that 
\begin{equation}
    \label{primal_rob}
    \rob_y^* = \min_{\pi_{y+1}^*,\ldots, \pi_b^*}\max_{i\in\{ y+1,\ldots,b \}} \ratio(\pi^*, i).
\end{equation}

We claim that $\pi_{y+1}^* \neq 0$. We prove this by contradiction, assuming $\pi_{y+1}^* = 0$. Let $r$ be the minimal index with non-zero probability mass, i.e. $r \coloneqq \min \{i \mid \pi_i^* \neq 0\}$. Consider another distribution $\pi'$ with $\pi'_{y+1} = \ep$, $\pi'_r = \pi_r^*-\ep$ and $\pi'_i = \pi_i^*$ for all $i \in \mathbb{Z}_+\setminus \{y+1, r\}$, where we denote $$\ep \coloneq \min \left\{\pi_r^*, \frac{y+1}{b-1}\cdot (\frac{\rob^*_y -1}{2})\right\}.$$ 

Note that 
\begin{equation*}
    \ratio(\pi',y+1) = \frac{(b+y)\cdot \pi_{y+1}' + (y+1)\cdot (1-\pi_{y+1})}{y+1} \leq \frac{\rob^*_y -1}{2} + 1 = \frac{\rob^*_y + 1}{2} < \rob^*_y.
\end{equation*}
Moreover, it follows that
\begin{equation*}
\ratio(\pi',i) < \ratio(\pi^*, i), \quad \forall i \ \in \{y+2,\ldots, b\}.
\end{equation*}
Therefore, 
\begin{equation*}
\max_{i\in\{ y+1,\ldots,b \}} \ratio(\pi',i) < \max_{i\in\{ y+1,\ldots,b \}} \ratio(\pi^*, i) = \rob_y^*,
\end{equation*}
which contradicts with \Cref{primal_rob}.

We then claim that $\pi_b^* \neq 0$. We prove this by contradiction, assuming $\pi_b^* = 0$. Consider another distribution $\pi'$ with $\pi_{y+1}' = \pi_{y+1}^* - \ep$, $\pi_b' = \ep$, and $\pi'_i = \pi_i^*$ for all $i \in \mathbb{Z}_+\setminus \{y+1, r\}$, where $\ep = \min \{ \pi^*_{y+1}, {\rob_y^*}/({2b-1}) \}$.
We verify that 
\begin{equation*}
\ratio(\pi',i) < \ratio(\pi^*, i), \quad \forall i \ \in \{y+1,\ldots, b-1\}.
\end{equation*}
Furthermore,
\begin{equation*}
    \ratio(\pi',b) < \frac{\sum_{i=y+1}^{b-1} (i-1+b)\cdot \pi_i^* + (2b-1)\cdot \ep}{b}  \leq \frac{(b-1)\cdot \rob_y^* + (2b-1)\cdot \ep}{b} \leq \rob^*_y.
\end{equation*}
Consequently,
\begin{equation*}
\max_{i\in\{ y+1,\ldots,b \}} \ratio(\pi',i) < \max_{i\in\{ y+1,\ldots,b \}} \ratio(\pi^*, i) = \rob_y^*,
\end{equation*}
which contradicts with \Cref{primal_rob}.

We further claim that $\pi_i^* \neq 0$, for all  $\ i \in \{y+2,\ldots, b-1\}$. We prove this by contradiction, assuming that there exists some $\ q \in \{y+2,\ldots, b-1\}$, such that $\pi_q^* = 0$. Let $r = \min \{i >q \mid \pi_i \neq 0\}$. Since $\pi_b^* \neq 0$, such $r$ is guaranteed to exist. Consider another distribution $\pi'$ with $\pi'_{y+1} = \pi^*_{y+1} - \ep_1$, $\pi'_q =\ep_1+\ep_2$, $\pi_r = \pi^*_{r} - \ep_2$, and $\pi'_i = \pi_i^*$ for all $i \in \mathbb{Z}_+\setminus \{y+1, q, r\}$, where we denote 
\begin{align*}
    \ep_1 & \coloneq \min \left\{\pi_{y+1}^*, \frac{r-q}{2(q-y-1)}\cdot \ep_2\right\},\\
    \ep_2 & \coloneq \min  \left\{ \pi^*_{r}, \frac{2(q-y-1)}{[(r-q)+2(q-y-1)](q-1+b)} \cdot \rob_y^*\right\}.
\end{align*}

Similarly, we verify that
\begin{equation*}
\ratio(\pi',i) < \ratio(\pi^*, i), \quad \forall i \ \in \{y+1,\ldots, q-1\}.
\end{equation*}
Note that $\ep_1 +\ep_2 \leq \frac{\rob_y^*}{q-1+b}$. We have
\begin{equation*}
    \ratio(\pi',q) < \frac{(q-1)\cdot \rob_y^* + (q-1+b)\cdot (\ep_1+\ep_2)}{q} \leq \rob^*_y.
\end{equation*}
Since $(q-y-1)\ep_1 < 2(q-y-1)\ep_1 \leq (r-q)\ep_2$, we have
\begin{equation*}
\ratio(\pi',i) < \ratio(\pi^*, i), \quad \forall i \ \in \{q+1,\ldots, b\}.
\end{equation*}
Therefore,
\begin{equation*}
\max_{i\in\{ y+1,\ldots,b \}} \ratio(\pi',i) < \max_{i\in\{ y+1,\ldots,b \}} \ratio(\pi^*, i) = \rob_y^*,
\end{equation*}
which contradicts with \Cref{primal_rob}.

As a result, we conclude that $\pi_i^* \neq 0, \forall \ i \in \{y+1, \ldots, b\}$.

Now, let the dual variables be $\lam_{y+1},\ldots, \lam_b, \lam$. The dual problem can be formulated as following:

\begin{align*}
    \label{Dual}
    \tag{Dual Problem}
    \max_{\lam_{y+1}, \ldots, \lam_b, \lam} \quad & \lam   \\
    \text{s.t.} \quad \quad \quad 
    & \sum_{i=y+1}^{x-1} i \cdot \lam_i + (b+x-1)\cdot\sum_{i = x}^b \lam_i + \lam \leq 0, \quad \forall\ x  \in \{y+1 , \ldots, b\}, \\
    & \sum_{i=y+1}^{b} - i \cdot \lam_i \leq 1, \\
    & \lam_i \leq 0, \quad \forall i \in \{y+1,\dots, b\},\\[2pt]
    & \lam \text{ is free.}
\end{align*}

Let $\lam_{y+1}^*,\ldots, \lam_b^*, \lam^*$ denote the optimal solution to \ref{Dual}.

By the \textit{complementary slackness}, 
\begin{equation*}
    \pi_x^* \cdot \left[\sum_{i=y+1}^{x-1} i \cdot \lam_i^* + (b+x-1)\cdot\sum_{i = x}^b \lam_i^* + \lam^*\right] = 0, \quad \forall\ x  \in \{y+1 , \ldots, b\};
\end{equation*}
\begin{equation*}
    \rob_y^* \cdot \left( 1 + \sum_{i=y+1}^b i\cdot \lam_i^* \right) =0.
\end{equation*}

Since $\pi_i^* \neq 0, \forall \ i \in \{y+1, \ldots, b\}$, we have
\begin{equation*}
    \sum_{i=y+1}^{x-1} i \cdot \lam_i^* + (b+x-1)\cdot\sum_{i = x}^b \lam_i^* + \lam^* = 0, \quad \forall\ x  \in \{y+1 , \ldots, b\}.
\end{equation*}
Because $\rob_y^* \neq 0$, we have
$
    \sum_{i=y+1}^b -i\cdot \lam_i^* = 1.
$

Let $\Lambda_i \coloneqq  - i \cdot \lam_i$ for all $i \in \{y+1,\ldots, b\}$. Then, the dual problem can be transformed into:
\begin{align*}
    \max_{\Lambda_{y+1}, \ldots, \Lambda_b, \lam} \quad & \lam   \\
    \text{s.t.} \quad \quad \quad 
    & \sum_{i=y+1}^{x-1} \Lambda_i + \sum_{i = x}^b \frac{(b+x-1)}{i} \cdot \Lambda_i = \lam, \quad \forall x \ \in \{y+1 , \ldots, b\}, \\
    & \sum_{i=y+1}^{b} \Lambda_i = 1, \\
    & \Lambda_i \geq 0, \quad \forall i \in \{y+1,\dots, b\},\\[2pt]
    & \lam \text{ is free.}
\end{align*}

Let $\Lambda^*_{y+1},\ldots, \Lambda_b^*,\lam^*$ denote the optimization problem.

We claim that $\Lambda_{y+1}^* \neq 0$. Consider $x = y+1$ and $x = y+2$, the optimal solution should satisfy that
\[
\begin{cases}
    \frac{b+y}{y+1}\cdot \Lambda_{y+1}^* + \ \ \ \frac{b+y}{y+2}\cdot \Lambda_{y+2}^* + \cdots + \quad \frac{b+y}{b}\cdot \Lambda_b^*  = \lambda^* \\
    \ \ \ \ \ \ \ \ \  \Lambda_{y+1}^* + \frac{b+y+1}{y+2}\cdot \Lambda_{y+2}^* +  \cdots + \frac{b+y+1}{b}\cdot \Lambda_b^*  = \lambda^* \\
\end{cases}
\]
If $\Lambda_{y+1}^* = 0 $, since $$\frac{b+y}{y+2} <\frac{b+y+1}{y+2} , \ldots, \frac{b+y}{b} <\frac{b+y+1}{b},$$ then we have $$\Lambda_{y+2}^* = \cdots= \Lambda_b^* = 0,$$ which contradicts with $\sum_{i=y+1}^b \Lambda_i^* = 1$.
We then claim that $\Lambda_i^* \neq 0$, for all $ i \in \{y+2, \ldots, b-1\}$. Consider $x = i$ and $x = i+1$, the optimal solution should satisfy that
\[
\begin{cases}
    \Lambda_{y+1}^* + \cdots + \Lambda^*_{i-1}+ \frac{b+i-1}{i} \cdot \Lambda_i^*+\frac{b+i-1}{i+1}\cdot \Lambda_{i+1}^* + \cdots + \frac{b+i-1}{b}\cdot \Lambda_b^*  = \lambda^* \\
    \Lambda_{y+1}^* + \cdots + \Lambda^*_{i-1}+ \quad \quad \ \ \ \ \ \Lambda_i^*+ \ \ \
    \frac{b+i}{i+1}\cdot \Lambda_{i+1}^* +  \cdots +\ \ \ \ \frac{b+i}{b}\cdot \Lambda_b^*  = \lambda^* \\
\end{cases}
\]
Assume $\Lambda_i^* = 0$. Since $$\frac{b+i-1}{i+1} < \frac{b+i}{i+1}, \ldots, \frac{b+i-1}{b} < \frac{b+i}{b},$$ we have $$\Lambda_{i+1}^* = \Lambda_{i+2}^* = \cdots = \Lambda_b^* = 0.$$

Similarly, consider $x = i-1$ and $x= i$, we have 
\[
\begin{cases}
    \Lambda_{y+1}^* + \cdots + \Lambda^*_{i-2}+ \frac{b+i-2}{i-1} \cdot \Lambda_{i-1}^*+\frac{b+i-2}{i}\cdot \Lambda_{i}^* + \cdots + \frac{b+i-2}{b}\cdot \Lambda_b^*  = \lambda^* \\
    \Lambda_{y+1}^* + \cdots + \Lambda^*_{i-2}+ \quad \quad \ \ \, \ \ \Lambda_{i-1}^*+
    \frac{b+i-1}{i}\cdot \Lambda_{i}^* +  \cdots +\frac{b+i-1}{b}\cdot \Lambda_b^*  = \lambda^* \\
\end{cases}
\]
We further have $\Lambda_{i-1}^* = 0$.

Repeating this process, we obtain $\Lambda_{y+1}^* = 0$, which leads to a contradiction.

We finally claim that $\Lambda^*_b \neq 0$. Consider $x = b-1$ and $x = b$,
\[
\begin{cases}
    \Lambda_{y+1}^* + \cdots + \Lambda_{b-2}^*+\frac{2b-2}{b-1}\cdot \Lambda^*_{b-1}+ \frac{2b-2}{b}\cdot \Lambda_b^*  = \lambda^* \\
    \Lambda_{y+1}^* + \cdots +\Lambda_{b-2}^*+  \ \ \ \ \ \ \ \ \ \ \  \Lambda^*_{b-1}+ \frac{2b-1}{b}\cdot \Lambda_b^*  = \lambda^* \\
\end{cases}
\]
If $\Lambda^*_b = 0$, we can deduce that $\Lambda^*_{b-1} = 0$, which contradicts our previous results.

Therefore, we have
\begin{equation*}
    \Lambda_i^* \neq 0, \quad \forall \ i \in \{y+1,\ldots,b\},
\end{equation*}
which is equivalent to 
\begin{equation*}
    \lambda_i^* \neq 0, \quad \forall \ i \in \{y+1,\ldots,b\}.
\end{equation*}

Furthermore, by \textit{complementary slackness}, 
\begin{equation*}
    \lambda_i^*\cdot \left[\sum_{j=y+1}^{i} \pi_j \cdot (b + j - 1) +  \sum_{j=i+1}^{b} \pi_j  \cdot i - \rob_{y} \cdot i\right] =  0, \quad \forall \ i \in \{y+1,\ldots,b\}
\end{equation*}
It follows that
\begin{equation*}
    \sum_{j=y+1}^{i} \pi_j \cdot (b + j - 1) +  \sum_{j=i+1}^{b} \pi_j  \cdot i = \rob_{y} \cdot i, \quad \forall \ i \in \{y+1,\ldots,b\}
\end{equation*}

By \Cref{def_Equalizing_dis}, we conclude that $\eq[y+1, b]$ is the most robust $1$-consistent distribution.
\end{proof}

\subsection{Proof of~\Cref{thm:KR} \label{appendix:kr}}

\begin{proof}[Proof of~\Cref{thm:KR}]
    Consider $y < b$. Kumar's algorithm uses distribution in~\eqref{eq:KR-distribution}, where $m = \lceil b / \lam \rceil > b$. We denote this distribution as $\pi$, and consider another distribution $\pi'$ that transfer all probability masses beyond $b$ to $b$ (i.e. $\pi_i' = \pi_i, \forall i \in [b-1]$ and $\pi_b' = \sum_{i=b}^{\infty} \pi_i$). Recall that $\con_\y (\pi), \rob_\y (\pi), \con_\y (\pi'), \rob_\y (\pi')$ are the consistency and robustness of $\pi$ and $\pi'$ under prediction $\y$, respectively.

    Since $y < b$, by \Cref{ps_con_rob_rsr}, we have $\con_\y (\pi') = \con_\y (\pi)$. Observe that the worst-case attack against $\pi$ only occurs at $x > b$, by \Cref{lemma:restrict}, we have $\rob_y (\pi') < \rob_y (\pi)$. By \Cref{def_strong}, $\KR$ is not \strongly.
\end{proof}


\subsection{Proof of \Cref{thm:orsr}}
\label{app:orsr}

\begin{proof}[Proof of \Cref{thm:orsr}]
We start by finding a \weakly algorithm with robustness $\overline{\rob}$ for any given $\overline{\rob} \in \Lambda_{\rob}$.

If $\overline{\rob} = e_b/(e_b-1)$, Karlin's algorithm \cite{Karlin1988} has exactly $\overline{\rob}$ robustness. Since it is the only $e_b/(e_b - 1)$-competitive algorithm, it is weakly-optimal.

As $b \to \infty$, for any $\overline{\rob} \in (e_b/(e_b-1), \infty)$, there exists $\lambda \in (1/b, 1)$, such that \KR's robustness is exactly $\overline{\rob}$. Furthermore, \KR is \weakly as $b \to \infty$. 

Therefore, we find a \weakly algorithm with robustness $\overline{\rob}$ for any $\overline{\rob} \in \Lambda_\rob$.

We now conclude with the strong optimality of Algorithm~\ref{alg:orsr}: by \Cref{lemma:restrict}, the bi-level optimization problems can be reduced to 
Problem~\ref{optimization_1} and Problem~\ref{optimization_2}. By immediate consequence of \Cref{prop:oba}, Algorithm~\ref{alg:orsr} is \strongly.
\end{proof}

\subsection{Proof of \Cref{thm:rsr} \label{app:proof_prsr}}

Before proving \Cref{thm:rsr}, we first prove the Pareto optimality of \PRSR (see Algorithm~\ref{alg:rsr})'s prediction-specific consistency and robustness. Our analysis is divided into two cases: $y < b$ and $y \geq b$.

For the case where $y < b$, we begin by considering the following optimization problem, referred to as~\ref{Prob A}, which plays a key role in the subsequent analysis.

\begin{align*}
    \label{Prob A}
    \tag{Problem A}
    \min_{\pi_1, \pi_2, \ldots, \pi_b} \quad & \frac{\sum_{i=1}^{y}\pi_i \cdot (b+i-1) + \sum_{i=y+1}^{b}\pi_i \cdot y}{y}\\
    \text{s.t.} \quad \quad 
    & \frac{\sum_{i=1}^{x} \pi_i \cdot (b + i - 1) + \sum_{i=x+1}^{b} \pi_i  \cdot x}{x} \leq \rob_y, \quad \forall \ x \in \{1, 2, \ldots, b\}, \\
    & \sum_{i=1}^{b} \pi_i = 1, \\
    & \pi_i \geq 0, \quad \forall i \in \{1, 2, \ldots, b\}.
\end{align*}
Let $\pi^A = (\pi_1^A, \pi_2^A,\ldots,\pi_b^A)$ denote the optimal solution to \ref{Prob A}.

\begin{lemma}
\label{lemma:pi_i_neq_0}
Assume that $\rob_y \in [\rob_\xi, \rob_\nu]$. For \ref{Prob A}, $\pi_i^A \neq 0,$ for all $ i \in \{y+1, y+2, \ldots,b\}$. 
\end{lemma} 

\begin{proof}
For any $\rob_y \in [\rob_\xi, \rob_\nu]$, $\eq[1, b]$ is always a feasible solution to \ref{Prob A}. This guarantees the existence of a optimal solution $\pi^A$.

We prove by contradiction. Suppose there exists $r \in \{y+1, y+2, \ldots, b\}$ such that $\pi_r^A = 0$. We analyze the following cases. 

\vspace{5pt}

\noindent\textbf{Case I:} \colorbox{gray!20}{$\sum_{i=1}^{y} \pi_i^A = 0$}.
In this case, $\pi_1^A =\cdots=\pi_y^A = \pi_r^A = 0$. When $\rob_y = \rob_\nu$, $\eq[y+1, r]$ is the only solution to \ref{Prob A} if we requires $\pi_1^A =\cdots=\pi_y^A = 0$. If we further require $\pi_r^A = 0$ for some $r \in \{y+1, y+2, \ldots, b\}$, there is no feasible solution to \ref{Prob A} for any $\rob_y \in [\rob_\xi ,\rob_\nu]$.

\vspace{5pt}

\textbf{Case II(a)}: \colorbox{gray!20}{$\sum_{i=1}^{y} \pi_i^A \neq 0$ and $\sum_{i=r+1}^{b} \pi_i^A \neq 0$}.
Let $p$ be the largest index in $[y]$ that has non-zero probability mass, i.e. $p \coloneqq \max\{i \in [y] \mid \pi_i^A \neq 0\}$. Let $q$ be the smallest index in $[b]\setminus[r]$ that has non-zero probability mass, i.e. $$q \coloneqq \min\left\{i \in [b]\setminus[r] \mid \pi_i^A \neq 0\right\}.$$ Since $\sum_{i=1}^{y} \pi_i^A \neq 0$ and $\sum_{i=r+1}^{b} \pi_i^A \neq 0$, such $p$ and $q$ are guaranteed to exist. Based on this, we can always construct another solution
\begin{equation*}
    \pi' = \left(\pi_1^A, \pi_2^A, \ldots, \pi_p^A-\eps_1, \ldots, \pi_{r-1}, \eps_1 + \eps_2, \ldots, \pi_q^A - \eps_2, \ldots, \pi_b^A \right),
\end{equation*}
where $\eps_2 = \min\left\{  \frac{(\sum_{i=1}^r \pi_i^A) b}{(q-r+b-1)(r-1)}, \pi_q^A  \right\} > 0$, $\eps_1 = \min\left\{ (\frac{q-r}{r-p})\eps_2, \pi_p^A \right\} > 0$. 

We first prove that $\pi'$ is still a feasible solution to \ref{Prob A}.

Recall that when $x \leq b$ (see the definition of $\ratio (\pi, x)$ in~\Cref{sec:rsr})
\begin{equation*}
    \ratio (\pi, x) = \frac{\sum_{i=1}^x \pi_i (i-1+b) + \left(\sum_{i=x+1}^b \pi_i\right) x}{x}.
\end{equation*}

\begin{enumerate}[leftmargin=8mm,
 labelsep=3pt, itemindent=5pt,
                  labelwidth= 3em]

\item[\textbf{(1)}]

 Consider $x < p$. $\ratio (\pi', x) = \ratio(\pi^A, x) \leq \rob_y$.

 \item[\textbf{(2)}] 

 Consider $p \leq x < r$. Since $\pi_i'= \pi_i^A$ for all $i \in [x]\setminus\{p\}$, $\pi_p' = \pi_p^A -\eps_1$, and $\left(\sum_{i=x+1}^b \pi_i'\right) = \left(\sum_{i=x+1}^b \pi_i^A\right) + \eps_1$, we have \begin{equation*}
    x\ratio(\pi', x) = x\ratio(\pi^A, x) - \eps_1(p-1+b) + \eps_1 x < x\ratio(\pi^A, x).
\end{equation*}
Therefore, we have 
\begin{equation*}
    \ratio(\pi', x) < \ratio(\pi^A, x) \leq \rob_y.
\end{equation*}

\item[\textbf{(3)}] Consider $x = r$. Note that 
\begin{equation}
\label{2a_1}
    r\ratio(\pi', r) = \sum_{i=1}^r \pi_i' (b+i-1) + \left(1 - \sum_{i=1}^r \pi_i'\right) r,
\end{equation}
\begin{equation}
\label{2a_2}
    r\ratio(\pi^A, r) = \sum_{i=1}^r \pi_i^A (b+i-1) + \left(1 - \sum_{i=1}^r \pi_i^A\right) r.
\end{equation}
Since $\pi_i'= \pi_i^A$ for all $i \in [r]\setminus\{p, r\}$, $\pi_p' = \pi_p^A - \eps_1$, $\pi_r' = \pi_r^A + (\eps_1 + \eps_2)$, and $(\sum_{i=r+1}^{b} \pi_i') = (\sum_{i=r+1}^b \pi_i^A) - \eps_2$, we have 
\begin{equation*}
    r\ratio(\pi', r) = r\ratio(\pi^A, r) - \eps_1(b+p-1)+(\eps_1 + \eps_2) (b+r-1) - \eps_2 r.
\end{equation*}
By arranging terms, we have
\begin{equation}
    \label{2a_3}
    r\ratio(\pi', r) = r\ratio(\pi^A, r) + \eps_1(r - p) + \eps_2 (b-1).
\end{equation}
Note that $r-p > 0$, $b-1 > 0$, and $\eps_1 \leq (\frac{q-r}{r-p})\eps_2$, $\eps_2 \leq \frac{(\sum_{i=1}^r \pi_i^A)b}{(q-r+b-1)(r-1)}$. We have
\begin{equation}
    \label{2a_4}
    \eps_1(r - p) + \eps_2 (b-1) \leq \frac{(\sum_{i=1}^r \pi_i^A)b}{r-1} \leq \frac{\sum_{i=1}^r \pi_i^A(b+i-1)}{r-1}.
\end{equation}
By \eqref{2a_2}, \eqref{2a_3} and \eqref{2a_4}, 
\begin{equation}
    \label{2a_5}
    r\ratio(\pi', r) \leq \frac{r\sum_{i=1}^r \pi_i^A (b+i-1) + r (r-1)(1 - \sum_{i=1}^r \pi_i^A)}{r-1}.
\end{equation}
Since $\pi_r^A = 0$, \eqref{2a_5} implies 
\begin{equation*}
    r\ratio(\pi', r) \leq \frac{r\sum_{i=1}^{r-1} \pi_i^A (b+i-1) + r (r-1)(1 - \sum_{i=1}^{r-1} \pi_i^A)}{r-1}.
\end{equation*}
Therefore,
\begin{equation}
    \label{2a_6}
    \ratio(\pi', r) \leq \frac{\sum_{i=1}^{r-1} \pi_i^A (b+i-1) +  (1 - \sum_{i=1}^{r-1} \pi_i^A)(r-1)}{r-1} = \ratio(\pi^A, r-1) \leq \rob_y.
\end{equation}

\item[\textbf{(4)}] Consider $r < x < q$. Note that $q = \min\{i \in [b]\setminus[r] \mid \pi_i^A \neq 0\}$. Thus $\pi_{r+1}^A = \cdots = \pi_{q-1}^A = 0$. Hence, $\pi_{r+1}' = \cdots = \pi_{q-1}' = 0$. It follows that 
\begin{equation}
    \label{2a_7}
    x \ratio(\pi', x) = r \ratio (\pi', r).
\end{equation}

By \eqref{2a_6} and \eqref{2a_7}, we have
\begin{equation*}
    \ratio(\pi', x) = \frac{r}{x}\cdot\ratio(\pi', r) < \ratio(\pi', r) \leq \rob_y.
\end{equation*}

\item[\textbf{(5)}] Consider $q \leq x \leq b$. Since $\pi_i'= \pi_i^A$ for all $i \in [b]\setminus\{p, r, q\}$, $\pi_p' = \pi_p^A - \eps_1$, $\pi_r' = \pi_r^A + (\eps_1 + \eps_2)$, and $\pi_q' = \pi_q^A - \eps_2$, we have 
\begin{equation*}
    x \ratio(\pi', x) = x\ratio(\pi^A, x) - \eps_1 (b+p-1) + (\eps_1 + \eps_2) (b+r-1) - \eps_2 (b+q-1).
\end{equation*}
By arrange terms, we have
\begin{equation}
    \label{2a_8}
    x \ratio(\pi', x) = x\ratio(\pi^A, x) +\eps_1 (r-p) - \eps_2 (q-r). 
\end{equation}
Note that $\eps_1 \leq (\frac{q-r}{r-p})\eps_2$, and $r-p > 0$, $q-r > 0$. \eqref{2a_8} further indicates that
\begin{equation*}
    x \ratio(\pi', x) \leq x\ratio(\pi^A, x).
\end{equation*}
Therefore, 
\begin{equation*}
    \ratio(\pi', x) \leq \ratio(\pi^A, x) \leq \rob_y.
\end{equation*}
\end{enumerate}

Consequently, $\ratio(\pi',x) \leq \rob_y$ for all $x \in [b]$. Note that $\sum_{i=1}^b \pi_i' = 1$ and $\pi_i' \geq 0 $ for all $i \in [b]$. Therefore, $\pi'$ is a feasible solution to \ref{Prob A}.

We then compare the objective value of \ref{Prob A} at $\pi'$ and $\pi^A$.

Note that $p \leq y < r$. Consequently, as established in \textit{Case II(a)(2)},
\begin{equation*}
    y\ratio(\pi', y) = y \ratio(\pi^A, y) - \eps_1 (b+p-1) + \eps_1 y < y \ratio(\pi^A, y).
\end{equation*}
This leads to
\begin{equation*}
    \ratio(\pi', y) <  \ratio(\pi^A, y).
\end{equation*}

This makes $\pi^A$ impossible to be the optimal solution to \ref{Prob A} under the assumptions of \textit{Case II (a)}, since there is always another feasible solution $\pi'$ that achieves a smaller objective value.

\vspace{5pt}

\textbf{Case II(b)}: \colorbox{gray!20}{$\sum_{i=1}^y \pi_i^A \neq 0 $ and $\sum_{i=r+1}^b \pi_i^A = 0$}.

Let $p$ be the largest index in $[y]$ that has non-zero probability mass, i.e. $p \coloneqq \max \{i \in [y]\mid \pi_i^A \neq 0\}$. We can always construct another solution 
\begin{equation*}
    \pi' = (\pi_1^A, \ldots,\pi_{p-1}^A, \pi_p^A - \epsilon, \ldots,\pi_{r-1}^A, \epsilon,\ldots,\pi_b^A),
\end{equation*}
where $\epsilon = \min\{\pi_p^A , \frac{b}{(r-1)\cdot(r-p)}\} > 0$.

Similarly we first investigate the feasibility of $\pi'$.


\begin{enumerate}[leftmargin=8mm,
 labelsep=3pt, itemindent=5pt,
                  labelwidth= 3em]

\item[\textbf{(1)}] Consider $x < p$. It is clear that
\begin{equation*}
    \ratio(\pi', x) = \ratio(\pi^A, x) \leq \rob_y.
\end{equation*}

\item[\textbf{(2)}] Consider $p \leq x < r$. Note that 
\begin{equation*}
    x \ratio(\pi', x) = x \ratio (\pi^A, x) - \eps (b+p-1) + \eps x <  x \ratio (\pi^A, x).
\end{equation*}
Therefore, 
\begin{equation*}
    \ratio(\pi', x) < \ratio(\pi^A, x) \leq \rob_y.
\end{equation*}

\item[\textbf{(3)}] Consider $x = r$. Since $\sum_{i=r+1}^b \pi_i^A = 0$, we have $\pi_i^A = 0, \forall \ i \in \{r+1, \ldots,b\}$. Note that 
\begin{equation*}
    r\ratio(\pi', r) = r \ratio(\pi^A, r) - \eps (b+p-1) + \eps(b+r-1).
\end{equation*}
By arranging terms, 
\begin{equation}
    \label{2b_1}
    r\ratio(\pi', r) = r \ratio(\pi^A, r) + \eps(r-p).
\end{equation}
Given that $\eps \leq \frac{b}{(r-1)(r-p)}$, 
\begin{equation}
    \label{2b_2}
    r \ratio(\pi^A, r) + \eps(r-p) \leq r \ratio(\pi^A, r) + \frac{b}{r-1}.
\end{equation}
Since $\sum_{i=r+1}^b \pi_i^A = 0$, we have  $\sum_{i=1}^r \pi_i^A = 1$. Therefore, 
\begin{equation}
    \label{2b_3}
     r \ratio(\pi^A, r) + \frac{b}{r-1} \leq r \ratio(\pi^A, r) + \frac{\sum_{i=1}^r 
     \pi^A_i (b+i-1)}{r-1}.
\end{equation}
By definition, 
\begin{equation*}
    \ratio(\pi^A, r) = \frac{\sum_{i=1}^r \pi_i^A (b+i-1) + \sum_{i=r+1}^b \pi_i^A r}{r}.
\end{equation*}
Note that $\sum_{i=r+1}^b \pi_i^A  = 0$. We further conclude
\begin{equation}
    \label{2b_4}
     r \ratio(\pi^A, r) + \frac{\sum_{i=1}^r  \pi^A_i (b+i-1)}{r-1} = \frac{r \sum_{i=1}^r \pi_i^A (b+i-1) }{r-1}.
\end{equation}
By \eqref{2b_1}, \eqref{2b_2}, \eqref{2b_3} and \eqref{2b_4},
\begin{equation*}
    r\ratio(\pi', r) \leq  \frac{r \sum_{i=1}^r \pi_i^A (b+i-1) }{r-1}.
\end{equation*}
Note that $\pi_r^A = 0$. We have 
\begin{equation*}
    \ratio(\pi', r) \leq  \frac{ \sum_{i=1}^{r-1} \pi_i^A (b+i-1) }{r-1} = \ratio(\pi^A, r-1) \leq \rob_y.
\end{equation*}

\item[\textbf{(4)}] Consider $r < x \leq b$. Consider $r < x \leq b$. Since $\pi_i' = \pi_i^A = 0$ for all $i \geq r+1$, it is clear that 
\begin{equation*}
    \ratio(\pi', x) < \ratio(\pi', r) \leq \rob_y.
\end{equation*}
\end{enumerate}

Consequently, $\ratio(\pi', x) \leq \rob_y$ for all $x \in [b]$. This implies that $\pi'$ is also a feasible solution to \ref{Prob A}.

However, 
\begin{equation*}
    y \ratio(\pi', y) = y \ratio(\pi^A, y) - \eps (b+p-1) + \eps y < y \ratio(\pi^A, y).
\end{equation*}
This implies
$
    \ratio(\pi', y) < \ratio(\pi^A, y),
$
which contradicts with the optimality of $\pi^A$.

In conclusion, $\pi_i^A \neq 0, \forall \ i \in \{y+1, y+2, \ldots,b\}$.
\end{proof}

\begin{lemma}
\label{lemma:_binding}
    For any $\rob_y \in [\rob_\xi, \rob_\nu]$ in \ref{Prob A}, at the optimal solution $\pi^A$, the $(y+1)$-th through $b$-th constraints are \textbf{binding}, \textit{i.e.}, they are satisfied as equalities.
\end{lemma}

\begin{proof}
    The dual problem to \ref{Prob A} (referred to as \ref{Prob B}) can be formulated as 
\begin{align*}
    \tag{Problem B}
    \label{Prob B}
    \max_{\lam_1, \lam_2, \ldots, \lam_{b+1}} \quad & 
    \rob_y \cdot \left(\sum_{i=1}^{b} \lam_i\right) +  \lam_{b+1}\\
    \text{s.t.} \quad \quad \quad 
    & \sum_{i=1}^{x-1} \lambda_i + \sum_{i=x}^{b}\frac{b+x-1}{i} \cdot \lambda_i + \lambda_{b+1}\leq \frac{b+x-1}{y}  , \quad \forall \ x \in \{1, 2, \ldots, y\}, \\
    & \sum_{i=1}^{x-1} \lambda_i + \sum_{i=x}^{b}\frac{b+x-1}{i} \cdot \lambda_i + \lambda_{b+1} \leq 1 , \quad \forall \ x \in \{y+1, \ldots, b\}, \\
    &\lambda_i \leq 0 , \quad \forall \ i \in \{1, 2, \ldots, b\}, \\
    & \lambda_{b+1} \text{ is free}.
\end{align*}
Let  $\lam^B = (\lam_1^B, \lam_2^B, \ldots, \lam_{b+1}^B)$ be the optimal solution to \ref{Prob B}.

By the \textit{complementary slackness}, 
\begin{equation*}
    \left[\sum_{i=1}^{x-1} \lambda_i^B + \sum_{i=x}^{b}\frac{b+x-1}{i} \cdot \lambda_i^B + \lambda_{b+1}^B - 1\right] \cdot \pi_x^A = 0,\quad  \forall\ x \in \{y+1, \ldots, b\}.
\end{equation*}

By \Cref{lemma:pi_i_neq_0}, $\pi_x^A \neq 0$ for all $x \in \{y+1, \ldots, b\}$, we have
\begin{equation}
    \sum_{i=1}^{x-1} \lambda_i^B + \sum_{i=x}^{b}\frac{b+x-1}{i} \cdot \lambda_i^B + \lambda_{b+1}^B = 1, \quad \forall \ x \in \{y+1, .\ldots, b\},
\end{equation}
where we later refer to them as $\mathrm{Eq.}(y+1)$ to $\mathrm{Eq.}(b)$.

We consider the following three cases.

\vspace{5pt}

\noindent\textbf{Case I:} \colorbox{gray!20}{$\rob_y = \rob_\nu$}.
Note that $\eq [y+1, b]$ is a feasible solution to \ref{Prob A} when $\rob_y = \rob_\nu$. Moreover, by \Cref{thm:eq_robust}, it's the only distribution that achieves $\ratio(\pi, y) = 1$. Therefore, $\eq [y+1, b]$ is the optimal solution to \ref{Prob A} when $\rob_y = \rob_\nu$. By~\Cref{def_Equalizing_dis}, $(y+1)$-th to $b$-th constraints in \ref{Prob A} are binding at $\pi^A=\eq [y+1, b]$.

\vspace{5pt}

\noindent\textbf{Case II:} \colorbox{gray!20}{$\rob_y \in (\rob_\xi, \rob_\nu)$}.
Let $\pri$ denote the optimal objective value in \ref{Prob A} (the \textit{primal problem}) and let $\du$ denote the optimal value in \ref{Prob B} (the \textit{dual problem}). It is clear that $\pri > 1$ when $\rob_y \in (\rob_\xi, \rob_\nu)$.

Since \ref{Prob A} is a convex optimization problem, and $\eq [1, b]$ (i.e. Karlin's distribution \cite{Karlin1988}) is always an interior point, by the \textit{Slater's condition}, $\du = \pri >1$.

We then show that in \ref{Prob B}, $\lam^B_i \neq 0$ for all $i \in \{y+1, \ldots,b\}$. We prove this by contradiction, assuming that there exists $ r \in \{y+1, \ldots, b\}$ such that $\lam^B_r = 0$.

\SetLabelAlign{newline}{%
  \begin{minipage}[t]{\dimexpr\linewidth-\labelsep\relax}%
    #1\vspace*{6pt}
  \end{minipage}%
}

\vspace{5pt}

\textbf{Case II(a):} \colorbox{gray!20}{$r = y+1$}.
Consider $\mathrm{Eq.}(y+1)$ and $\mathrm{Eq.}(y+2)$:
\[
\begin{cases}
    \lambda_1 + \cdots + \lambda_y + \frac{b+y}{y+1}\cdot \lambda_{y+1} + \ \ \ \ \frac{b+y}{y+2} \cdot \lambda_{y+2} + \cdots + \quad \frac{b+y}{b}\cdot \lambda_b + \lambda_{b+1} = 1 & \mathrm{Eq.}(y+1)\\
    \lambda_1 +  \cdots + \lambda_y + \quad \quad \ \ \lambda_{y+1} + \frac{b+y+1}{y+2} \cdot \lambda_{y+2} + \cdots + \frac{b+y+1}{b}\cdot \lambda_b + \lambda_{b+1} = 1 & \mathrm{Eq.}(y+2)\\
\end{cases}
\]
Note that $\frac{b+y}{y+2} < \frac{b+y+1}{y+2}, \ldots,\frac{b+y}{b}< \frac{b+y+1}{b}$. If $\lam^B_{y+1} = 0$, then $(\sum_{i=1}^y \lam^B_i) + \lam^B_{b+1} = 1$.

\vspace{5pt}

\textbf{Case II(b):} \colorbox{gray!20}{$y+1 < r < b$}.
Consider $\mathrm{Eq.}(r)$ and $\mathrm{Eq.}(r+1)$:
\[
\begin{cases}
    \lambda_1 + \cdots + \lambda_{r-1} + \frac{b+r-1}{r}\cdot \lambda_{r} + \frac{b+r-1}{r+1} \cdot \lambda_{r+1} + \cdots +  \frac{b+r-1}{b}\cdot \lambda_b + \lambda_{b+1} = 1 & \mathrm{Eq.}(r)\\
    \lambda_1 + \cdots + \lambda_{r-1} + \quad \quad \quad \ \lambda_{r} + \ \ \ \ \frac{b+r}{r+1} \cdot \lambda_{r+1} + \cdots + \quad  \frac{b+r}{b}\cdot \lambda_b + \lambda_{b+1} = 1 & \mathrm{Eq.}(r+1)\\
\end{cases}
\]

Note that $\frac{b+r-1}{r+1} < \frac{b+r}{r+1}, \ldots, \frac{b+r-1}{b} < \frac{b+r}{b}$. If $\lam^B_r = 0$, then $\left(\sum_{i=1}^{r-1} \lam^B_i\right) + \lam^B_{b+1} = 1$. 

Based on this, consider $\mathrm{Eq.}(r-1)$ and $\mathrm{Eq.}(r)$:
\[
\begin{cases}
    \lambda_1 + \cdots + \lambda_{r-2} + \frac{b+r-2}{r-1}\cdot \lambda_{r-1} + \frac{b+r-2}{r} \cdot \lambda_{r} + \cdots +  \frac{b+r-2}{b}\cdot \lambda_b + \lambda_{b+1} = 1 & \mathrm{Eq.}(r-1)\\
    \lambda_1 + \cdots + \lambda_{r-2} + \quad \quad \quad  \lambda_{r-1} + \frac{b+r-1}{r} \cdot \lambda_{r} + \cdots +  \frac{b+r-1}{b}\cdot \lambda_b + \lambda_{b+1} = 1 & \mathrm{Eq.}(r)\\
\end{cases}
\]
We can deduce $\lam^B_{r-1} = 0$. Thus, $\left(\sum_{i=1}^{r-2} \lam^B_i\right) + \lam^B_{b+1} = 1$.

Repeating this process, we obtain $\left(\sum_{i=1}^y \lam^B_i\right) + \lam^B_{b+1} = 1$.

\textbf{Case II(c):} \colorbox{gray!20}{$r=b$}.
Consider $\mathrm{Eq.}(b-1)$ and $\mathrm{Eq.}(b)$:
\[
\begin{cases}
    \lambda_1 + \cdots + \lambda_{b-2} + \frac{2b-2}{b-1}\cdot \lambda_{b-1} + \frac{2b-2}{b} \cdot \lambda_{b} = 1 & \mathrm{Eq.}(b-1)\\
    \lambda_1 + \cdots + \lambda_{b-2} + \quad \quad \ \ \ \lambda_{b-1} + \frac{2b-1}{b} \cdot \lambda_{b} = 1 & \mathrm{Eq.}(b)\\
\end{cases}
\]
If $\lam^B_b = 0$, then $\lam^B_{b-1} = 0$. This implies that there exists $r' \in (y+1,b)$ such that $\lam^B_{r'} = 0$. Based on the analysis in (b), we have $(\sum_{i=1}^y \lam^B_i) + \lam^B_{b+1} = 1$.

Note that in Problem~\ref{Prob B}, we have $\lambda_i \le 0$ for all $i \in [b]$, $\lambda_{b+1}$ is a free variable, and $\rob_y > 1$. Therefore, the optimal solution to \ref{Prob B} is $\lam_i^B = 0, \forall \ i \in [b]$ and $\lam^B_{b+1}= 1$, achieving an objective value of $\du = 1$. This leads to a contradiction.
It thus follows that $\lam^B_i \neq 0 $ for all $i \in \{y+1, \ldots,b\}$.

Now, by the complementary slackness,
\begin{equation*}
    \lam^B_x \cdot \left[ \frac{\sum_{i=1}^x \pi_i^A (b+i-1) + \sum_{i=y+1}^b \pi_i^A x}{x} - \rob_y \right] = 0, \quad \forall \ x \in \{y+1, \ldots,b\}.
\end{equation*}

Since $\lam^B_i \neq 0 $ for all $i \in \{y+1, \ldots,b\}$, we have
\begin{equation*}
    \frac{\sum_{i=1}^x \pi_i^A (b+i-1) + \sum_{i=y+1}^b \pi_i^A x}{x} = \rob_y, \quad \forall \ x \in \{y+1, \ldots,b\}.
\end{equation*}

Therefore, the $(y+1)$-th to $b$-th constraints in \ref{Prob A} are binding at $\pi^A$.

\vspace{5pt}

\noindent\textbf{Case III:} \colorbox{gray!20}{$\rob_y = \rob_\xi$}.
Note that $\eq [1, b]$ (i.e. Karlin's distribution \cite{Karlin1988}) is the only probability distribution that is $\rob_\xi $-competitive. We can verify that the $(y+1)$-th to $b$-th constraints in \ref{Prob A} are binding at $\pi^A = \eq[1,b]$.
\end{proof}

\begin{lemma}
\label{lemma:prepare}
    Consider $y < b$. Let $\rob \in [\rob_\xi, \rob_\nu]$, where $\rob_\xi = \rob(\eq[1, b]) = \frac{e_b}{e_b -1}$ and $\rob_\nu = \rob(\eq[y+1, b])$. Let the consistency and robustness of $\opB(\eq[y+1, b], \rob)$ (see Algorithm~\ref{alg:op_B}) under prediction y be $\con_y$ and $\rob_y$, respectively. Then, any $\rob_y$-robust algorithm's consistency under prediction $y$ is at least $\con_y$.
\end{lemma}

\begin{proof} Let $\Del$ denote the set of probability distributions on $\mathbb{N}_+$, i.e. $\Del \coloneqq \{\pi \mid \sum_{i=1}^{\infty}\pi_i = 1\}$. Let $\Del_b$ be the set of probability distributions on $[b]$, i.e. $\Del_b \coloneqq \{\pi \mid \sum_{i=1}^b \pi_i  =1\}$. Let $\Del_b'$ be the set of probability distributions on $[b]$ that achieve equalizing ratios on $\{y+1,\ldots,b\}$, i.e. $$\Del_b' \coloneqq \left\{\pi \mid \sum_{i=1}^b \pi_i = 1; \ratio(\pi, x) = \CR(\pi), \forall \ x \in \{y+1, \ldots, b\}\right\}.$$ It is straightforward to see that $\Del_b' \subseteq \Del_b \subseteq \Del$.

Let $\Del_{\rob_y}$ denote the set of probability distributions on $\mathbb{N}_+$ that is $\rob_y$-robust under prediction $y$, i.e. $$\Del_{\rob_y} \coloneqq \left\{\pi \mid \sum_{i=1}^{\infty}\pi_i = 1;\rob_y(\pi) = \rob_y\right\}.$$ It is straightforward to see that $\Del_{\rob_y} \subseteq \Del$.

To prove that any $\rob_y$-robust algorithm's consistency under prediction $y$ is at least $\con_y$, it suffices to show 
\begin{equation*}
    \min_{\pi \in \Del_{\rob_y}} \con_y(\pi) = \con_y.
\end{equation*}
Consider $\pi \in \Del \setminus \Del_b$. Let $\pi'$ be the distribution after transferring all probability mass at $i > b$ to $b$. It is clear that $\con_y(\pi') = \con_y(\pi), \forall \ y< b$. By~\Cref{lemma:restrict}, $\rob_y (\pi') \leq \rob_y (\pi)$. Therefore, 
\begin{equation}
    \label{trans_1}
    \min_{\pi \in \Del_{\rob_y}} \con_y(\pi) = \min_{\pi \in\Del_{\rob_y} \cap \Del_b} \con_y (\pi).
\end{equation}
Note that $\min_{\pi \in\Del_{\rob_y} \cap \Del_b} \con_y (\pi)$ is the optimal objective value of \ref{Prob A}.

By \Cref{lemma:_binding}, \ref{Prob A} reduces to the following problem (referred to as \ref{Prob C}):
\begin{align*}
    \label{Prob C}
    \tag{Problem C}
    \min_{\pi_1, \pi_2, \ldots, \pi_b} \quad & \frac{\sum_{i=1}^{y}\pi_i \cdot (b+i-1) + \sum_{i=y+1}^{b}\pi_i \cdot y}{y}\\
    \text{s.t.} \quad \quad 
    & \frac{\sum_{i=1}^{x} \pi_i \cdot (b + i - 1) + \sum_{i=x+1}^{b} \pi_i  \cdot x}{x} \leq \rob_y, \quad \forall \ x \in \{1, 2, \ldots, y\}, \\
    & \frac{\sum_{i=1}^{x} \pi_i \cdot (b + i - 1) + \sum_{i=x+1}^{b} \pi_i  \cdot x}{x} = \rob_y, \quad \forall \ x \in \{y+1, \ldots, b\}, \\
    & \sum_{i=1}^{b} \pi_i = 1, \\
    & \pi_i \geq 0, \quad \forall \ i \in \{1, 2, \ldots, b\}.
\end{align*}
In other words, we have
\begin{equation}
\label{trans_2}
    \min_{\pi \in\Del_{\rob_y} \cap \Del_b} \con_y (\pi) = \min_{\pi \in\Del_{\rob_y} \cap \Del_b'} \con_y (\pi).
\end{equation}

Let $\pi^*$ denote $\opB(\eq[y+1, b], \rob)$. Suppose $\pi^*$ has positive support over $\{1,\ldots, k, y+1,\ldots,b\}$. From the construction of \opB (see Algorithm~\ref{alg:op_B}), it follows that
\begin{subequations}
    \begin{align}
    \label{opB: ratio1}
        \ratio(\pi^*, x) = \rob, &\quad \forall\ x \in \{1,\ldots, k-1, y+1, \ldots,b\}\\
        \label{opB: ratio2}
        \ratio(\pi^*, x) \leq \rob, &\quad \forall\ x \in \{k, \ldots, y\}
    \end{align}
\end{subequations}
\Cref{opB: ratio1} and \eqref{opB: ratio2} together imply $\rob = \rob_y(\pi^*) = \rob_y$. Since \ref{Prob C} requires $\ratio(\pi, x) = \rob_y$ for all $x \in \{y+1, \ldots,b\}$, by~\Cref{lemma:eq_ratio}, a necessary condition for the optimal solution is 
\begin{equation}
    \label{necessary_condition}
    \pi_{i+1} = \frac{b}{b-1}\cdot \pi_i , \quad \forall \ i \in \{y+2,\ldots,b-1\}.
\end{equation}

Consider $\pi \in (\Del_{\rob_y}\cap \Del_b') \setminus \{\pi^*\}$. Let $r$ denote the minimal index such that $\pi^*_r \neq \pi_r$, i.e. $r = \min \{i \mid \pi_i^* \neq \pi_i\}$. Obviously, $r < b$. We consider the following cases.

\vspace{5pt}

\noindent\textbf{Case I:} \colorbox{gray!20}{$y < r < b$}.
In this case, $\pi_i = \pi^*_i$, for all $i \in [y]$. Since $\pi_{y+1} \neq \pi_{y+1}^*$ leads to $\ratio(\pi, y+1) \neq \ratio(\pi^*, y+1) = \rob_y$, we have $r\neq y+1$. Thus, $\pi_i = \pi^*_i$, for all $i \in [y+1]$. Based on this, we have $\sum_{i=y+2}^{b} \pi_i = \sum_{i=y+2}^{b} \pi_i^*$. If there exists $r \in \{y+2, \ldots,b\}$, such that $\pi_r \neq \pi^*_r$, then $\pi$ violates  Condition~\ref{necessary_condition}. Therefore, $\pi$ is not the optimal solution to \ref{Prob C}.

\vspace{5pt}

\noindent\textbf{Case II:} \colorbox{gray!20}{$r \leq y$}.

\vspace{5pt}

\textbf{Case II(a):} \colorbox{gray!20}{$\pi_r > \pi^*_r$}.
Note that $\ratio(\pi^*, x) = \rob_y$ for all $x \in [k-1]$. This makes $r > k-1$. Therefore, we have 
\begin{align*}
    \pi_i &= \pi_i^*, \quad \forall i \in [k-1]\\
    \pi_k \geq \pi^*_k \ \text{ and }\pi_i \geq \pi^*_i &= 0, \quad \ \ \forall i \in \{k+1, \ldots, y\}
\end{align*}
with some $r \in \{k, \ldots,y\}$ such that $\pi_r > \pi^*_r$. Thus, $\sum_{i=y+1}^{b} \pi_i < \sum_{i=y+1}^b \pi^*_i$. Combing them together, 
\begin{equation*}
    \con_y (\pi) = \ratio(\pi, y) > \ratio(\pi^*, y) = \con_y(\pi^*).
\end{equation*}
Therefore, $\pi$ cannot be an optimal solution to~\ref{Prob C}.

\vspace{5pt}

\textbf{\textbf{Case II(b):}} \colorbox{gray!20}{$\pi_r < \pi^*_r$}.
Note that $\pi^*_i = 0$ for all $i \in \{k+1, \ldots, y\}$. Therefore, $r \leq k$. Since $\pi_i = \pi^*_i$ for all $i \in [r-1]$ and $\pi_r < \pi^*_r$, we have $\sum_{i= r+1}^y \pi_i \neq 0$; otherwise $\rob_y (\pi) > \rob_y (\pi^*) = \rob_y$. Let $r'$ denote the minimal index in $\{r+1, \ldots,y\}$ such that $\pi_{r'} \neq 0$, i.e. $r' = \min \left\{i \in \{r+1, \ldots, y\} \mid \pi_i \neq 0\right\}$. Consider $\pi' = (\pi_1, \ldots, \pi_r + \eps,\ldots, \pi_{r'} - \eps, \ldots)$, where $\eps = \min\{\pi^*_r - \pi_r, \pi_{r'}\}$. It is straightforward to verify that 
\begin{equation}
\label{contra_1}
    \rob_y(\pi') = \rob_y (\pi), \quad \con_y (\pi') < \con_y (\pi).
\end{equation}
Note that 
\begin{equation*}
    \con_y (\pi') \geq \min_{\pi \in \Del_{\rob_y}} \con_y(\pi) \geq \min_{\pi \in \Del_{\rob_y} \cap \Del_b'} \con_y (\pi).
\end{equation*}
If $\con_y (\pi) = \min_{\pi \in \Del_{\rob_y} \cap \Del_b'} \con_y (\pi)$, then we have 
\begin{equation}
\label{contra_2}
    \con_y (\pi') \geq \con_y (\pi).
\end{equation} 
However, \eqref{contra_2} Contradict with \eqref{contra_1}. Therefore, $\con_y (\pi) \neq \min_{\pi \in \Del_{\rob_y} \cap \Del_b'}$, which makes $\pi$ suboptimal to \ref{Prob C}.

In conclusion, considering \textbf{(a)} and \textbf{(b)}, $\forall \ \pi \in (\Del_{\rob_y}\cap \Del_b') \setminus \{\pi^*\}$, $\pi$ is not an optimal solution to \ref{Prob C}. In other words, $\pi^*$ is the optimal solution to \ref{Prob C} and thus
\begin{equation}
    \label{final}
    \con_y = \con_y(\pi^*) =  \min_{\pi \in\Del_{\rob_y} \cap \Del_b'} \con_y (\pi).
\end{equation}

By \Cref{trans_1}, \Cref{trans_2} and \Cref{final}, we have
\begin{equation*}
    \min_{\pi \in \Del_{\rob_y}} \con_y(\pi) = \con_y.
\end{equation*}
Therefore, any $\rob_y$-robust algorithm's consistency under prediction $y$ is at least $\con_y$.
\end{proof}

\begin{lemma}

\label{lemma:y<b}
    Consider $y < b$. Let $\rob \in [\rob_\xi, \rob_\nu]$, where $\rob_\xi = \rob(\eq[1, b]) = \frac{e_b}{e_b -1}$ and $\rob_\nu = \rob(\eq[y+1, b])$. Let $\con_y(\rob)$ and $\rob_y(\rob)$ denote the consistency and robustness of $\opB(\eq[y+1, b], \rob)$ (see Algorithm~\ref{alg:op_B}) with respect to prediction $y$, respectively. $\con_y(\rob)$ and $\rob_y(\rob)$ are jointly Pareto optimal.
\end{lemma}

\begin{proof} Consider $\opB$ $(\eq[y+1, b], \rob)$ and
let $\con_y(\rob)$ and $\rob_y (\rob)$ denote its prediction-specific consistency and robustness. Note that $\con_y(\rob)$ is a strictly decreasing function of $\rob$, and $\rob_y (\rob) = \rob$. By~\Cref{lemma:prepare}, for any $\rob \in [\rob_\xi,\rob_\nu]$, any $\rob_y(\rob)$-robust algorithm is at least $\con_y(\rob)$-consistent under prediction $y$. Then, the only way to disprove the Pareto optimality of $(\con_y(\rob), \rob_y(\rob))$ is to find another $\con_y(\rob)$-consistent and $\rob_y^{\mathcal{Q}}$-robust algorithm $\mathcal{Q}$, where $\rob_y^{\mathcal{Q}} < \rob_y(\rob)$. 

Assume such algorithm $\mathcal{Q}$ exists. Since $\rob_y^{\mathcal{Q}} < \rob_y(\rob)$ and $\con_y(\rob)$ is a strictly decreasing function of $\rob$, we have 
\begin{equation}
\label{cont_1}
    \con_y(\rob_y^{\mathcal{Q}}) > \con_y(\rob_y(\rob)) = \con_y(\rob).
\end{equation}

Next, we invoke~\Cref{lemma:y<b} once again, any $\rob_y^{\mathcal{Q}}$-robust algorithm is at least $\con_y (\rob_y^{\mathcal{Q}})$-consistent. Note that $\rob_y^{\mathcal{Q}}$-robust algorithm $\mathcal{Q}$ is $\con_y(\rob)$-consistent. Thus 
\begin{equation}
\label{cont_2}
    \con_y(\rob_y^{\mathcal{Q}}) \leq\con_y(\rob).
\end{equation}

\Cref{cont_1} and \Cref{cont_2} lead to a Contradiction.
Therefore, for any $\rob \in [\rob_\xi,\rob_\nu]$, $\con_y(\rob)$ and $\rob_y(\rob)$ are jointly Pareto optimal.   
\end{proof}

Consider $\opA(\eq[1,n], y)$ and denote its consistency and robustness under prediction $y$ by $\con_y$ and $\rob_y$, respectively.
For the case when $y \geq b$, the following result holds.

\begin{lemma}
\label{lemma:y>=b}
    Suppose $y \geq b$ and $1 < n \leq b$. The consistency and robustness of $\opA(\eq[1,n], y)$ (see Algorithm~\ref{alg:op_A}) under prediction $y$, denoted by $\con_y$ and $\rob_y$ are jointly Pareto optimal.
\end{lemma}

\begin{proof}
Let $\pi^*$ denote $\opA(\eq[1, n], y)$. Our proof is divided into two cases.   

\vspace{5pt}

\noindent\textbf{Case I:} \colorbox{gray!20}{$\pi^*_2 = 0$}.
Note that $n > 1$, from the construction of Algorithm~\ref{alg:op_A}, $\pi^*_2 = 0$ could only happen when $y = b$. In this case, we have $\pi_1^* = 1/b$, $\pi_{b+1}^* = (b-1) / b$, and $\pi^*_i = 0$ for all $i \in \mathbb{N}_+ \setminus \{1, b+1\}$. Note that its prediction-specific consistency and robustness are $1$ and $2 - (1/b)$, respectively, i.e. $\con_y = \con_y (\pi^*) = 1$ and $\rob_y = \rob_y (\pi^*) = 2 - (1/b)$.

Note that all $1$-consistent algorithm under $y = b$ can only have probability mass on $\{1, b+1\}$. Let $\pi_1$ and $1 - \pi_1$ represent the probability mass on $1$ and $b+1$, respectively. then 
\begin{equation*}
    \rob_y (\pi) = \max_{\pi_1 \in [0,1]} \{ b \pi_1 + (1-\pi_1), \pi_1 + 2 (1-\pi_1) \} \geq 2 - (1/b).
\end{equation*}

Therefore, $\con_y$ and $\rob_y$ are jointly Pareto optimal.

\vspace{5pt}

\noindent\textbf{Case II:} \colorbox{gray!20}{$\pi^*_2 \neq 0$}.
Let $k \coloneqq \max\{i \leq b \mid \pi^*_i \neq 0\}$. Since $\pi^*_2 \neq 0 $, we have $k\geq 2$. According to the structure of Algorithm~\ref{alg:op_A}, we have that at least one of the following holds:
\begin{align}
    \label{Condition_1}
    \ratio(\pi^*, y+1) &= \rob_y,\\
    \label{Condition_2}
    y-b + 1 &= k.
\end{align}

Consider $\pi \neq \pi^*$. Let $r$ be the minimal index such that $\pi_i \neq \pi^*_i$, i.e. $ r \coloneqq \min \{i \mid \pi_i \neq \pi^*_i\}$.

\vspace{5pt}

\textbf{Case II(a):} \colorbox{gray!20}{$\pi_r > \pi_r^*$}.
In this case, we must have $r < y+1$. 
If $r \leq k-1$, based on the design of $\eq[1,n]$ and \opA (see Algorithm~\ref{alg:op_A}), we have $\ratio(\pi^*, x) = \rob_y$ for all $x \in [k-1]$, therefore, $\pi_r > \pi_r^*$ will lead to $\rob_y (\pi) > \rob_y (\pi^*) = \rob_y$. 

If $k \leq r < y+1$, then we must have $\pi^*_{y+1} \neq 0$, and thus $k-1 + b \geq y$. It's clear that $\con_y (\pi^*) \leq \con_y(\pi)$. Note that $\rob_y(\pi) \geq \ratio(\pi, 1) = \ratio(\pi^*, 1) = \rob_y(\pi^*) = \rob_y$. Therefore, $\pi$ cannot achieve strictly better consistency or robustness with respect to $y$.

\vspace{5pt}

\textbf{Case II(b):} \colorbox{gray!20}{$\pi_r < \pi_r^*$}.
A necessary condition for $\pi$ to dominate $\pi^*$ is that $\pi$ must be $\rob_y$-robust. In the subsequent analysis, we proceed under this assumption. Note that it suffices to restrict the support of the distribution to $[y+1]$, since assigning probability mass beyond y+1 yields no improvement in consistency while potentially worsening robustness. Then, we have $r < y+1$. Since $\pi^*_i = 0$ for all $k < i < y+1$, we further have $r \leq k$.

Let $v \coloneqq \min \{i > r \mid \pi_i \neq 0\}$. Consider 
\begin{equation*}
    \pi' = \{\pi_1, \ldots,\pi_{r-1}, \pi_r + \eps, \ldots,\pi_v - \eps,\ldots\},
\end{equation*}
where $\eps = \{\pi_v, \pi^*_r - \pi_r\}$. It is straightforward to verify $\rob_y (\pi') \leq \rob_y$.

If $b+k-1 < y$, by construction of \opA (see Algorithm~\ref{alg:op_A}), this can only happen when $b+n-1 < y$ and the operation terminates upon encountering the first occurrence of \texttt{Step 0}. Since $r < v \leq k < y$, it's clear that $\con_y (\pi') < \con_y(\pi)$.

If $b+k-1 = y$, since $r \leq k$, we have $r-1 + b \leq y$, thus
\begin{equation*}
    b \con_y (\pi') = b \con_y (\pi) + \eps \cdot [\mathds{1}\{v \leq y\}(r-v) + \mathds{1}\{v > y\}(r-1+b-y)] \leq b\con_y (\pi).
\end{equation*}
Equivalently,
\begin{equation*}
    \con_y (\pi') \leq  \con_y (\pi).
\end{equation*}
Attaining equality requires $r-1+b = y$ and $r = k$. 

If $b+k-1 > y$, then Condition~\eqref{Condition_2} does not hold; thus, Condition~\eqref{Condition_1} must hold, i.e., we have $\ratio(\pi^*, y+1) = \rob_y$. If $v = y+1$, then $\pi_i = \pi^*_i$ for all $i \in [r-1]$, $\pi_r < \pi_r^*$, $\pi_i = 0$ for all $r<i<y+1$. It's clear that $\rob(\pi) > \rob_y$. This makes $v < y+1$. Therefore,
\begin{equation*}
    b\con_y (\pi') = b \con_y (\pi) + \eps (r-v) < b \con_y (\pi).
\end{equation*}
Equivalently,
\begin{equation*}
    \con_y (\pi') < \con_y (\pi).
\end{equation*}

Therefore, for all $\pi \neq \pi^*$, $\pi$ is not the most consistent $\rob_y$-robust distribution except that $r -1+b = y$ and $r = k$ hold simultaneously.

When $r = k \geq 2$ and $r-1+b=y$ hold simultaneously, we have $\pi_1 = \pi_1^*$. Thus, $\rob_y (\pi) \geq \ratio (\pi, 1) = \ratio(\pi^*, 1) = \rob_y(\pi^*) = \rob_y$. Since $\pi$ is $\rob_y$-robust, $\rob_y (\pi) \leq \rob_y$. Thus, $\rob_y (\pi) = \rob_y$. Note that
\begin{equation*}
    \con_y (\pi) = \ratio(\pi, y) = \frac{\sum_{i=1}^y (i-1+b)\pi_i + \pi_{y+1} y}{b},
\end{equation*}
and $i-1+b \geq r-1+b = y$ for all $k \leq i \leq y$. This makes $\con_y (\pi) \geq \con_y (\pi^*)$. 

Therefore, any $\pi \neq \pi^*$ makes no Pareto improvement over $\pi^*$. 

In conclusion, $\con_y$ and $\rob_y$ are jointly Pareto optimal.
\end{proof}

With those foundational lemmas, we prove Theorem~\ref{thm:rsr}.
\begin{proof}[Proof of Theorem~\ref{thm:rsr}]

Assume that the user-specified parameter is $\overline{\rob}$. Let $n = \lceil \log_{b/(b-1) } 1 + 1 / (\overline{\rob} + 1)\rceil$. Let $\rob' = \rob(\eq[1, n]) = 1 + [(\frac{b}{b-1})^n-1]^{-1}$. Let $\pi$ denote the final distribution obtained after applying either \opA or \opB.

\vspace{5pt}

\noindent\textbf{Case I:} \colorbox{gray!20}{$y \geq b$}.
Let $\pi^\mathcal{A}$ denote $\opA(\eq[1, n], y)$. 

\vspace{5pt}

\textbf{Case I(a):} \colorbox{gray!20}{$\pi^\mathcal{A}$ is a two-point distribution}.
We have $\con_y(\pi^\mA) = 1$ and $\rob_y(\pi^\mA) = \frac{2b-1}{b}$. By design, $\opA$ never increases the robustness of the distribution at any point in time. Therefore, the a condition for forming a two-point should be $\rob' \geq \rob_y(\pi^\mA)$. In this case, it's easy to verify that 

\begin{equation}
\label{weak_1a}
    \con_y(\pi^\mathcal{A}) \leq \rob' \log\left(1 + \frac{1}{\rob'  -1}\right).
\end{equation}

\vspace{5pt}

\textbf{Case I(b):} \colorbox{gray!20}{$\pi^\mathcal{A}$ is not a two-point distribution}.
In this case, we have
\begin{equation*}
    \pi_1^\mathcal{A} = \eq_1[1,n] = [(b-1)((\frac{b}{b-1})^n-1)]^{-1}.
\end{equation*}
Thus, we have
\begin{equation*}
     \rob_y(\pi^\mathcal{A}) = \ratio(\pi^\mathcal{A}, 1) = \pi^\mathcal{A}_1 \cdot b + (1-\pi_1^\mathcal{A}) = 1 + [(\frac{b}{b-1})^n-1]^{-1} = \rob'.
\end{equation*}

Note that as $b \to \infty$, Kumar's algorithm is $\frac{\lam}{1 - e^{-\lam}}$-consistent and $\frac{1}{1-e^{-\lam}}$-robust. (see~\cite{Kumar2018}) Take $\lam = -\log(1-\rob'^{-1})$. In this case, Kumar's algorithm is $\rob' \log(1 + \frac{1}{\rob' -1})$-consistent and $\rob'$-robust. Since $\overline{\rob} < b-2$ and $n = \lceil \log_{b/(b-1) } 1 + 1 / (\overline{\rob} + 1)\rceil$, we have $n>1$. By~\Cref{lemma:y>=b}, $(\con_y(\pi^\mathcal{A}), \rob_y (\pi^\mathcal{A}))$ is Pareto optimal. Therefore, we have
\begin{equation*}
    \con_y(\pi^\mathcal{A}) \leq \rob_y (\pi^\mathcal{A}) \log\left(1 + \frac{1}{\rob_y(\pi^{\mathcal{A}}) -1}\right).
\end{equation*}

Since $\rob' = \rob_y(\pi^\mathcal{A})$, we have for any $y \geq b$,
\begin{equation}
\label{weak_1b}
    \con_y(\pi^\mathcal{A}) \leq \rob' \log\left(1 + \frac{1}{\rob'  -1}\right).
\end{equation}

\vspace{5pt}

\noindent\textbf{Case II:} \colorbox{gray!20}{$y < b$}.
Let $\rob'' = \min \{\rob_\nu, \rob'\}$. Recall that $\rob_\nu = \rob (\eq[y+1, b])$. Let $\pi^\mathcal{B}$ denote $\opB(\eq[y+1, b], \rob'')$. Note that \begin{equation*}
\rob_\nu = \rob(\eq[y+1, b]) > \rob(\eq[1, b]) = \rob_\xi,
\end{equation*}
and
 \begin{equation*}
1 + [(\frac{b}{b-1})^n-1]^{-1} = \rob(\eq[1, n]) \geq \rob(\eq[1, b]) =  \rob_\xi.  
\end{equation*}
Therefore, $\rob'' = \min \{\rob_\nu, \rob'\} \geq \rob_\xi$, and $\rob'' = \min \{\rob_\nu, \rob'\} \leq \rob_\nu$. Thus, we have 
\begin{equation*}
    \rob_\xi \leq \rob'' \leq \rob_\nu.
\end{equation*}
Therefore, $\rob_y (\pi^\mathcal{B}) = \rob''$. We consider the following two cases.

\vspace{5pt}

\textbf{Case II(a):} \colorbox{gray!20}{$\rob'' = \rob'$}.
In this case, we consider $\lam = -\log(1-\rob_y^{-1})$, Kumar's algorithm is $\rob_y \log(1 + \frac{1}{\rob_y -1})$-consistent and $\rob_y$-robust. Since $\rob_\xi \leq \rob'' \leq \rob_\nu$, by~\Cref{lemma:prepare},
\begin{equation*}
    \con_y (\pi^\mathcal{B}) \leq \rob_y(\pi^\mathcal{B}) \log\left(1 + \frac{1}{\rob_y (\pi^\mathcal{B}) -1}\right).
\end{equation*}

Since $\rob' = \rob'' = \rob_y (\pi^\mathcal{B})$, we have for any $y < b$,
\begin{equation}
\label{weak_2a}
    \con_y (\pi^\mathcal{B}) \leq \rob' \log\left(1 + \frac{1}{\rob' -1}\right).
\end{equation}

\vspace{5pt}

\textbf{Case II(b):} \colorbox{gray!20}{$\rob'' = \rob_\nu < \rob'$}.
In this case, $\pi^\mathcal{B} = \eq[y+1, b]$, and $\con_y(\pi^\mathcal{B}) = 1$. We can verify that
\begin{equation}
\label{weak_2b}
    \con_y (\pi^\mathcal{B}) \leq \rob' \log\left(1 + \frac{1}{\rob' -1}\right).
\end{equation}

By \eqref{weak_1a}, \eqref{weak_1b}, \eqref{weak_2a} and \eqref{weak_2b}, we conclude that for any $y \in \mathbb{N}_+$, 
\begin{equation*}
    \con_y (\pi^\mathcal{B}) \leq \rob' \log\left(1 + \frac{1}{\rob' -1}\right).
\end{equation*}

Therefore, by \Cref{def:cl_con_rob} and \Cref{def:ps_con_rob},
\begin{equation*}
    \con(\pi) = \sup_{y \in \mathbb{N}_+} \con_y (\pi) = \max\left\{\sup_{y \geq b} \con_y (\pi^\mA), \sup_{y < b}(\con_y(\pi^\mathcal{B}))\right\} \leq \rob' \log \left(1 + \frac{1}{\rob' -1}\right).
\end{equation*}

Across all of the cases discussed above, it holds that $\rob_y (\pi) \leq \rob'$. Similarly, by \Cref{def:cl_con_rob} and \Cref{def:ps_con_rob}, $\rob(\pi) =  \sup_{y \in \mathbb{N}_+} \rob_y (\pi) \leq \rob'$. 

Since $F(x) \coloneqq x \log (1 + \frac{1}{x-1})$ is decreasing on $x \in (1,+\infty)$, we further have
\begin{equation*}
    \con(\pi)  \leq \rob(\pi) \log \left(1 + \frac{1}{\rob(\pi) -1}\right).
\end{equation*}

Wei's lower bound~\cite{Wei2020} and~\Cref{def_weak}, \PRSR is \weak-optimal.

By~\Cref{lemma:y<b} and \Cref{lemma:y>=b}, \PRSR's prediction-specific consistency and robustness are Pareto optimal.
According to~\Cref{def_strong}, $\PRSR$ is \strong-optimal.

\end{proof}

\newpage
\section{Proofs for \Cref{sec:one-max_search} \label{appendix:oms}}

In this section, we prove \Cref{thm:sun} and \ref{thm:oms}.

\subsection{Proof of \Cref{thm:sun} \label{appendix:oms_1}}
\begin{proof}[Proof of \Cref{thm:sun}]
Let $\con_y^\mS$ and $\rob_y^\mS$ denote the prediction-specific consistency and robustness of Sun's algorithm with respect to prediction $y$. Let $\con^\mS$ and $\rob^\mS$ denote the consistency and robustness of Sun's algorithm.
    
Consider $y \in [L, L \con^\mS)$. In this case, Sun's algorithm sets thresholds to $\Phi = L\con^\mS$. To analyze the prediction-specific consistency, we assume that the maximum price is exactly $y$. Since $y < L\con^\mS $, we have $\alg = L$ and $\opt = y$. By~\Cref{def:ps_con_rob}, $\con_y^\mS = y/L$. To obtain the robustness under prediction $y$, we consider incorrect predictions. Observe that in the worst case, $\alg = L\con^\mS$, $\opt = U$. Therefore, $\rob_y^S = U / ({L\con^{\mS}}) = \theta / \con^\mS$. Consider the canonical competitive algorithm that adopts a fixed threshold policy with $\Phi = \sqrt{LU}$, which has prediction-specific consistency $\con_y^\mathcal{C} = y / L$ and robustness $\rob_y^\mathcal{C} =\sqrt{\theta}$. Note that $\forall \ \lam \in [0,1)$, $\rob^\mS > \sqrt{\theta}$, which makes
\begin{equation*}
    \con_y^\mathcal{S} = \con_y^\mathcal{C}, \quad \rob_y^\mathcal{S} > \rob_y^\mathcal{C}.
\end{equation*}
By~\Cref{def_strong}, Sun's algorithm is not \strongly.
\end{proof}

\subsection{Proof of \Cref{thm:oms} \label{appendix:oms_2}}
\begin{proof}[Proof of \Cref{thm:oms}] We first prove the prediction-specific consistency and robustness of \OMS. We consider the following three cases.

\vspace{5pt}

\noindent\textbf{Case I:} \colorbox{gray!20}{$y \in [L, \lam L + (1-\lam) \sqrt{LU})$}. In this case, \OMS sets the threshold at $\Phi =  \sqrt{LU}$. Since $\lam L + (1-\lam) \sqrt{LU} \leq \sqrt{LU}$, we have $ y < \sqrt{LU}$, therefore $\con_y = y/L$. Since $\Phi = \sqrt{LU}$, $\rob_y = \max\{\Phi / L, U/\Phi\} = \sqrt{\theta}$.

\vspace{5pt}

\noindent\textbf{Case II:} \colorbox{gray!20}{$y \in [\lam L + (1-\lam) \sqrt{LU}, \sqrt{LU}]$}. In this case, \OMS sets the threshold at $\Phi = y$. Obviously, $\con_y = 1$. Since $y \leq \sqrt{LU}$, we have $\rob_y = \max \{ \Phi / L, U/ \Phi \} = U/y $.

\vspace{5pt}

\noindent\textbf{Case III:} \colorbox{gray!20}{$y \in (\sqrt{LU}, U]$}. In this case, \OMS sets the threshold at $\Phi = \mu \sqrt{LU} + (1-\mu)y$, where $$\mu = \frac{(1-\lam) \sqrt{\theta}}{(1-\lam)\sqrt{\theta} + \lam}.$$ Since $\sqrt{LU} < y$, $\Phi = \mu \sqrt{LU} + (1-\mu)y \leq y$. Under worst-case construction, we conclude
\begin{align*}
    \con_y &= \frac{y}{\Phi} = \frac{y}{\mu \sqrt{LU} + (1-\mu) y} = \frac{(1-\lam)\sqrt{\theta}y + \lam y}{(1-\lam)U+\lam y},\\
    \rob_y &= \max\{\Phi / L, U / \Phi\} = \frac{\Phi}{L} = \frac{(1-\lam)U + \lam y}{(1-\lam)\sqrt{LU} + \lam L}.
\end{align*}

We then prove the \str optimality of \OMS. As in worst-case instances any deterministic algorithm performs equivalently to a threshold algorithm, it is not restrictive to only consider \texttt{OTA}s.

By considering the worst-case predictions for both prediction-specific consistency and robustness, we can conclude that \OMS is $\lam + (1-\lam)\sqrt{\theta}$-consistent and $\frac{\theta}{\lam + (1-\lam) \sqrt{\theta}}$-robust. Building on the lower bound established by Sun et al.~\cite{Sun2021}, \OMS is \weakly.

Consider $\Phi' \neq \Phi$. Let $(\con_y', \rob_y')$ denote the consistency and robustness of \texttt{OTA} that uses threshold $\Phi'$ with respect to $y$. 

Assume $y \in [L, \lam L + (1-\lam) \sqrt{LU})$. Since $\Phi' \neq \Phi = \sqrt{LU}$, $\rob_y' = \max\{\Phi'/L, U/\Phi'\} > \sqrt{LU} = \rob_y$. Therefore, $(\con_y, \rob_y)$ is Pareto optimal. 

Assume $y \in [\lam L + (1-\lam) \sqrt{LU}, \sqrt{LU}]$. Since $\Phi' \neq \Phi =  y$, $\con_y > 1 = \con_y$. Thus, $(\con_y, \rob_y)$ is Pareto optimal. 

Assume $y \in (\sqrt{LU}, U]$. If $\Phi' < \Phi$, then $\con_y' = y / \Phi' > y /\Phi = \con_y$. If $\Phi' > \Phi$, since $\Phi' > \Phi \geq \sqrt{LU}$, we have $\rob_y = \Phi / \sqrt{LU}$ and $\rob_y' = \Phi' / \sqrt{LU}$, thus $\rob_y' > \rob_y$. Therefore, $(\con_y, \rob_y)$ is Pareto optimal.

By~\Cref{def_strong}, \OMS is \strongly.
\end{proof}

\newpage

\section{Proofs for \Cref{sec:brittle}\label{appendix:brittle}}
In this section, we prove \Cref{thm:tol_1}, \ref{thm:tol_2} and \ref{thm:tol_3}.

\subsection{Proof of \Cref{thm:tol_1}}
\begin{proof}[Proof of \Cref{thm:tol_1}]
    We consider the following five cases to investigate the prediction-specific $\ep$-consistency and robustness. 

\vspace{5pt}
    
\noindent\textbf{Case I:} \colorbox{gray!20}{$y \in [L, M-2\ep]$}. 
In this case, $\ep$-Tolerant \OMS decides to set the threshold to $\Phi = \sqrt{LU}$ and the price range relevant to $\ep$-consistency is restricted to $[\max\{L, y-\ep\}, y + \ep]$. To determine the $\ep$-consistency, we assume that the highest price $x \in [\max\{L, y-\ep\}, y+\ep]$. Note that $y + \ep \leq (M-2\ep) + \ep < \sqrt{LU} = \Phi$. We have $\con_y^\ep = (y+\ep)/L$. To obtain the robustness, we consider incorrect predictions. $\rob_y = \max \{\Phi/L, U/ \Phi\} = \sqrt{\theta}$.

\vspace{5pt}

\noindent\textbf{Case II:} \colorbox{gray!20}{$y \in (M-2\ep, M)$}. 
In this case, $\ep$-Tolerant \OMS decides to set the threshold to $\Phi = M-\ep$ and the price range relevant to $\ep$-consistency is restricted to $[ y-\ep, y + \ep]$. For $\ep$-consistency, we observe the best attack can potentially occur at $x = M-\ep-\delta$ or $x = y+\epsilon$, where $\delta$ is the infinitesimal quantity. The corresponding ratios are $r_1 (M) \coloneqq (M-\ep)/L$ and $r_2 (M) \coloneqq (M+\ep)/(M-\ep)$. Note that $r_1 (M)$ is an increasing function of $M$ and $r_2 (M)$ is a decreasing function of $M$. Since $$r_1(L+3\ep) = \frac{L+2\ep}{L} \geq \frac{L+4\ep}{L+2\ep} = r_2 (L+3\ep),$$ we can conclude $r_1 (M) \geq r_2 (M)$, for all $M \in [L+3\ep, \sqrt{LU} - \ep]$.
thus $\con_y^\ep = (M-\ep)/L$. To obtain the robustness, we consider incorrect predictions. Since $M-\ep < \sqrt{LU}$, $\rob_y = U/(M-\ep)$.

\vspace{5pt}

\noindent\textbf{Case III:} \colorbox{gray!20}{$y \in [M, \sqrt{LU}+\ep]$}. 
In this case, $\ep$-Tolerant \OMS decides to set the threshold to $\Phi = y-\ep$ and the price range relevant to $\ep$-consistency is restricted to $[y-\ep, y + \ep]$. To determine the $\ep$-consistency, we assume that $x \in [y-\ep, y+\ep]$. Observe that the worst case occurs at $x = y+\ep$, we have $\con_y^\ep = (y+\ep)/(y-\ep)$. To obtain the robustness, we consider incorrect predictions. Note that $$y - \ep \leq (\sqrt{LU} + \ep) - \ep = \sqrt{LU}.$$ We have $\rob_y = U/\Phi = U / (y-\ep)$.

\vspace{5pt}

\noindent\textbf{Case IV:} \colorbox{gray!20}{$y \in (\sqrt{LU}+\ep, U - \ep)$}. 
In this case, $\ep$-Tolerant \OMS decides to set the threshold to $\Phi = \mu\sqrt{LU} + (1-\mu)(y-\ep)$ and the price range relevant to $\ep$-consistency is restricted to $[y-\ep, y+\ep]$. Since $\mu = \frac{(U-2\ep) - LU/(M-\ep)}{(U-2\ep) - \sqrt{LU}}$, by noting that
\begin{align*}
    (U-2\ep) - \frac{LU}{(L+3\ep - \ep )} &> 0,\\
     (U-2\ep) - \frac{LU}{(\sqrt{LU}-\ep - \ep)} &< (U-2\ep) - \sqrt{LU},
\end{align*}
we have $\mu \in (0, 1)$. To determine the $\ep$-consistency, we assume that $x \in [y-\ep,  y + \ep]$. Note that $\Phi \leq y-\ep$. We have $$\con_y^\ep = \frac{y+\ep}{\Phi} = \frac{y+\ep}{\mu\sqrt{LU} + (1-\mu)(y-\ep)}.$$ To obtain the robustness, we consider incorrect predictions. Note that $\Phi \geq \sqrt{LU}$. We have $\rob_y = \Phi/L =[\mu\sqrt{LU} + (1-\mu)(y-\ep)]/L$.

\vspace{5pt}

\noindent\textbf{Case V:}  \colorbox{gray!20}{$y \in [ U - \ep, U]$}. 
In this case, $\ep$-Tolerant \OMS decides to set the threshold to $\Phi = LU/(M-\ep)$ and the price range relevant to $\ep$-consistency is restricted to $[y-\ep, U]$. Since $\ep \leq (\sqrt{LU} - L)/4 < (U-L)/2$, we have 
\begin{align*}
    \Phi = \frac{LU}{M-\ep} \leq \frac{LU}{L+2\ep} \leq U-2\ep \leq y - \ep, 
\end{align*}
thus, $\con_y^\ep = \frac{U}{\Phi} =  \frac{M-\ep}{L}$. Furthermore, $\rob_y = \frac{\Phi}{L} = \frac{U}{M-\ep}$.

\vspace{2pt}

Then, we focus on the $\ep$-consistency and robustness of $\ep$-Tolerant \OMS.

\vspace{2pt}

In \textit{Case III}, we have
\[\con_y^\ep = \frac{(y+\ep)}{(y-\ep)} < \frac{L+2\ep}{L} \leq \frac{M-\ep}{L}.\]

\vspace{2pt}

In \textit{Case IV}, since $\mu = \frac{(U-2\ep) - LU/(M-\ep)}{(U-2\ep) - \sqrt{LU}}$, we get
\begin{align*}
    \con_y^\ep \leq \frac{U}{\mu \sqrt{LU} + (1-\mu)(U-2\ep)} = \frac{M-\ep}{L}.
\end{align*}
Similarly, we obtain
\begin{align*}
\rob_y = \frac{\mu\sqrt{LU} + (1-\mu)(y-\ep)}{L} \leq \frac{\mu\sqrt{LU} + (1-\mu)(U-2\ep)}{L} = \frac{U}{M-\ep}.
\end{align*}

By considering the worst-case $y$ over $\mY=[L,U]$, we conclude that $\ep$-Tolerant \OMS 's $\ep$-consistency and robustness is $\con^\ep = (M-\ep)/L$ and $\rob = U/(M-\eps)$, respectively.
\end{proof}

\subsection{Proof of \Cref{thm:tol_2}}
\begin{proof}[Proof of \Cref{thm:tol_2}]
By the lower bound provided by Sun et al.~\cite{Sun2021}, any $\rob$-robust algorithm must be at least $(\theta/\rob)$-consistent. By \Cref{def:cl_con_rob} and the definition of $\ep$-consistency, we have $\con^\ep \geq \con$ for any $\ep > 0$. Therefore, any $\rob$-robust algorithm has at least $(\theta/\rob)$ $\ep$-consistency.

Similarly, given $\ep > 0$, any algorithm that achieves $\con^\ep$ $\ep$-consistency is $\con^\ep$-consistent. Based on the lower bound by Sun et al.~\cite{Sun2021}, it must be at least $(\theta/ \con^\ep)$-robust.
\end{proof}

\subsection{Proof of \Cref{thm:tol_3}}

\begin{proof}[Proof of \Cref{thm:tol_3}]

By \Cref{thm:tol_1} and \Cref{thm:tol_2}, we conclude that $\ep$-Tolerant \PST's $\ep$-consistency and robustness are jointly Pareto optimal. 

To prove the Pareto optimality of prediction-specific $\ep$-consistency and robustness, we consider $\Phi' \neq \Phi$, which achieves $\ep$-consistency ${\con_y^\ep}'$ and robustness ${\rob_y}'$ with respect to $y$.

\vspace{5pt}
    
\noindent\textbf{Case I:} \colorbox{gray!20}{$y \in [L, M-2\ep]$}. 
In this case, $\rob_y = \sqrt{\theta}$. Note that $\Phi = \sqrt{LU}$ is the only threshold that achieves robustness $\sqrt{\theta}$. We can conclude that $(\con_y^\ep, \rob_y)$ is Pareto optimal.

\vspace{5pt}

\noindent\textbf{Case II:} \colorbox{gray!20}{$y \in (M-2\ep, M)$}. 
In this case, $\ep$-Tolerant \OMS sets the threshold to $\Phi = M-\ep$, achieving $\con_y^\ep = (M-\ep)/L$ and $\rob_y = U/(M-\ep)$. If $\Phi' < \Phi$, then $\rob_y' > \rob_y$. If $\Phi' > \Phi$, then consider the attack $x = \min \{\Phi' - \delta, y+\ep\}  $ for threshold $\Phi'$, where $\delta$ is the infinitesimal quantity. This makes $${\con_y^\ep}' \geq \min \{\Phi', y+\ep\} / L > (M-\ep)/L = \con_y^\ep.$$ Therefore, $(\con_y^\ep, \rob_y)$ is Pareto optimal.

\vspace{5pt}

\noindent\textbf{Case III:} \colorbox{gray!20}{$y \in [M, \sqrt{LU}+\ep]$}. 
In this case, $\ep$-Tolerant \OMS sets the threshold to $\Phi = y-\ep$, achieving $\con_y^\ep = (y+\ep)/(y-\ep)$ and $\rob_y = U/(y-\ep)$. Note that $\Phi = y-\ep < \sqrt{LU}$. If $\Phi' < \Phi$, then $\rob_y' > \rob_y$. If $\Phi' > \Phi$, consider $x = \min \{\Phi' - \delta, y+\ep\} $ for threshold $\Phi'$, where $\delta$ is the infinitesimal quantity. This guarantees the following inequality: $${\con_y^\ep}' \geq \min\{\Phi', y+\ep\}/L > (y-\ep)/L.$$ Note that $r'(y) \coloneqq (y-\ep)/L$ is an increasing function of $y$, and $r(y) \coloneqq (y+\ep)/(y-\ep)$ is a decreasing function of $y$. Since $y \geq M \geq L + 3\ep$ and $$r(L+3\ep) = (L+4\ep)/(L+2\ep) < (L+2\ep)/ L = r'(L+3\ep),$$ we have $r(y) < r'(y), \forall y \in [M, \sqrt{LU} + \ep]$. This gives $\con_y^\ep < {\con_y^\ep}'$. Therefore, $(\con_y^\ep, \rob_y)$ is Pareto optimal.

\vspace{5pt}

\noindent\textbf{Case IV:} \colorbox{gray!20}{$y \in (\sqrt{LU}+\ep, U - \ep)$}. 
In this case, $\ep$-Tolerant \OMS sets the threshold to $\Phi =  \mu\sqrt{LU} + (1-\mu)(y-\ep)$, achieving $\con_y^\ep = (y+\ep)/\Phi$ and $\rob_y = \Phi/L$. If $\Phi' < \Phi$, it follows that $${\con_y^\ep}' \geq (y+\ep)/\Phi' >(y+\ep)/\Phi = \con_y^\ep.$$ If $\Phi' > \Phi$, since $\Phi > \sqrt{LU} + \ep > \sqrt{LU}$, $\rob_y' > \rob_y$. Therefore, $(\con_y^\ep, \rob_y)$ is Pareto optimal.

\vspace{5pt}

\noindent\textbf{Case V:}  \colorbox{gray!20}{$y \in [ U - \ep, U]$}. 
In this case, $\ep$-Tolerant \OMS sets the threshold to $\Phi =  LU/(M-\ep)$, achieving $\con_y^\ep = (M-\ep)/L$ and $\rob_y = U/(M-\ep)$. If $\Phi' > \Phi$, then $\rob_y' > \rob_y$. If $\Phi' < \Phi$, consider highest price $x = U$ with $|x-y|\leq \ep$, ${\con_y^\ep}' \geq U/\Phi' > U/\Phi = \con_y^\ep$. Therefore, $(\con_y^\ep, \rob_y)$ is Pareto optimal.
\end{proof}

\end{document}